\documentclass[11pt, usenames, dvipsnames]{article}
\usepackage{natbib}
\usepackage{amssymb,amsmath,amsthm,latexsym,amsfonts,amscd,dsfont,enumerate}
\usepackage{todonotes}
\usepackage{xcolor}
\usepackage{graphicx}
\usepackage{bm}
\usepackage[T1]{fontenc}
\usepackage{lmodern}
\usepackage{hyperref}
\usepackage{cleveref}
\usepackage{autonum}

\usepackage[top=1.2in, left=0.9in, bottom=1.2in, right=0.9in]{geometry}

\newcommand{\buysell}{}

\usepackage{tikz}
\usetikzlibrary{arrows,positioning}
\tikzset{
    >=stealth',
    punkt/.style={
           rectangle,
           rounded corners,
           draw=black, very thick,
           text width=6.5em,
           minimum height=2em,
           text centered},
    pil/.style={
           ->,
           thick,
           shorten <=2pt,
        shorten >=2pt,
   }
}

\usepackage[]{hyperref}
 \hypersetup{
	hypertexnames=false, 
 colorlinks=true, 
 breaklinks=true, 
 urlcolor= blue, 
 linkcolor= blue, 
 bookmarksopen=true, 
 pdfauthor={Bichuch, M., Capponi, A., Sturm, S.}
 }

\newcommand{\R}{\mathbb{R}}

\newcommand{\Px}{\mathbb{P}}
\newcommand{\Exx}{\mathbb{E}}
\newcommand{\E}{\Exx}
\newcommand{\Qxx}{\mathbb{Q}}
\newcommand{\Q}{\Qxx}
\newcommand{\Qx}{\Qxx}
\newcommand{\Fx}{\mathbb{F}}

\newcommand{\VaR}{{VaR}}

\def\ind{{\mathchoice{1\mskip-4mu\mathrm l}{1\mskip-4mu\mathrm l}
{1\mskip-4.5mu\mathrm l}{1\mskip-5mu\mathrm l}}}
\newcommand{\abs}[1]{ \left \vert #1 \right \vert}
\newcommand{\define}{:=}

\newcommand{\rcp}{r_m^+}
\newcommand{\rcm}{r_m^-}
\newcommand{\rfp}{r_f^+}
\newcommand{\rfm}{r_f^-}

\newcommand{\XVA}{\mbox{XVA}}
\newcommand{\rXVA}{\mbox{rXVA}}

\newcommand{\circu}{\overset{\circ}{u}}
\newcommand{\circuo}{\overset{\bullet}{u}}

\DeclareMathOperator*{\essup}{ess\, sup}

\usepackage{accents}
\newlength{\dhatheight}

\newtheorem{theorem}{Theorem}[section]
\newtheorem{definition}[theorem]{Definition}
\newtheorem{corollary}[theorem]{Corollary}
\newtheorem{proposition}[theorem]{Proposition}
\newtheorem{remark}[theorem]{Remark}

\newtheorem{example}[theorem]{Example}
\newtheorem{assumption}[theorem]{Assumption}
\usepackage[normalem]{ulem}

\long\def\symbolfootnote[#1]#2{\begingroup\def\thefootnote{\fnsymbol{footnote}}\footnote[#1]{#2}\endgroup}

\usepackage{cleveref}

\title{Robust XVA\footnote{Data sharing is not applicable to this article as no new data were created or analyzed in this study.}}

\author{
Maxim Bichuch \thanks{Email: mbichuch@jhu.edu, Department of Applied Mathematics and Statistics, Johns Hopkins University. {Work partially supported by NSF (DMS-1736414), and by the Acheson J. Duncan Fund for the Advancement of Research in Statistics.} }
\and
Agostino Capponi \thanks{Corresponding Author. Email: ac3827@columbia.edu, Industrial Engineering and Operations Research Department, Columbia University.}
\and
Stephan Sturm \thanks{Email: ssturm@wpi.edu, Department of Mathematical Sciences, Worcester Polytechnic Institute.}
}

\begin{document}

\maketitle

\begin{abstract}
We introduce an arbitrage-free framework for robust valuation adjustments. An investor trades a credit default swap portfolio with a risky counterparty, and hedges credit risk by taking a position in defaultable bonds. The investor does not know the exact return rate of her counterparty's bond, but she knows it lies within an uncertainty interval. We derive both upper and lower bounds for the XVA
process of the portfolio, and show that these bounds may be recovered as solutions of nonlinear ordinary differential equations. The presence of
collateralization and closeout payoffs leads to important differences with respect to classical credit risk valuation. The value of the super-replicating
portfolio cannot be directly obtained by plugging one of the extremes of the uncertainty interval in the valuation equation, but rather depends on the
relation between the XVA replicating portfolio and the close-out value throughout the life of the transaction.
Our comparative statics analysis indicates that credit contagion has a nonlinear effect on the replication strategies and on the XVA.

\end{abstract}

\vspace{5mm}

\begin{flushleft}
\textbf{Keywords:} robust XVA, counterparty credit risk, backward stochastic differential equation, arbitrage-free valuation. \\
\textbf{Mathematics Subject Classification (2010): } {91G40, 91G20, 60H10}\\
\textbf{JEL classification: }{G13, C32}
\end{flushleft}

\section{Introduction}
Dealers need to account for market inefficiencies related to funding and credit valuation adjustments when marking their swap books. Those include the capital needed to support the trading position,
the losses originating in case of a premature default by either of the trading parties, and the remuneration of funding and collateral accounts. It is common market practice to refer to these costs as the XVA of the trade. Starting from 2011, major dealer banks have started to mark these valuation adjustments on their balance sheets; see, for instance, \cite{Cameron} and \cite{Beker}.

A large body of literature has studied the implication of such costs on the valuation and hedging of derivatives positions. One stream of the literature has analyzed XVA from a replication perspective. These works include \cite{Crepeya} and \cite{Crepeyb}, who use backward stochastic differential equations to replicate the transaction cash flows, accounting for funding constraints. 
	\cite{BrigoPalCCP} postulate the existence of a risk-neutral pricing measure, and obtain a valuation equation which accounts for counterparty credit risk, funding, and collateral servicing costs. \cite{br} construct a semimartingale framework and provide a backward stochastic differential equation
	(BSDE) for the wealth process that replicates a default-free claim, assuming the trading parties to be default-free. Building on \cite{br},
	\cite{NiRut} study pricing of contracts both from the perspective of the investor and her counterparty, and provide the range of fair bilateral prices.
	The default risk of the trading parties involved in the transaction is accounted for by \cite{BurgardCR}, who derive the partial differential equation representations for the derivative value using replication arguments. \cite{DuffieAndersen} develop a model, that is consistent with asset pricing theories, and importantly account for the impact of funding strategies on the market valuation of the claim. We refer to \cite{CrepeyBieleckiBrigo} for an overview of the literature on valuation adjustments.
	
	A second stream of literature has modeled XVA from a corporate finance perspective. Noticeable contributions in this direction include \cite{AlbaAndersen} who focus primarily on funding valuation adjustments, and \cite{AlbaCrepey} who also consider capital valuation adjustments. These papers view the bank as financed through debt and equity by bondholders and shareholders. Shareholders control the bank and make investment decisions before the bank defaults, while bondholders represent the senior creditors of the bank. Shareholders are wiped out when the bank defaults, while creditors have no decision power until the time of default, but are protected by laws forbidding certain trades that would trigger wealth transfers from them to shareholders.
	In these models, it is assumed that the actual claim is completely hedged by a back-to-back hedge and XVA only arises because of market incompleteness, which originates from two sources. First, counterparty risk cannot be perfectly replicated. Second, the manager of a bank cannot hedge wealth transfers from shareholders to creditors, and shareholders cannot acquire all debt of a bank. Despite the merits of this approach, in particular not having to rely on replication arguments for the value adjustments, it makes two critical assumptions. First, it assumes that the counterparty-free payoffs of the contract are perfectly replicated, rather than designing the replication strategy from first principles (and ignoring potential interaction of risk factors). Second, and most importantly, they assume that the historical and risk-neutral probability measure coincide. This, of course, exposes the calculation of the valuation adjustments to a substantial amount of model risk, which can be accounted for by the techniques proposed in this paper.

We consider a market environment, in which an investor transacts credit default swaps with a counterparty and wants to compute the XVA of her trading position. The trading inefficiencies contributing to the XVA include funding costs due to the difference between treasury borrowing and lending rates, losses originating from premature default of the investor or her counterparty, and costs of posting initial and variation margin collateral. Existing literature on valuation adjustments for credit default swaps has focused on credit and debit valuation adjustments using reduced form models (e.g. \cite{BrigoCapPal}), structural credit models (e.g. \cite{Lipton}), and Markov models based on dynamic copulas to account for wrong-way risk (\cite{BielVal}). These works neither account for the additional costs of funding, nor for model uncertainty.


The distinguishing feature of our framework, relative to the literature surveyed above, is that the investor is uncertain about the rate of return of the counterparty
bond used to hedge counterparty credit risk, {and we compute a robust pricing for the underlying credit default swap portfolio.\footnote{ Several studies have investigated the determinants of {bond returns, including default risk and market liquidity. \cite{Acharya} bucket the bonds into rating classes, ranging from AAA through CCC. They show that the economic contribution of interest rate and default risks to bond returns is larger than the contribution of liquidity under both stressed and normal market regimes.
}} Recent work by \cite{Schmidt} develops a framework that incorporates model uncertainty into defaultable term structure models. {They assume lower and upper bounds for the default intensity and construct uncertainty intervals for the defaultable bond prices, ignoring valuation adjustments.} Our theory parallels that for uncertain volatility introduced by \cite{ALP}. Therein, the authors consider a Black-Scholes type model, in which the volatility of the underlying asset is unknown and only a priori deterministic bounds for its value are prescribed. They derive the {Black-Scholes-Barenblatt} equation characterizing the value of European options in this model; see also \cite{Lyons} for the case of one-dimensional barrier options. \cite{Fouque} generalize the analysis to the case that the volatility fluctuates between two stochastic bounds, arguing that this better captures the behavior of options with longer maturity. Other related works include \cite{Hobson}, \cite{EKJS}, and \cite{DenisMartini} who provide a probabilistic description using the theory of capacities. 

We focus on the impact that uncertainty on the return of the counterparty account has on the valuation of the trade, and compute upper and lower bounds for the XVA. There are both similarities and differences between our setup and the uncertain volatility setup of \cite{ALP}. On the one hand, the differential equations yielding the robust XVA
 are ordinary and of first order, as opposed in uncertain volatility where the price bounds are obtained by solving second-order partial differential equations. Additional simplifications arise in our framework because we do not need to deal with the singularity of probability measures. On the other hand, new technical challenges appear due to the complex relationship between the valuation of the replicating portfolio, the determination of collateral levels, and the close-out requirements of the valuation party.

In our framework, the investor uses {her defaultable bond, the bond of her counterparty, and the bond of reference entities in the CDS portfolio}  to replicate
the XVA process associated with the credit default swaps portfolio.\footnote{{The replication approach to XVA has also been adopted in other studies, including} \cite{CrepeyBieleckiBrigo} and \cite{BurgardCR,BurgardRrisk}.}
{We conduct the analysis in the paper using defaultable (money market) accounts as opposed to bonds, because the value processes of those
accounts are continuous up to the firm's default. In contrast, the bond price of a firm may jump if another firm in the portfolio defaults, introducing
nontrivial technical challenges. We show in Remark~\ref{bond-acc} that there exists a one-to-one relation between defaultable bonds and defaultable accounts,
hence our choice of working with defaultable accounts comes without any loss of generality.

We derive the nonlinear valuation equation that takes into account counterparty credit risk and closeout payoffs exchanged at default. Our valuation equation is a BSDE driven by L\'{e}vy processes, which contains jump-to-default but no diffusion terms. {We sketch in Section \ref{sec:jumpdiffrisk}  how the valuation equation would change when the default intensity processes of reference entities in the underlying swap portfolio are additionally driven by diffusion processes.} We characterize the super-replicating price of the transaction as the solution to a nonlinear {system of} ODEs, obtained from projecting the nonlinear BSDE tracking the XVA process onto the smaller filtration exclusive of investor and counterparty credit events information. 
The system consists of an ODE, whose solution is the value of the transaction cash flows ignoring market inefficiencies, and additional ODEs that yield the XVA of the portfolio.
 Intuitively, the super-replicating price is the value attributed to the trade by an investor who positions herself in the worst possible economic scenario.

 We find that the super-replicating price and the corresponding super-replicating strategies may not be recovered by simply plugging one of the extremes of the uncertainty interval into the valuation equation. 
Our analysis indicates that, whether to use the lower or upper extreme of the uncertainty interval in the super-replication strategy, depends on the relation between the current value of the XVA replicating portfolio and the close-out value of the transaction. The trader wants to be robust against the most negative outcome, and therefore will choose the extreme of the interval that minimizes the instantaneous change in the value of the position. This will in turn require the investor to initially hold the maximal wealth to implement this replicating strategy, hence leading to the maximum initial value of the portfolio. As long as the portfolio replicates the trade at the terminal time, its initial value provides an upper bound on the value of the XVA.
For example, if the strategy replicating the XVA requires, at a given time, the investor to be short the counterparty's defaultable account, i.e., a positive jump would arise at the counterparty default (this would be the case if the value of the XVA replicating portfolio lies below the close-out value), then the trader would choose to use the upper extreme of the uncertainty interval, because this corresponds to the maximal default intensity and thus minimizes the instantaneous change in value. As the required replicating position may switch from short to long and vice versa several times before the close-out time, the extreme of the default interval used in the valuation of the superreplicating strategy will change, too. 

We perform a comparative statics analysis to quantify the dependence of the XVA and its replication strategy on portfolio credit risk and default contagion. We use a model of direct contagion, in which default intensities of reference entities, investor, and counterparty are piecewise constant and only jump when one of these firms default. {Our analysis finds that if the default intensity of the investor's counterparty increases, either due to idiosyncratic motives or to contagion effects triggered by the default a reference entity in the portfolio, the XVA decreases in absolute value.} This is because, under these circumstances, the investor needs to replicate the underlying portfolio transaction for a smaller period of time, and thus incurs smaller financing costs. As direct contagion increases, defaults tend to cluster and amplify the impact of portfolio credit risk on the default intensity of investor and her counterparty. The financing costs of the replication strategy get lower, and a payer CDS investor needs to use a larger number of shares of her defaultable account to replicate the jump to closeout at her default time, compared with the number of shares of her counterparty defaultable account needed to replicate the jump to closeout at her counterparty's default time. This is because if the investor (who is replicating her long portfolio position) defaults, then she needs to replicate a larger jump to closeout if the portfolio credit risk is higher and thus her moneyness increases. 

}

The rest of the paper is organized as follows. We develop the market model in Section \ref{sec:model}. We introduce the valuation measure, collateral process and close-out valuation in Section \ref{sec:claim}. We introduce the replicating wealth process
and the notion of arbitrage in Section~\ref{sec:wealth}. We develop a robust analysis of the XVA process in Section~\ref{sec:OneEntity}, and discuss how the valuation equations generalize to an underlying CDS portfolio that also presents diffusion risk. Section~\ref{sec:numeranalysis} presents a numerical analysis of XVA and its replication strategies on a multi-name portfolio. Section \ref{sec:conclusions} concludes.

\section{Model}\label{sec:model}

Our framework builds on that proposed by \cite{BCS} in that it uses a reduced form model of defaults and maintains the distinction between universal and investor specific instruments. The model economy consists of $N$ firms, indexed by $i=1,\ldots,N$, whose default events constitute the sources of risks in the portfolio. We use $I$ and $C$ to denote, respectively, the trader (also referred to as investor throughout the paper) executing the transaction and her counterparty. Let $\left( \Omega, \mathcal F, \Px\right)$ be a probability space rich enough to support the following constructions. We assume the existence of $N+2$ independent and identically distributed unit mean exponential random variables $\mathcal E^i$, $i = 1, \ldots, N, I,C$. The default time of each firm $i$ is defined to be the first time its cumulative intensity process exceeds the corresponding exponentially distributed random variable, i.e., $\tau^i = \sup\bigl\{ t\ge0 \colon \int_0^t h^{i,\Px}_sds > \mathcal E^i\bigr\}$. Accordingly, we use the default indicator process $H^i_t= \ind_{\{\tau^i\leq t\}}$, $t\geq0$, to track the occurrence of firm $i$'s default. The background filtration $\mathbb{F} := \bigl(\mathcal{F}_t\bigr)_{t \geq 0}$, where $\mathcal F_t := \sigma\bigl( H^j_u; u \leq t \colon j\in\{1, .., N\} \bigr)$, contains information about the risk of the portfolio, i.e., of the default of the $N$ firms referencing the traded securities, but not about the defaults of the investor $I$ and her counterparty $C$. {The default intensity processes $\bigl( h^{i,\Px}_t \bigr)_{t\geq 0}$, $i\in\{1, ..., N,I,C\}$, are constructed so that they are adapted to the background filtration $\mathbb{F}$, i.e., the default intensity at a given time $t$ depends on the firms' defaults occurring before time $t$. We report the details of this construction in Appendix~\ref{sec:defintconstr}.

We denote the filtration containing information about the investor and counterparty defaults by $\mathbb{H} = \bigl(\mathcal{H}_t\bigr)_{t \geq 0}$, where $\mathcal H_t = \sigma\bigl( H^j_u; u \leq t \colon j\in\{I, C\} \bigr)$. By construction, the default intensities $h^{i,\Px},~i\in\{1, .., N, I, C\}$, are piecewise  deterministic functions of time (we thus work in the framework of piecewise-deterministic Markov processes, see \cite{Davis}). We furthermore require that they are piecewise continuous and uniformly bounded. The enlarged filtration, including both portfolio risk (default events of the $N$ firms referencing portfolio securities) and counterparty risk (default events of investor and her counterparty), is denoted by $\mathbb{G} = \bigl(\mathcal{G}_t)_{t \geq 0} = \bigl(\mathcal{F}_t \vee \mathcal{H}_t\bigr)_{t \geq 0}$. We will consider the augmented filtrations, i.e., the smallest complete and right-continuous filtrations encompassing the natural filtrations, and denote them by $\mathbb{F}$, $\mathbb{H}$, $\mathbb{G}$ (with a slight abuse of notation). For future purposes, we define the martingale compensator processes $\varpi^{i,\Px}$ of $H^i$ as
\[
\varpi^{i,\Px}_t := H^i_t - \int_0^t \bigl(1-H^i_s\bigr)h^{i,\Px}_s \, ds,~i \in \{1,\ldots,N,I,C\}.
\]
By construction, these compensator processes are $\mathbb{F}$-martingales for $i \in \{1,\ldots,N\}$, and $\mathbb{G}$-martingales for $i \in \{1,\ldots,N,I,C\}$.

The defaultable account rates of all firms in the portfolio and of the investor are known to all market participants. The trader, however, only has limited information about the actual rate of the counterparty defaultable account, and in particular only knows its upper and lower bound.


\paragraph{Replicating instruments}
The goal of the investor is to replicate a portfolio of credit default swaps (CDS) written on $N$ different reference entities, denoted by $1,2,\ldots,N$. All CDSs are assumed to mature at the same time $T$. The credit risk exposure associated with this portfolio is replicated using both \textit{universal} and \textit{investor specific} instruments. The universal instruments are available to all market participants, while the investor specific instruments are accessible solely to the investor and not to other market participants. The universal instruments include (defaultable) bonds underwritten by the reference entities in the credit default swaps portfolio as well as by the trader and her counterparty. {As opposed to modeling defaultable bonds directly, we model the defaultable accounts associated with investor, counterparty, and reference entities in the CDS portfolio. {These securities are typically employed as numeraires until they default, and are used to define the survival measures. For instance, \cite{CollinDufresne} study a survival spot measure where the numeraire is a defaultable account. 
We provide more details on the relationship between defaultable bonds and defaultable accounts in Remark \ref{bond-acc}.}

Under the physical measure $\Px$, for $i\in \{1,\ldots,N, I, C\}$, and $0 \leq t \leq T$, the dynamics of the defaultable {account} processes with zero recovery at default are given by
\begin{equation}\label{eq:priceproc}
dB^i_t = \mu_t^i B^i_t \, dt - B^i_{t-} \,dH^i_t, \qquad B^i_0 = 1,
\end{equation}
where $\bigl(\mu^i_t\bigr)$, $i \in \{1,\ldots,N,I,C\}$, are $\Fx$-adapted {and thus piecewise deterministic processes, potentially jumping at discrete times corresponding to default events. 
	We assume that the rates $\mu^i$, $i \in \{1,\ldots,N,I\}$ are observable while the investor has no further information about $\mu^C$ except for that it is constrained to lie in the interval $[\underline \mu^C, \overline \mu^C]$.}

The \textit{investor specific} instruments include her funding and collateral accounts. We assume that the investor lends and borrows from her treasury desk at, possibly different, rates  $\rfp$ (the lending rate) and $\rfm$ (the borrowing rate). Denote by $B^{r_f^{\pm}}$ the cash accounts corresponding to these funding rates.
An investment strategy of $\xi^f:=(\xi^f_s; \;  s\geq0)$ shares in the funding account yields an account value $B^{r_f} := (B^{r_f}_s; \; s\geq 0)$ given by
  \begin{equation}\label{eq:Brf}
    B^{r_f}_t  := B^{r_f}_t \bigl(\xi^f) = e^{\int_0^t r_f(\xi^f_s) ds},
  \end{equation}
where
  \begin{equation}\label{eq:rrf}
    r_f  := r_f(y)= \rfm \ind_{\{{y < 0}\}}+\rfp \ind_{\{{y > 0}\}}.
  \end{equation}

\paragraph{Collateral}

The trader and the counterparty use a collateral account to mitigate counterparty risk. Following the standards set by the Basel  Committee on Banking
Supervision  (BCBS)  and  the  International  Organization  of  Securities  Commissions  (IOSCO) (see \cite{BISMargin}), the collateral consists of variation margins, tracking the changes in market value of the traded portfolio and denoted by $VM$, and initial margins that are used to mitigate the gap risk at the close-out of the transaction and denoted by $IM$.\footnote{Notice that initial margins are updated on a regular basis (not just posted once at the inception of the trade as the name might suggest), as it is the case for variation margins. Variation margins are usually updated at a higher frequency (intraday or at most daily) than initial margins, which are resettled daily or even at lower frequency.} The European Market Infrastructure Regulation (EMIR) posits at least daily updates for variation margins and requires a revaluation of initial margins at least every ten days (see \cite{EMIR}). In the United States, the Commodity Futures Trading Commission requires daily updates on initial margins (\cite{CFTC}). Mathematically, the collateral process $M:={(M_t; \; t\geq 0)}$, $M = VM + IM$, is an $\mathbb{F}$ adapted stochastic process which we assume to be positive if the investor posts collateral (is \textit{collateral provider}) and negative if she receives collateral (is the \textit{collateral taker}).

Denote by $\rcp$ the interest rate on collateral demanded by the investor when she posts to her counterparty, and by $\rcm$ the rate on collateral demanded by the counterparty when the investor is the collateral taker.
The value of the collateral account at time $t$ is then given by
\[
    B^{r_m}_t = e^{\int_0^t r_m(M_s) ds},
\]
where
\[
    r_m := r_m(x) = \rcm \ind_{\{x<0\}} + \rcp \ind_{\{x>0\}}.
\]

Denoting by $\psi_t^m$ the number of shares of collateral account $B^{r_m}_t$ held by the trader at time $t$, we have the following relation
\begin{equation}\label{eq:collrel}
    \psi^{m}_t B^{r_m}_t  = - M_t.
\end{equation}
The collateral amount $M_t$ received or posted at time $t$ will be determined by a valuation party, as discussed in the next section. Figure \ref{fig:transflow} describes the mechanics of the entire flow of transactions.

\begin{figure}[ht]
    \centering
    \begin{tikzpicture}[thick,scale=0.9, every node/.style={transform shape}]
        \node[punkt, inner sep=10pt] (trader) {Trader};
        \node[punkt, inner sep=10pt,  left=2cm of trader] (funder) {Treasury Desk}
            edge[pil, bend left=35, blue, dotted] (trader)
            edge[pil, <-, bend right=35, blue, dotted] (trader)
            edge[pil, <-, bend left=10, blue, dashed] (trader)
            edge[pil, bend right=10, blue, dashed] (trader);
        \node[above left =1cm of trader] (rfp) {$\rfp$};
        \node[below left =1cm of trader] (rfm) {$\rfm$};
        \node[left =0.5cm of trader] (fundingcash) {Cash};
        \node[punkt, inner sep=10pt,  right=5.5cm of trader] (bond) {Defaultable accounts}
            edge[pil, bend left=30, black, solid] (trader)
            edge[pil, <-, bend right=30, black, solid] (trader);
        \node[right =0.3cm of trader] (fundingcash) {Defaultable accounts $B_i$};
        \node[punkt, inner sep=10pt,  below=4cm of trader] (counter) {Counterparty}
            edge[pil, <-, bend right=50, blue, dotted] (trader)
            edge[pil, bend left=50, blue, dotted] (trader)
            edge[pil, <-, bend left=35, blue, dashed] (trader)
            edge[pil, bend right=35, blue, dashed] (trader);
        \node[below =1.5cm of trader] (vmv) {Variation};
        \node[below =1.9cm of trader] (vmi) {Margin};
                \node[below = 3.5cm of trader.west] (rcp) {$\rcp$};
        \node[below =3.5cm of trader.east] (rcm) {$\rcm$};
        \node[punkt, inner sep=10pt, below right=4cm of trader] (seg) {Segregated Account}
         edge[pil, ->, bend right=4, blue, dotted] (trader)
         edge[pil, <-, bend left=4, blue, dashed] (trader);
        \node[above right =1.6cm of counter] (imi) {Initial};
        \node[below =0.0cm of imi] (imm) {Margin};
        \node[above right =0.2cm of imi] (rcp2) {$\rcp$};
    \end{tikzpicture}
    \caption{Trading: Solid lines are purchases/sales, dashed lines borrowing/lending, dotted lines interest due; blue lines are cash, and black lines are defaultable account purchases for cash. Note the difference between the two-sided variation margin and the initial margin that is kept in a segregated account.}
    \label{fig:transflow}
\end{figure}
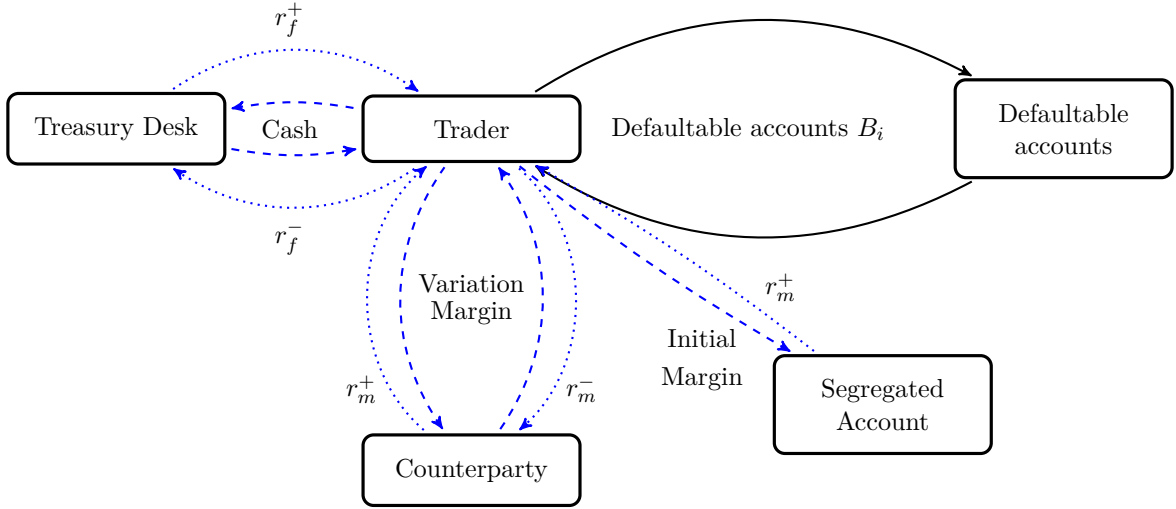

\section{{Valuation Measure}, Collateralization and Close-out} \label{sec:claim}

{We take the perspective of a trader who sells a portfolio of CDSs, and determines its value by constructing a replicating portfolio. Such a portfolio accounts for spread payments throughout the life of the contracts, and when any reference entity defaults the value of the corresponding CDS contract equals its loss rate. The wealth process associated with the portfolio uses defaultable accounts of the underlying reference entities to replicate the market risk of the transaction, and defaultable accounts of the trading parties to replicate the counterparty risk of investor and of her counterparty. 
	Because the trader does not know the exact default intensity of her counterparty, such a replication argument can only provide price bounds. In particular, the upper bound provides a reliable benchmark to measure the potential losses incurred by the trader if the portfolio is sold at a price lower than the upper bound.}

\begin{remark}
The trader aims to compute the difference between the price at which the transaction is settled, and the market value of the transaction, so that she can identify the underlying risk factors and allocate them to different desks within the bank. This difference is referred to as XVA.

It is important to introduce a finer distinction between the different sources of surcharges and unreplicable risk (referred to, e.g., as  CVA, FVA, KVA) to correctly allocate them to the managing desks. Hence, when calculating the exit price, i.e., the price at which the portfolio can be liquidated on the open market (this is relevant for tax and regulatory purposes), one needs to account for these components at a higher level of granularity. One of these components is the KVA, defined as the financing cost for the capital at risk, set aside by shareholders of the investor's firm. KVA should be calculated under the historical measure, which is typically assumed to be the same as the risk neutral measure to preserve analytical tractability. {Such an approach is followed, for example, by \cite{Albanese} who define the KVA as the solution to a BSDE under the risk neutral measure. \cite{Green} derive the KVA using an extension of the semi-replication approach in \cite{BurgardCR} by grouping together all capital dependent terms in the Feynman-Kac representation of their pricing PDE.} Our analysis deals with entrance prices, i.e., prior to decomposing the trade into risk sources and splitting it to the various desks. Nevertheless, we compute the (super)replicating price of the transaction, which is robust against the specific choices of physical and pricing measure {because the assumed bounds for the  {account} rates are constant, and thus independent of the choice of the measure. From a different perspective, our methodology can be seen as providing general bounds for XVA, that in turn yield bounds for its individual components such as KVA.} {These bounds can be seen in analogy to the bounds for equity options: The price of market incompleteness originated from stochastic volatility can be bounded after specifying bounds on the volatility process. In the present model, the financing costs for the capital at risk, incurred for the impossibility of fully hedging counterparty risk, are not explicitly accounted for. However, after deducting credit, debit and funding valuation adjustments, the residual costs accounts for KVA. The bounds on the counterparty bond rate (respectively, the risk neutral default intensity) would then imply bounds for KVA, and the lower and upper bounds would coincide if the counterparty's default intensity is known with certainty.}
	



\end{remark}

Next, we discuss \textit{public} and \textit{private} valuations. Private valuations are based on discount rates, which depend on investor specific characteristics, while public valuations depend on publicly available discount factors.
Specifically, public valuations are needed for the determination of collateral requirements and the close-out value of the transaction. They are determined by a valuation agent who might be either one of the parties involved in the transaction or a third party, in accordance with market practices reviewed by the International Swaps and Derivatives Association (ISDA). The valuation agent determines the closeout value of the transaction by calculating the so-called clean price of the derivative, using the discount rate $r_D$ and the {account rates of the firms in the portfolio, $\mu^i$, $i \in \{1, \ldots N\}$} (we recall that the latter are known to the valuation agent). Throughout the paper, we will use the superscript $\wedge$ when referring specifically to public valuations.

The replicating process will stop before maturity if the trader or her counterparty were to default prematurely. We thus define the terminal time of the trade (i.e., the earliest between the default time of either party or the maturity $T$ of the transaction) as $\tau := \tau^I \wedge \tau^C \wedge T$. The valuation done by the agent is mathematically represented as pricing the trade under the valuation measure $\Qxx$ associated with the publicly available discount rate $r_D$ chosen by the agent. The measure  $\Qxx$ is equivalent to $\Px$ and their relation is specified by the Radon-Nikod\'{y}m density
\begin{equation}\label{eq:RN}
    \frac{d\Qxx}{d\Px} \bigg|_{\mathcal{F}_{{\tau \wedge( \tau^1 \vee ...\vee \tau^N)     }}} = \prod_{i\in \{1, \ldots, N, I, C\}} \Biggl(\frac{(\mu^{i} - r_D)(\tau \wedge \tau^{i})}{\int_0^{\tau \wedge \tau^{i}} h^{i,\Px}_sds }\Biggr)^{H^i_{\tau \wedge \tau^i}} e^{\int_0^{\tau \wedge \tau^{i}}(r_D-\mu^{i}+h^{i,\Px}_u)du}.
\end{equation}

\begin{remark}
	As the valuation measure is used to determine the clean price of the transaction, it needs not depend on the default intensities of the investor $I$ and her counterparty $C$. Nevertheless, we have included both of these default intensities in the definition of $\Qxx$ because this will simplify the exposition in later sections of the paper. In particular, we do not need to introduce a different measure for the investor's valuation.
\end{remark}
%

The $\Qxx$-dynamics of the defaultable accounts follow from Girsanov's theorem and are given by
  \begin{equation}
    dB^i_t  = r_D B^i_t \, dt - B^i_{t-} d\varpi^{i,\Qxx}_t,
    \label{Bieq}
  \end{equation}
where $\varpi^{i,\Qxx} := (\varpi^{i,\Qxx}_t; \; 0 \leq t \leq \tau \wedge \tau^{(N)})$ are $(\mathbb{F}, \Qxx)$-martingales, and $\tau^{(i)}$ denotes the $i$-th order statistics of the default times, {$i=1, ...,N$}. These martingales can be represented explicitly as
$\varpi^{i,\Qxx}_t = \varpi^{i,\Px}_t + \int_0^t \bigl(1-H^i_u\bigr) (h^{i,\Px}_u - h^{i,\Qxx}_u) du$, where the processes $h^{i,\Qxx} = \mu^i-r_D$, $ i \in \{ 1, \ldots, N, I, C\}$, (and $\mu^i$, $i \in \{1,\ldots,N, I, C\}$, are the rate of returns of the {defaultable accounts} associated with the reference entities, trader and her counterparty), are the firms' default intensities under the valuation measure and assumed to be positive.

\begin{remark}\label{bond-acc}
	We note that under the valuation measure $\Q$, {the actual discounted bond price processes} have to be $\mathcal{G}$-martingales. {Therefore, for $i \in \{1,\ldots,N, I, C\}$ and denoting by $P_i$ the price of the defaultable bond of firm $i$, we have}
	\[
	\frac{P^i_t}{B^i_t} = \E^\Q\Biggl[ \frac{P^i_T}{B^i_T} \, \Bigg\vert \, \mathcal{G}_t \Biggr]
	\]
	and hence, as the bond value at maturity is equal to one, it follows that the actual bond prices are determined from the modeled {defaultable accounts} via
	\[
	P^i_t = B^i_t \E^\Q\Bigl[\frac{1}{B^i_T} \, \Big\vert \, \mathcal{G}_t \Bigr].
	\]
Such a modeling approach has the advantage that the value processes of the defaultable accounts are continuous up to default, while the actual bond prices may
jump down at the time when another bond defaults. To see this,
notice that the terminal condition of a unit notational bond that has not defaulted is always one. If the return rate of a bond changes at the time another bond defaults, then the discount rate used in the valuation of the surviving bond would be different, leading to a change in the bond price.
This is empirically relevant and, in quite a few cases, captures accurately the behavior of corporate bonds (a clear example are the sovereign defaults that caused local corporate defaults in the 1997 Asian financial crisis).
\end{remark}

\subsection{Collateral} \label{sec:repclaim}

The public valuation process of the credit default swap portfolio, as determined by the valuation agent, is given by
\[
	\hat{V}_t = \sum_{i=1}^N z^i \hat{C}^i_t,
\]
where $\hat{C}^i_t$ is the time $t$ value of the credit default swap referencing entity $i$. The variable $z^i$ indicates if the trader sold the $i$-th swap to her counterparty ($z^i = 1$) or purchased it from her counterparty ($z^i = -1$). In the case the swap is purchased, the trader pays the spread times the notional to her counterparty, and receives the loss rate times the notional at the default time of the reference entity, if it occurs before the maturity $T$. This is the so-called  ``clean price'', and does not account for credit risk of the counterparty, collateral or funding costs. Clearly, the public valuation of the portfolio is just the sum of the valuation of the individual CDSs.

{The Basel Committee on Banking Supervision (BCBS) and the International Organization of Securities Commissions (IOSCO) released a second consultative document on margin requirements for non-centrally cleared derivatives in February, 2013; see \cite{BCBS}. This document provides minimum standards for initial margin posting related to non-centrally cleared derivatives. It highlights the importance of separating between the initial margin posted by the counterparty, and the initial margin posted by the investor, so to avoid any netting between these two accounts and protecting each party from gap risk. Our collateralization process is consistent with these market practices.} The variation margin is set to be a fixed ratio of the public valuation of the portfolio, while the initial margin is  designed to mitigate the gap risk and is calculated using value at risk. Such a risk measure is set to cover a number of days of adverse price/credit spread movements for the portfolio
position with a target confidence level.\footnote{Both EU and US authorities require initial margins to cover losses over a liquidation period for ten days in 99\% of all realized scenarios (\cite{EMIR, CFTC}).} Note that there is an important difference between initial and variation margins. Variation margins are always directional and can be rehypotecated (i.e., it flows from the paying party to the receiving party; the latter may use it for investment purposes), whereas initial margins have to be posted by \textit{both} parties and need to be kept in a segregated account, thus they cannot be used for portfolio replication.
Rehypothecation is a very popular practice because it lowers the cost of collateral remuneration (\cite{Singh}) and has been accounted for by existing literature on XVA (e.g. \cite{BrigoPerPal}).
Moreover, we assume that collateral is posted and received in the form of cash, which is practically the most common form of collateral.\footnote{More precisely, cash is the predominant form of collateral used for variation margins, and it accounts for about 80\% of the total posted variation margin amount. Initial margins are usually delivered in the form of government securities (see, for instance, page 7 of \cite{ISDA17}). Overall, the amount of variation margin posted for bilaterally cleared derivatives contracts was about \$ 173 billion in 2017, whereas the variation margin accounted for \$870 billions (see page 1 therein).}

Thus, on the event that neither the trader nor her counterparty have defaulted by time $t$, and the reference entities in the portfolio have not all defaulted, the collateral process is defined as
  \begin{equation}\label{eq:rulecoll}
    M_t : =IM_t +VM_t = \biggl(\beta\Bigl( \VaR_{q} (\hat V_{(t+\delta) \wedge T}-\hat V_t\, \vert \, \tau^{(N)} > t)\Bigr)^{+}+ \alpha \hat{V}_t\biggr)\ind_{\{\tau \wedge \tau^{(N)} > t\}},
  \end{equation}
where for a real number $x$ we are using the notation $x^+ = \max(x,0)$. In the above expression, $0 \leq \alpha \leq 1$ is the collateralization level, $\delta>0$ is the delay in collateral posting, $q$ is the level of risk tolerance and $\beta$ is stress factor. The case $\alpha = 0$ corresponds to zero collateralization, while $\alpha = 1$ means that the transaction is fully collateralized. The positive part of the value at risk quantity captures the fact that initial margins cannot be rehypothecated. Hence, the wealth process associated with the investor's trading strategy does not include received initial margins.

\subsection{Close-out value of transaction}\label{sec:closeout}

We follow the risk-free closeout convention in the case of default by the trader or her counterparty. According to this convention, each party liquidates the position at the market value when the other trading counterparty defaults. Hence, the value of the replicating portfolio will coincide with the third party valuation if the amount of available collateral is sufficient to absorb all occurred losses. If this is not the case, the trader will only receive a recovery fraction of her residual position, i.e., after netting losses with the available collateral.
{Note that, in practice, the settlement at the third party valuation takes some time. This induces gap risk because the value of the transaction typically
fluctuates between the actual default time and the settlement time (see \cite{BrigoPalCCP}).} While superhedging of the actual settlement price (including gap risk) is not possible, our approach aims at finding a superhedge of the market valuation at default time. {Initial margins act as a cushion against gap risk, and are computed using a tail risk measure.}
 Let us denote by $\theta$ the value of the replicating portfolio at $\tau <T$. This is given by
  \begin{equation}
\theta  := \hat{V}_{\tau} + \ind_{\{\tau^C<\tau^I \}} L^C Y^- - \ind_{\{\tau^I<\tau^C \}} L^I Y^+
  \label{eq:closeoutterm}
  \end{equation}
where, for a real number $x$, we are using the notation $x^- = \max(-x,0)$. In the above expression, $Y:= \hat{V}_\tau - M_{\tau -} =(1-\alpha)\hat{V}_\tau - \beta\Bigl( \VaR_{q} (\hat V_{(\tau+\delta)\wedge T}-\hat V_\tau)\Bigr)^{+}$ is the value of the claim at default netted of the posted collateral, and  $0 \leq L^I, L^C \leq 1$ are the loss rates on the trader and counterparty claims, respectively. Alternatively, we can represent the value of the portfolio at default as
  \[
 \theta =  \theta(\tau, \hat{V},M)  = \ind_{\{\tau^I<\tau^C \}} \theta^{I}(\hat V_{\tau}, M_{\tau-}) +  \ind_{\{\tau^C<\tau^I \}} \theta^{C}(\hat V_{\tau}, M_{\tau-}),
  \]
where we define
\[
   \theta^{I}(\hat{v}, m)   := \hat{v} -   L^I  \bigl(\hat{v} - m\bigr)^{+}, \qquad \qquad
   \theta^{C}  (\hat{v}, m)   := \hat{v} +   L^C  \bigl(\hat{v} - m\bigr)^{-},
\]
and recall
\[
M_t  =  \alpha\hat{V}_t + \beta\Bigl( \VaR_{q} (\hat V_{(t+\delta) \wedge T}-\hat V_t\, \vert \, \tau \wedge \tau^{(N)}> t)\Bigr)^{+}.
\]

Note that $\bigl( \theta^I(\hat{V}_t, M_t) \bigr)_{t\geq 0}$ and $\bigl( \theta^C(\hat{V}_t, M_t) \bigr)_{t\geq 0}$ are also piecewise deterministic and piecewise continuous $\mathbb{F}$-adapted processes.
\begin{remark}\label{rem:lossuncert}
In practice, the actual value of the recovery rate is also uncertain and unknown till the end of the resolution process. Empirical research has shown that it tends to be inversely related to the default probability of the bond issuer (\cite{Altman}). From a mathematical perspective, adding uncertainty in the counterparty's account rate does not introduce conceptual challenges. This is
because the closeout value of the transaction given in~\eqref{eq:closeoutterm} is affine in the loss rate. Apart from introducing additional notational burden, such an uncertainty can be handled by a straightforward adaptation of the comparison argument in Theorem \ref{thm:comp}, where we would exploit the monotonicity of the {function $\theta^C$ in the loss rate $L^C.$} 
\end{remark}

\section{Wealth process \& Arbitrage} \label{sec:wealth}
{We analyze a stylized model of single name credit default swaps. If the trader purchases protection from her counterparty against the default of the $i$-th firm, then the trader makes a stream of continuous payments at a rate $S_i$ of the notional to her counterparty, up until contract maturity or the arrival of the credit event, whichever occurs earlier. Upon arrival of the $i$-th firm's default event, and if this occurs before the maturity $T$, the protection seller pays to the protection buyer the loss on the notional, obtained by multiplying the loss rate $L_i$ by the notional. As the notional enters linearly
in all calculations, we fix it to be one.}

Recall that $\xi^i$ denotes the number of shares of the defaultable account associated with the reference entity $i$, $\xi^{f}$ the number of shares in the
funding account, and we use $\xi^I$ and $\xi^C$ to denote the number of shares of trader and counterparty defaultable accounts, respectively. Using the identity~\eqref{eq:collrel}, we may write the wealth process as a sum of contributions from each individual account:
\begin{equation}\label{eq:wealth}
V_t := \sum_{i=1}^N \xi_t^i B_t^i + \xi_t^I B_t^I + \xi_t^C B_t^C + \xi_t^f B_t^{r_f} - \psi_t^{m} B_t^{r_m}.
\end{equation}

For the purpose of arbitrage-free valuation, it is important to consider not only the actual CDS portfolio, but an arbitrary multiple of it. Hence, we will consider a multiple $\gamma$ of the acquired portfolio, and focus on self-financing strategies.
\begin{definition}
	A collateralized trading strategy ${\bm\varphi} := \bigl(\xi_t^1,\ldots,\xi_t^N,\xi_t^f,\xi_t^I, \xi_t^C \; t \geq 0\bigr)$ associated with $\gamma$ shares of a portfolio $w=(w^1,w^2,\ldots,w^N)$, where $w^i \in \{0,1\}$ for $i=1,\ldots,N$, is \textit{self-financing} if, for $t \in [0, \tau \wedge \tau^{(N)}]$, it holds that
	\begin{align}
	\nonumber V_t(\gamma) &:= V_0(\gamma) +  \sum_{i=1}^N \int_0^t \xi_u^i \, dB_u^i+ \int_0^t \xi_u^I \, dB_u^I + \int_0^t \xi_u^C \, dB_u^C + \int_0^t \xi_u^f \, dB_u^{r_f} -\int_0^t \psi_u^{m} \, dB_u^{r_m} \\
	& \phantom{==}+ \gamma  \sum_{i=1}^N w^i \eta^i \; (\tau^i \wedge t) .
	\label{eq:vself}
    \end{align}
	The above expression takes into account the running spread payments $\eta^i$, $i=1,\ldots,N$ received/paid by the investor for the $i$-th CDS contract sold to (resp. purchased from) her counterparty. The set of admissible trading strategies consists of $\mathbb{F}$-predictable processes ${\bm\varphi}$ such that the portfolio process $V_t(\gamma)$ is bounded from below (cf. \cite{Delbaen}).
\end{definition}


%

Before discussing the arbitrage-free valuation of the CDS portfolio, we have to clarify the assumptions under which the underlying market  is free of arbitrage from the investor's perspective (conceptually, we follow \cite[Section 3]{br}). Thus, to start with, we exclude the CDS instruments from our consideration, and consider a trader who is only allowed to buy or sell shares of the defaultable
accounts (associated with the reference entities, her counterparty or the investor's firm itself) and to borrow or lend money from the treasury desk.

\begin{definition}\label{def:no-arb}
 The market $(B^1,B^2,\ldots,B^N,B^I,B^C)$ admits \textit{investor's arbitrage} if, given a non-negative initial capital $x \geq 0$, there exists an admissible
 trading strategy ${\bm\varphi} = (\xi^1,\xi^2,\ldots,\xi^N,\xi^f,\xi^I,\xi^C)$ such that $\Px \bigl[V_\tau \geq e^{\rfp \tau}x \bigr] =1$ and
 $\Px\bigl[V_\tau > e^{\rfp \tau}x\bigr] >0$. If the market does not admit investor's arbitrage for a given level $x \geq 0$ of initial capital,
 the market is said to be arbitrage free from the investor's perspective.
\end{definition}

We impose the following assumption and argue that it provides a necessary and sufficient condition for the absence of arbitrage.

\begin{assumption}\label{ass:nec+suff}
$ r_D \vee \rfp < \min_{i \in \{1,\ldots N, I\}} \mu^i {\wedge \underline \mu^C}$.
\end{assumption}

\begin{remark}\label{rem:measure}
Necessity: The condition $ r_D < \min_{i \in \{1,\ldots N, I, C\}} \mu^i$ is needed for the existence of the valuation measure defined in Eq.~\eqref{eq:RN} ($h^{i,\Qxx} = \mu^i-r_D$ and risk-neutral default intensities must be positive). {Thus, we should impose $\mu^C > r_D$, but as the true {account} rate $\mu^C$ is unobservable, we impose instead the slightly stronger $\underline{\mu}^C > r_D$.}
The condition $\rfp < \min_{i \in \{1,\ldots N, I\}} \mu^i\wedge \underline \mu^C$ has an even more practical interpretation because it precludes the arbitrage opportunity of short selling the defaultable accounts while investing the proceeds in the funding account. Strictly speaking, the condition $r_D < \mu^I \wedge \underline\mu^C$ is not necessary from an arbitrage point of view, because it addresses only the soundness of the market from the perspective of the valuation party. While $ r_D < \min_{i \in \{1,\ldots N\}} \mu^i$ is necessary to conclude that the valuation party's market model is free of arbitrage, one might hypothesize a situation in which $r_D {\geq } \mu^I \wedge \underline\mu^C$. From a practical perspective, this is however rather unlikely, as $r_D$ is typically assumed to be an overnight index swap (OIS) rate and as such lower than the return rates of the defaultable {accounts}.
\end{remark}

Having argued about the necessity in the above remark, we show that Assumption \ref{ass:nec+suff} is also sufficient to guarantee that the underlying market (i.e., excluding the credit default swap securities) is free of arbitrage. The proof proceeds along very similar lines as Proposition 4.4 in \cite{BCS}, and is delegated to the Appendix.

\begin{proposition}\label{prop:arb-market}
Under Assumption \ref{ass:nec+suff}, the model does not admit arbitrage opportunities for the investor for any $x \geq 0$.
\end{proposition}

As in \cite{BCS}, we will define the notion of an arbitrage free price of a derivative security from the investor{'}s perspective. We assume that the investor has zero initial capital, or equivalently, she does not have liquid initial capital that can be used for replicating the claim until maturity. The
replicating portfolio will thus be implemented through purchases/sales of shares of the defaultable accounts and of the funding account.

\begin{definition}\label{def:arb-price}
	The valuation $P \in \mathbb{R}$ of a derivative security with terminal payoff $\vartheta \in \mathcal{F}_T$ is called \textit{investor's arbitrage-free} if for all $\gamma \in \mathbb{R}$, buying $\gamma$ securities for $\gamma P$ and trading in the market with an admissible strategy and zero initial capital, does not create investor's arbitrage.
\end{definition}

Let $V_t$ represent the price process of the replicating portfolio, and given by the supremum over all arbitrage free prices. Then we define  the \textit{total valuation adjustment} $\XVA$ as the difference between this upper arbitrage price and the clean price, i.e.,
\begin{equation}
\XVA_t(\gamma) = V_t(\gamma) - \gamma \hat{V}_t.
\label{eq:XVAdef}
\end{equation}
$\XVA$ thus quantifies the total costs (including collateral, funding, and counterparty risk related costs) incurred by the trader to replicate the sold CDS portfolio. Notice that, at time $t$, the investor does not know the actual {counterparty {account} rate $\mu^C$} for the time interval $[t,\tau]$. Hence, she is not able to execute the replication strategy yielding the value process $V$, because all what she knows about {the {account} rate is that $\underline \mu^C \le \mu^C \le \overline \mu^C$}. Therefore, she will have to consider the worst case, accounting for all possible $\mathbb{F}$-predictable dynamics of {the {account} rate process in the interval $\bigl[ \underline \mu^C, \overline \mu^C \bigr]$}. Denote the valuation of the replicating portfolio {when $\mu^C = \mu$ by $V^{\mu}$}. 
 The \textit{robust XVA} is defined as
\begin{equation}
{\rXVA_t(\gamma) = \essup_{\mu \in [\underline \mu^C, \overline \mu^C] }V_t^\mu(\gamma) \ind_{\{t<\tau\}}  + V_t(\gamma) \ind_{\{t \geq\tau\}} - \gamma \hat{V}_t.}
\label{eq:rXVAdef}
\end{equation}
Notice that the supremum is taken over all admissible valuations only prior to the trader or her counterparty's default. {In particular, the valuation process at and after default depends only on the closeout value and thus does not depend on the extremes $\underline \mu_C$ and $\overline \mu_C$ of the uncertainty interval.} 

\section{Robust XVA for Credit Swaps}\label{sec:OneEntity}
In this section, we derive explicit representations for the robust XVA of a credit default swap portfolio. To highlight the main mathematical arguments and economic implications of the results, we start analyzing the case of a single credit default swap in Section \ref{sec:dynXVA}. We develop a comparison argument to establish the uniqueness of the robust XVA process and of the corresponding super-replicating strategies in Section~\ref{sec:ODE}. We provide an explicit computation of margins under the proposed framework in Section~\ref{sec:margins}. We generalize the analysis to a portfolio of credit default swaps in Section~\ref{sec:creditport}.

\subsection{BSDE representation of XVA}\label{sec:dynXVA}
This section characterizes the $\XVA$ process given in Eq.~\eqref{eq:XVAdef} as the solution to a BSDE. We start analyzing the dynamics of the process $V_t(\gamma)$. Given a self financing strategy, the investor's wealth process in~\eqref{eq:vself} under the risk neutral measure $\Qxx$ follows the dynamics
\begin{align}\label{eq:vtlast}
    \nonumber dV_t(\gamma) &= \bigl(r_f \xi_t^f B_t^{r_f} + r_D \xi_t^1 P_t^1 + r_D \xi_t^I B_t^I + r_D \xi_t^C B_t^C - r_m \psi_t^m B_t^{r_m} + \gamma \eta^1   \bigr) \, dt  \\
    \nonumber & \phantom{==}- \xi_t^1 B_t^1   \, d\varpi_t^{1,\Qxx}  - \xi_t^I B_{t-}^I \, d\varpi_t^{I,\Qxx}  - \xi_t^C B_{t-}^C  \, d\varpi_t^{C,\Qxx}\\
    \nonumber &= \Bigl( \rfp \bigl(\xi_t^f B_t^{r_f}\bigr)^+ -\rfm \bigl(\xi_t^f B_t^{r_f}\bigr)^- + r_D \xi_t^1 B_t^1 + r_D \xi_t^I B_t^I + r_D \xi_t^C B_t^C \Bigr) \, dt \\
     & \phantom{==} + \Bigl(\rcp \bigl(\gamma M_t\bigr)^+ - \rcm \bigl(\gamma M_t\bigr)^-  + \gamma \eta^1  \Bigr) \, dt  -  \xi_t^1 B_{t-}^1 +  \, d\varpi_t^{1,\Qxx} - \xi_t^I B_{t-}^I   \, d\varpi_t^{I,\Qxx}  - \xi_t^C B_{t-}^C d\varpi_t^{C,\Qxx}.
\end{align}
Setting
\begin{equation}\label{eq:Zetas}
    Z^{1,{\gamma}}_t := -\xi_t^1 B_{t-}^1 ,\qquad Z^{I,{\gamma}}_t := -\xi_t^I B_{t-}^I,\qquad Z_t^{C,{\gamma}} := -\xi_t^C B_{t-}^C,
\end{equation}
and using Eq.~\eqref{eq:wealth}, we obtain that
\begin{equation}\label{eq:funding}
    \xi_t^f B_t^{r_f} = V_t(\gamma) - \xi_t^1 P^1_t - \xi_t^I P^I_t - \xi_t^C P^C_t - \gamma M_t.
\end{equation}
We may then rewrite the wealth dynamics as
\begin{align}\label{eq:vtlast2}
    \nonumber dV_t(\gamma) &= \Bigl(\rfp \bigl(V_t(\gamma) + Z_t^{1,{\gamma}} + Z_t^{I,{\gamma}} + Z_t^{C,{\gamma}} - \abs{\gamma} M_t \bigr)^+ -\rfm \bigl(V_t(\gamma) + Z_t^{1,^{\gamma}} + Z_t^{I,{\gamma}} + Z_t^{C,{\gamma}} -  \abs{\gamma} M_t \bigr)^-\\
    \nonumber & \phantom{=}  - r_D Z_t^{1,{\gamma}} - r_D Z_t^{I,{\gamma}} - r_D Z_t^{C,{\gamma}}  +  \rcp\abs{\gamma}   M_t^+ - \rcm\abs{\gamma}  M_t^-  + \gamma \eta^1 \Bigr)\,  dt \\
    & \phantom{=} + Z_t^{1,{\gamma}} \,d\varpi_t^{1,{ \Qxx}}  + Z_t^{I,{\gamma}} \, d\varpi_t^{I,{ \Qxx}}  + Z_t^{C,{\gamma}} \, d\varpi_t^{C,\Qxx}.
 \end{align}
To study the robust replicating strategy, we use the above dynamics to formulate the BSDE associated with the portfolio replicating
 the credit default swap. This is given by
  \begin{align}
  \nonumber -dV_t(\gamma) &= f\bigl(t,V_t(\gamma),Z_t^{1,\gamma},Z_t^{I,\gamma},Z_t^{C,\gamma},\gamma; M_t\bigr) \, dt - Z^{1,\gamma}_t\, d\varpi_t^{1,\Qxx}  - Z_t^{I,\gamma} \, d\varpi_t^{I,\Qxx}  - Z_t^{C,\gamma} \, d\varpi_t^{C,\Qxx},\\
  V_{\tau \wedge \tau^1}(\gamma) & =   \gamma L^1 \ind_{\tau^1 < \tau} + \theta_I(\gamma\hat{V}_\tau, \abs{\gamma}M_{\tau-})\ind_{\{\tau < \tau^1 \wedge \tau^C \wedge T\}} + \theta_C(\gamma\hat{V}_\tau, \abs{\gamma}M_{\tau -})\ind_{\{\tau < \tau^1 \wedge \tau^I \wedge T\}},
 \label{eq:BSDE-sell}
  \end{align}
where the driver $f \, : \, \Omega \times [0,T] \times \R^5$, $(\omega, t,v,z,z^I,z^C,\gamma) \mapsto f\bigl(t,v,z,z^I,z^C,\gamma; M\bigr)$ is given by
\begin{align}
    f\bigl(t,v,z^1,z^I,z^C,\gamma; M\bigr) &:= -\Bigl(\rfp \bigl(v+ z^1 + z^I + z^C  -  \abs{\gamma} M_t\bigr)^+ -\rfm \bigl(v + z^1 + z^I + z^C  -   \abs{\gamma} M_t\bigr)^- \nonumber\\
    & \phantom{==} - r_D z^1 - r_D z^I - r_D z^C + \rcp  \abs{\gamma} M_t^+ - \rcm \abs{\gamma} M_t^-  +\gamma \eta^1  \Bigr ).\label{eq:f+}
\end{align}
In the above expression, we highlight the dependence on the collateral process $M$ that is used to mitigate the default losses associated with the $\gamma$ units of the traded CDS contract. In the case the reference entity defaults before the investor or her counterparty, $\tau^1 <\tau$, the terminal condition is given by the loss term $-\gamma L^1$. This is consistent with the fact that, at this time, the value of the transaction from the investor's point of view corresponds with the third party valuation $\gamma\hat V_{\tau^1} = {\gamma}\hat{C}^1_{\tau^1} = {\gamma} L^1 \ind_{\{\tau^1 \leq T\}}$. By positive homogeneity of the driver $f$ with respect to $\gamma>0$, we will assume that
$\gamma=1$ throughout the paper and suppress it from the superscript. The case $\gamma=-1$ follows from symmetric arguments.

Next, we study the dynamics of the credit default swap price process $\hat{V}$, viewed from the valuation agent's perspective. Such a process satisfies a BSDE
that can be derived similarly to Eq.~\eqref{eq:BSDE-sell} (essentially ignoring the terms $Z^I, Z^C$ as well as the collateral terms, setting $\rfm = \rfp = r_D${, and normalizing $\gamma=1$}). This is given by
\begin{align}
- d\hat{V}_t = \bigl(-r_D \hat{V}_t  - \eta^1 \bigr)\, dt -\hat{Z}_t^1 \, d \varpi_t^{1,\Qxx}, \qquad \qquad
\hat{V}_{\tau^1 \wedge T} &= L^1 \ind_{\tau^1 < T}.\label{eq:BSDE-hatV}
\end{align}

This BSDE is well known to admit the unique solution $(\hat{V}_t,\hat{Z}_t^1)$, where $\hat{V}$ can be represented explicitly (see the Appendix) as
\begin{align}
	\hat{V}_t = \hat{C}^1_t &= -\Exx^{\Qxx}\biggl[ \int_t^{ T}  e^{-\int_t^u (h^{1,\Qxx}_s  +r_D) ds} \eta^1 \, du - \int_t^T L^1 h^{1,\Qxx}_u e^{-\int_t^u (h^{1,\Qxx}_s + r_D) ds} du \, \bigg\vert \, \mathcal{F}_t\biggr]\ind_{\{t \leq \tau^1\}}.
	\label{eq:creditswap}
\end{align}
%
We immediately obtain a BSDE for the $\XVA$ process given by
\begin{align}
	-d\XVA_t^{\buysell} & = \tilde{f}^{}\bigl(t,\XVA_t^{\buysell},\tilde{Z}_t^{1},\tilde{Z}_t^{I},\tilde{Z}_t^{C}; M\bigr) \, dt - \tilde{Z}_t^{1}\, d\varpi_t^{1,\Qxx}  - \tilde{Z}_t^{I} \, d\varpi_t^{I,\Qxx}  - \tilde{Z}_t^{C} \, d\varpi_t^{C,\Qxx},\nonumber\\
	\XVA_{\tau \wedge \tau^1}^{\buysell} &= \tilde{\theta}^{C}(\hat V_\tau, M_{\tau -}) \ind_{\{\tau <\tau^1 \wedge \tau^I \wedge T\}} + \tilde{\theta}^{I}(\hat V_\tau, M_{\tau -})\ind_{\{\tau <\tau^1 \wedge \tau^C \wedge T \}},
\label{eq:XVABSDE}
\end{align}
where
\begin{align}
	\tilde{Z}_t^{1} & := Z_t^{1} - \hat{Z}^1_t, \qquad \tilde{Z}_t^{I} := Z_t^{I}, \qquad \tilde{Z}_t^{C} := Z_t^{C}, \nonumber \\  \tilde{\theta}^{C}(\hat v, m)  & := L^C ( \hat{v} -m)^{-}, \qquad \tilde{\theta} ^{I}(\hat v, m)  := - L^I (\hat{v} -m)^{+},
	\label{eq:hats}
\end{align}
and	
\begin{align}
&\!\!\!\!\!\tilde{f}\bigl(t,xva,\tilde{z}^1,\tilde{z}^I,\tilde{z}^C; M\bigr)  \nonumber\\
 & \qquad\qquad\qquad:= -\Bigl(\rfp \bigl(xva+ \tilde{z}^1 + \tilde{z}^I + \tilde{z}^C +  L^1- M_t \bigr)^+ -\rfm \bigl(xva + \tilde{z}^1 + \tilde{z}^I + \tilde{z}^C +  L^1- M_t \bigr)^- \nonumber\\
&\qquad\qquad\qquad  - r_D \tilde{z}^1 - r_D \tilde{z}^I - r_D \tilde{z}^C + \rcp \bigl(M_t\bigr)^+ - \rcm\bigl(M_t\bigr)^- -r_D L^1\Bigr).
\end{align}
Above, we have used the fact that $\hat{Z}_t^1 = L^1 - \hat{V}_{t-} = L^1 - \hat{V}_t $ by stochastic continuity and thus $Z_t^1 = \tilde{Z}_t^1 + \hat{Z}_t^1 = \tilde{Z}_t^1 + L^1 - \hat{V}_t$.

We can now apply the reduction technique developed by \cite{CrepeyRed} to find a continuous ordinary differential equation describing the $\XVA$ prior to the investor and her counterparty's default.
\begin{proposition}[] \label{thm:reduction}
The BSDE
\begin{align}\label{eq:reduced}
	-d\check{U}_t^{\buysell}  = \check{g}\bigl(t,\check{U}_t^{\buysell}; \hat{V}, M\bigr) \, dt, \qquad \qquad
	\check{U}_{T}^{\buysell} = 0,
\end{align}
in the (trivial) filtration $\mathbb{F}$, with driver
\begin{align}
	\check{g}\bigl(t,\check{u}; \hat{V},M\bigr) &= h^{I,\Qxx}\bigl(\tilde{\theta} ^{I}(\hat{V}_t, M_{t-})-\check{u}\bigr) + h^{C,\Qxx}\bigl(\tilde{\theta} ^{C}(\hat{V}_t, M_{t-})-\check{u}\bigr) - h^{1,\Qxx} \check{u} \nonumber\\
	& \phantom{==}+ \tilde{f}\bigl(t, \check{u}, -\check{u}, \tilde{\theta}^{I}( \hat{V}_t, M_{t-})-\check{u}, \tilde{\theta}^{C}( \hat{V}_t, M_{t-}) -\check{u}; M\bigr)\label{eq:g+}
\end{align}
admits a unique solution $\check{U}^{\buysell}$, that is related to the unique solution $\bigl(\XVA^{\buysell}, \tilde{Z}^{1}, \tilde{Z}^{I}, \tilde{Z}^{C}\bigr)$ of the BSDE in Eq.~\eqref{eq:XVABSDE} as follows. On the one hand
\begin{equation}\label{eq:reduced_identity1}
	\check{U}_t^{\buysell}  :=  \XVA_{t\wedge (\tau \wedge \tau_1)-}^{\buysell}
\end{equation}
is a solution to the ODE (reduced BSDE) in Eq.~\eqref{eq:reduced}, and on the other hand a solution to the full $\XVA$ BSDE~\eqref{eq:XVABSDE} is given by
\begin{align}
	 \XVA_t^{\buysell}  &= \check{U}_t^{\buysell} \ind_{\{t<\tau \wedge \tau^1 \}} + \Bigl( \tilde{\theta}^{C}(\hat{V}_{\tau^C}, M_{\tau^C-}) \ind_{\{\tau^C < \tau^1 \wedge \tau^I\wedge T\} }  + \tilde{\theta} _{I}(\hat{V}_{\tau^I}, M_{\tau^I -}) \ind_{\{\tau^I < \tau^1 \wedge \tau^C\wedge T\}} \Bigr) \ind_{\{t \geq \tau \wedge \tau^1\}}, \label{eq:reduced_identity2}\\
	\nonumber \tilde{Z}_t^1 & =  -\check{U}_t^{\buysell} \ind_{\{t \leq \tau \wedge \tau^1\}}, \quad \tilde{Z}_t^{I} = \Bigl(\tilde{\theta}^I(\hat{V}_t, M_{t-}) - \check{U}_t^{\buysell}\Bigr)\ind_{\{t \leq \tau \wedge \tau^1\}}, \quad
	\tilde{Z}_t^{C} = \Bigl(\tilde{\theta} ^C(\hat{V}_t, M_{t-}) - \check{U}_t^{\buysell}\Bigr)\ind_{\{t \leq \tau \wedge \tau^1\}}.
\end{align}
\end{proposition}

The uniqueness of the solution to the original BSDE for $V$ as well as to their projected versions in the $\mathbb{F}$-filtration follows from the definition of $\XVA$.

\begin{corollary}\label{cor:BSDE-V-red}
The BSDE~\eqref{eq:BSDE-sell} admits a unique solution. This solution is related to the unique solution $\bar{U}$ of the ODE
\begin{align}
	-d\bar{U}_t  = g\bigl(t,\bar{U}_t; \hat{V}, M\bigr) \, dt, \qquad \qquad
	\bar{U}_T = 0,
\label{eq:reduced-V}
\end{align}
in the filtration $\mathbb{F}$ with
\begin{align}
	g\bigl(t,\bar{u}; \hat{V}, M\bigr) &= h^{I,\Qxx}\bigl(\theta^{I}(\hat{V}_t, M_{t-})-\bar{u}\bigr) + h^{C,\Qxx}\bigl(\theta^{C}(\hat{V}_t, M_{t-})-\bar{u}\bigr)-
	h^{1,\Qxx} \bar{u}\\
	& \phantom{==}
	+f\bigl(t, \bar{u}, L^1 -\bar{u}, \theta^{I}( \hat{V}_t, M_{t-})-\bar{u}, \theta^{C}( \hat{V}_t, M_{t-}) -\bar{u};M\bigr)
\end{align}
via the following relations. On the one hand
\[
	\bar{U}_t  := V_{t\wedge (\tau \wedge \tau^1)-}
\]
is a solution to the reduced BSDE~\eqref{eq:reduced-V}, while on the other hand a solution to the full BSDE~\eqref{eq:BSDE-sell} is given by
\begin{align}
	\nonumber V_t  &:= \bar{U}_t \ind_{\{t<\tau \wedge \tau^1\}} + \Bigl(  L^1 \ind_{\tau^1 < \tau} + \theta^{C}(\hat{V}_{\tau^C}, M_{\tau^C -}) \ind_{\{\tau^C < \tau^1 \wedge \tau^I\wedge T\} }  + \theta^{I}(\hat{V}_{\tau^I}, M_{\tau^I -}) \ind_{\{\tau^I < \tau^1 \wedge \tau^C\wedge T\}} \Bigr) \ind_{\{t \geq \tau \wedge \tau^1\}}, \\
	\nonumber Z_t^1  &:= L^1 -\bar{U}_t \ind_{\{t<\tau \wedge \tau^1 \}}, \quad Z_t^{I} :=  \Bigl(\theta^I(\hat{V}_t, M_{t-}) - \bar{U}_t\Bigr)\ind_{\{t\leq \tau \wedge \tau^1 \}}, \quad Z_t^{C} :=  \Bigl(\theta^C(\hat{V}_t, M_{t-}) - \bar{U}_t\Bigr)\ind_{\{t\leq\tau \wedge \tau^1 \}}.
\end{align}
\end{corollary}

Using the above representation, we can provide explicit representations for the replication strategies of the XVA. We will use the tilde symbol to denote these replicating strategies (e.g., $\tilde{\xi}^1, \tilde{\xi}^{I}, \tilde{\xi}^{C}$ denote, respectively, the number of shares of the defaultable accounts associated with the reference entity, trader and her counterparty) so to distinguish them from the strategies used to replicate the CDS price process. Using the martingale representation theorem for the probability space $(\Omega, \mathcal{F}, \mathbb{F}, \Qxx)$ and account price dynamics, we obtain that
\begin{equation}\label{eq:xi1}
\tilde{\xi}^{1}_t = - \frac{\tilde Z_t^{1}}{B_{t-}^1} \ind_{\{t < \tau \wedge \tau^1\}} = \frac{\check{U}_t}{B_{t-}^1} \ind_{\{t < \tau \wedge \tau^1\}}.
\end{equation}
Invoking Theorem \ref{thm:reduction} along with equations~\eqref{eq:Zetas} and~\eqref{eq:hats} we conclude that
\begin{align}
\tilde{\xi}_{t}^{I} &= - \frac{\tilde Z_t^{I}}{B_{t-}^I}  \ind_{\{t \leq \tau \wedge \tau^1 \}} = \frac{L^I  (\hat{V}_t - M_{t-})^{+} + \check{U}_t}{B_{t-}^I} \ind_{\{t \leq \tau \wedge \tau^1 \}}, \label{eq:xii}\\
\tilde{\xi}_{t}^{C} &= - \frac{\tilde Z_t^{C}}{B_{t-}^C}  \ind_{\{t \leq \tau \wedge \tau^1 \}} = \frac{ -L^C (\hat{V}_t - M_{t-})^{-}+ \check{U}_t^{\buysell}}{B_{t-}^C}\ind_{\{t \leq \tau \wedge \tau^1\}}, \label{eq:xic}
\end{align}
and from equations~\eqref{eq:collrel} and~\eqref{eq:rulecoll} it follows that
\begin{equation}\label{eq:psim}
\tilde\psi_t^{m}=  - \frac{M_{t-}}{B_t^{r_m}}  \ind_{\{\tau \wedge \tau^1 > t\}}.
\end{equation}
Finally, using Eq.~\eqref{eq:funding} and the identity $V_t=XVA_t^{\buysell} + \hat{V}_t$, we obtain
\begin{align}
\tilde{\xi}_t^{f} &=  \frac{V_t - \hat V_t- \tilde{\xi}_t^1 B^1_{t-} - \tilde{\xi}_t^I B^I_{t-} - \tilde{\xi}_t^C B^C_{t-}  - M_{t-}}{B_t^{r_f}}  \ind_{\{\tau \wedge \tau^1 > t\}}  \nonumber\\
&=\frac{ - 2 \check{U}_t + L^C (\hat{V}_t - M_{t-})^{-}- L^I  (\hat{V}_t -M_{t-})^{+} - M_{t-}}{B_t^{r_f}}  \ind_{\{\tau \wedge \tau^1 > t\}}, \label{eq:xif}
\end{align}
where in the last equality above, we have used the definition of $\XVA$ given in Eq.~\eqref{eq:XVAdef} together with the identity $\XVA_t = \check{U}_t$ on $\{\tau \wedge \tau^1> t\}$.

Note that the replicating strategies are specified only in terms of account prices and are thus known to the investor at time $t$. However, they neither give information on the value of the XVA process nor on the evolution of the replicating strategy, because the default intensity process $h^{C,\Qx}$ is unknown to the investor.

\subsection{Comparison Pricing and Super-replicating Strategies}\label{sec:ODE}
{This section develops a comparison principle for the reduced BSDE \eqref{eq:reduced} solved by the XVA process. We subsequently use this result to construct a super-replicating strategy for the XVA.}

The BSDE given in Eq.~\eqref{eq:reduced} is effectively an ODE. To maintain consistency with the theory of ODEs, we switch the direction of time by defining $\hat v_t := \hat V_{T-t}$ and $m_t: = M_{T-t}.$ It follows from Eq.~\eqref{eq:creditswap} that $\hat v$ is bounded, i.e., $\abs{\hat v}\le M_0$ for some constant $M_0$. Similarly, set $\check u_t = \check U_{T-t}$. {Applying the reduction technique of \cite{CrepeyRed} to Eq.~\eqref{eq:BSDE-hatV}, similarly to how it was done above in Proposition \ref{thm:reduction}, we get
\begin{align}
\partial_t \hat v =-(\eta^1 - h^{1,\Qxx} L^1) - (h^{1,\Qxx}+r_D) \hat v,\label{eq:hatV-ODE}\qquad \qquad
\hat v_0 =0.
\end{align}
}
We may then rewrite Eq.~\eqref{eq:reduced} as
\begin{align}
\partial_t \check u = \check{g}(t,\check u; \hat v, m),\qquad \qquad
\check u_0 =0.
\label{eq:creditswap1}
\end{align}
The functions  $h^{1,\Qxx}_t, h^{I,\Qxx}_t$, $h^{C,\Qxx}_t$, $t\in[0,T],$ are all piecewise (deterministic) continuous. The following theorem provides an existence and uniqueness result.

\begin{proposition}\label{thm:existence}
There exists a unique (piecewise) classical solution to the {system of} ODEs \eqref{eq:hatV-ODE}--\eqref{eq:creditswap1}.
\end{proposition}

The following comparison principle, whose proof is reported in the Appendix, will be used to find the super-replicating XVA.
%
%

\begin{theorem}[Comparison Theorem]
\label{thm:comp}
Assume that there exists $\overline \mu^C \ge \underline \mu^C >r_D$ such that
$\overline \mu^C \ge  \mu^{C,\Qxx} \ge \underline \mu^C$ and let $\check u$ be the solution of ODE~\eqref{eq:creditswap1}. Let
\begin{align}
{(\mu^C)^*}(\hat{v}, m,  \check u) &= {\overline \mu^C} \ind_{\{ \tilde{\theta} ^{C}(\hat{v},m)-\check u\ge 0\}} + {\underline \mu^C} \ind_{\{ \tilde{\theta} ^{C}(\hat{v},m)-\check u \le 0\}},\\
{{(\mu^C)}_*}(\hat{v}, m,  \check u) &= {\underline \mu^C} \ind_{\{ \tilde{\theta} ^{C}(\hat{v},m)-\check u\ge 0\}} + {\overline \mu^C} \ind_{\{ \tilde{\theta} ^{C}(\hat{v},m)-\check u \le 0\}},
\end{align}
and define the drivers ${g}^{*}$ and ${g}_{*}$ by plugging the default intensities $(\mu^{C,\Qxx})^{*}$ and $(\mu^{C,\Qxx})_*$ into the expression of $\check g$ given by~\eqref{eq:g+}, i.e.,
\begin{align}
{g}^{*}\bigl(t,\check u; \hat{v},m\bigr) &= h^{I,\Qxx}\bigl(\tilde{\theta} ^{I}(\hat{v}_t,m_t)-\check{u}_t\bigr) + {({ (\mu^C)^*}(\hat v_t,m_t,\check u_t) - r_D)} \bigl(\tilde{\theta} ^{C}(\hat{v}_t,m_t)-\check{u}_t\bigr) - h^{1,\Qxx} \check{u}_t\nonumber \\
&+ \tilde{f}\bigl(t, \check{u}_t, -\check{u}_t, \tilde{\theta}^{I}( \hat{v}_t,m_t)-\check{u}_t, \tilde{\theta}^{C}( \hat{v}_t,m_t) -\check{u}_t;\hat{v},m\bigr),\nonumber\\
{g}_{*}\bigl(t,\check u; \hat{v},m\bigr) &= h^{I,\Qxx}\bigl(\tilde{\theta} ^{I}(\hat{v}_t,m_t)-\check{u}_t\bigr) +{({ (\mu^C)_*}(\hat v_t,m_t,\check u_t) - r_D)} \bigl(\tilde{\theta} ^{C}(\hat{v}_t,m_t)-\check{u}_t\bigr) - h^{1,\Qxx} \check{u}_t \nonumber\\
&+ \tilde{f}\bigl(t, \check{u}_t, -\check{u}_t, \tilde{\theta}^{I}( \hat{v}_t,m_t)-\check{u}_t, \tilde{\theta}^{C}( \hat{v}_t,m_t) -\check{u}_t;\hat{v},m\bigr).\nonumber
\end{align}
Let $\check{u}^{*}$ and $\check{u}_{*}$ be the solutions to ODE~\eqref{eq:creditswap1} where $\check{g}$ is replaced by ${g}^{*}$ and ${g}_{*}$ respectively, i.e.,
\begin{align}
 \partial_t \check{u}^{*} &= g^{*}(t,\check{u}^{*}; \hat v,m), \qquad \check{u}^{*}_0 =0, \nonumber \\
 \partial_t \check{u}_{*} &= g_{*}(t,\check{u}_{*}; \hat v,m), \qquad {\check{u}_{*}}_0 =0.
\label{eq:creditswap2}
\end{align}
Then $\check{u}_{*} \le \check u \le \check{u}^{*}.$
\end{theorem}

{The valuation process calculated based on the extremes of the uncertainty interval $\bigl(\mu^C\bigr)^*$ and $\bigl(\mu^C\bigr)_*$ are denoted, by $V^{(\mu^C)^*}$ and $V^{(\mu^C)_*}$ respectively.}

The ODEs~\eqref{eq:creditswap2} may be understood as the credit risk counterparts of the \textit{Black-Scholes-Barenblatt} PDEs for the uncertain volatility model; see \cite{ALP}. The main difference between our study and theirs is that, in their paper, the uncertainty comes from the volatility which appears as a second order term in the differential operator. Hence, the indicator function specifying the value of volatility to use in the pricing formula depends on the second order derivative of the option price with respect to the underlying, i.e., the Gamma of the option. In our setting, the indicator function specifying the value of counterparty's {account} rate to use depends on the relation between the current value of the XVA replication and the close-out value. The value of the replicating trade jumps to the close-out value when the counterparty defaults. If the size of this jump is positive, i.e., the close-out value of the transaction is higher, then the trader needs to be short  the counterparty's {account} to replicate this jump-to-default risk. As the trader wants to consider the worst possible scenario for her trade, she would choose the largest value of the {counterparty account rate $\overline{\mu}^C$} because this yields the lowest rate of return on her short position. Vice-versa, if the jump is negative, the trader {needs to be long} the counterparty defaultable account. Consequently, the trader would use the smaller counterparty's account rate $\underline{\mu}^C$ to deal with the worst possible replication scenario.

Our objective is to provide a tight upper bound for the XVA price process, because this implies a tight super-replicating price. We connect such a super-replicating price to the $\rXVA$ defined in Eq.~\eqref{eq:rXVAdef}. Define the process
$\check{U}_t^{*} := \check{u}^{*}_{T-t}$. {The following theorem shows that the $\rXVA$ coincides with the super-replicating price, and additionally specifies the super-replicating strategy. The latter is obtained by taking the strategy given in \eqref{eq:xi1}--\eqref{eq:xif} and using the super-replicating price $\check{U}^{*}$ in place of $\check{U}$. }

\begin{theorem}\label{thm:robust}
	The robust $\XVA$ admits the explicit representation given by
\begin{equation}\label{eq:rXVA}
\rXVA_t  = \check{U}_t^{*} \ind_{\{t<\tau \wedge \tau^1 \}} + \Bigl( \tilde{\theta} ^{C}(\hat{V}_{\tau^C}, M_{\tau^C-}) \ind_{\{\tau^C < \tau^1 \wedge \tau^I\wedge T\} }  + \tilde{\theta}^{I}(\hat{V}_{\tau^I}, M_{\tau^I -}) \ind_{\{\tau^I < \tau^1 \wedge \tau^C\wedge T\}} \Bigr) \ind_{\{t \geq \tau \wedge \tau^1\}},
\end{equation}

and the corresponding super-replicating strategies for $\rXVA$ are given by
\begin{align}
\xi^{1,*}_t & = \frac{\check{U}_t^*}{B_{t-}^1} \ind_{\{t < \tau \wedge \tau^1\}}, \nonumber \\
\xi_{t}^{I,*} &= \frac{L^I (\hat{V}_t - M_{t-})^{+} + \check{U}_t^*}{B_{t-}^I} \ind_{\{t \leq \tau \wedge \tau^1 \}}, \nonumber\\
\xi_{t}^{C,*} &= \frac{ -L^C (\hat{V}_t - M_{t-})^{-}+ \check{U}_t^{*}}{B_{t-}^C}\ind_{\{t \leq \tau \wedge \tau^1\}}, \nonumber\\
\psi_t^{m,*}&=  - \frac{M_{t-}}{B_t^{r_m}}  \ind_{\{t < \tau \wedge \tau^1 \}}, \nonumber\\
\xi_t^{f,*} &= \frac{- 2 \check{U}_t^* + L^C (\hat{V}_t - M_{t-})^{-}- L^I  (\hat{V}_t -M_{t-})^{+} - M_{t-}}{B_t^{r_f}}  \ind_{\{ t < \tau \wedge \tau^1 \}}.
\label{eq:rob_strat1}
\end{align}
\end{theorem}

We notice that if we use the robust super-replicating strategies given in~\eqref{eq:rob_strat1} and start with an initial capital $\rXVA_0$, then there will be no tracking error in the sense of \cite{EJP}. In other words, the error committed for implementing the robust strategy $\bigl(\xi_t^{1,}, \xi_t^{I,}, \xi_t^{C,}, \xi_t^{f,},  \psi_t^{m,*}\bigr)$ in the real market ({where the return rate of the counterparty {account} is $\mu^C$}) instead of the robust market model (where the return rate of the counterparty  {account} is $(\mu^C)^*$) is zero. This may be understood as follows: Eq.~\eqref{eq:super-hedge-value} shows that the value of the super-replicating portfolio is always
$\rXVA$. However, until the earliest among the default time of the counterparty, investor, or maturity of the CDS contract, whichever comes first, the super-replicating portfolio keeps generating profits because the change in the value of the super-replicating portfolio is greater than the change in the value of the $\check{U}^{*}$, as shown in \eqref{eq:superhedge-ineq}. In other words, during a time interval $dt$, the investor pockets an extra cash
${\left ( (\mu^C)^{*}(\hat V_t,M_t,\check U^{*}_t) - \mu^C \right)}  \bigl(\tilde{\theta} ^{C}(\hat{V}_t, M_t)-\check{U}^{*}_t\bigr)dt$ at any time prior to the end of the replication strategy.

{The robust strategies depend only on the XVA price process and the {account} prices, and are independent of the default intensity $h^{C,\Qxx}$, the value of which is unknown to the investor.}

\begin{figure}[ht!]
	\begin{center}
	\includegraphics[scale=0.6]{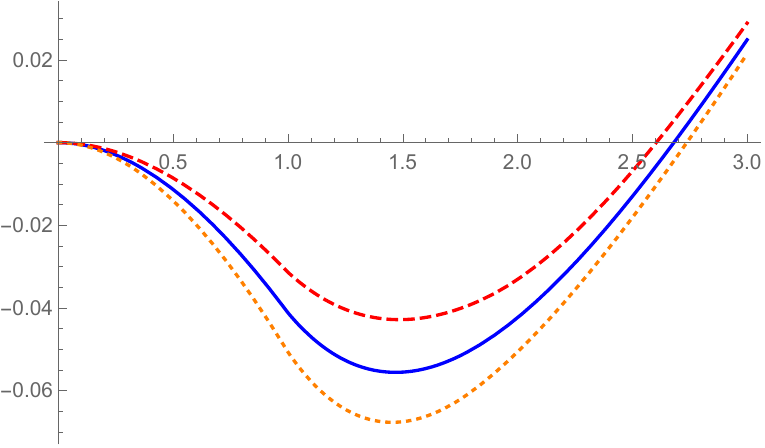}
	\hspace{20pt}
	\includegraphics[scale=0.6]{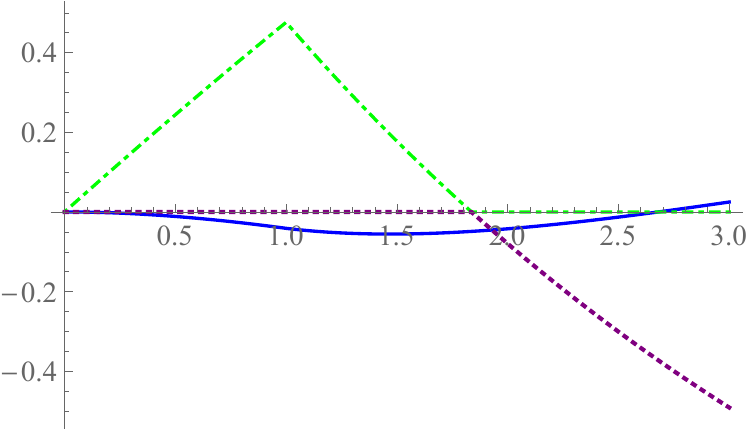}
	\end{center}
	\caption{We use the following  parameters: $r_f^{\pm} = r_D=0.001$, $\alpha=\beta=0$,  $T=3, L^I=L^C =0.5$, $\eta^1 = 2, h^1_t = 0.1 \ind_{\{ 0\le t <1 \}} +  0.3\ind_{\{ 1\le t <T \}} , L^1= 10, {\mu^I=0.2001},  {\overline{\mu}^C=0.2501}, {\underline{\mu}^C =0.1501}, {\mu^C = 0.2001}.$  Left panel: Plot of $\check u$ (solid), $\check u^{*}$ (dashed) and $\check u_{*}$ (dotted) as a function of time.  Right panel: Plot of $\tilde\theta^C(\hat v,0)$ (dash-dotted), $-\tilde\theta^I(\hat v,0)$ (dotted) and $\check u$ (solid) as a function of time. In the left panel, the default intensity at which we switch between the sub-and super-solutions is the crossing point of the dashed and dotted lines with the $x$-axis, that occurs at approximately $t=2.67$. In the right panel, the {third party valuation $\hat{v}$} becomes positive at approximately {$t=2.67$}.}
	\label{fig:comp}
\end{figure}
{In the case of zero margins, it follows directly from Eq.~\eqref{eq:hats} that the third party valuation $\hat v = \hat{v}^+ - \hat{v}^-$ may be expressed in terms of the closeout value, and given by
$-\frac{\tilde\theta^C(\hat v,0)}{L^C} - \frac{\tilde\theta^I(\hat v,0)}{L^I}$.} Hence, we deduce from
the right panel of Figure \ref{fig:comp} that the third party valuation is negative prior to $t=1.83$, and positive for $t > 1.83$. Figure \ref{fig:comp} also shows that the super-replicating strategy is non-trivial in the sense that it is not monotone in the default intensity. As it can be seen from the right panel of Figure \ref{fig:comp}, the quantity $\tilde\theta^C(\hat v,0) - \check u$ is zero at $t_0\approx 2.67$, non-negative for $t<t_0$, and strictly
negative for $t>t_0$. This implies that {$(\mu^C)^{*} = \overline \mu^C$ prior to time $t_0$ and while $\tilde\theta^C(\hat v,0) - \check u\ge0$, whereas after time $t_0$, $(\mu^C)^{*} = \underline \mu^C$, because we then have $\tilde\theta^C(\hat v,0) - \check u\le0$.}
In other words, prior to $t_0$ the trader will use the largest value of the {account} rate $\overline \mu^C$ for her super-replicating portfolio because the jump of the super-replicating portfolio to the close-out value when the counterparty defaults, given by $\tilde\theta^C(\hat v,0) - \check u$, is positive. After time $t_0$, the trader will choose the smallest value $\underline \mu^C$ of the account rate because this jump would be negative. This is directly visible from the right panel of Figure~\ref{fig:comp}, because the dash-dotted line dominates the solid one until time $t_0$, and after $t_0$ the situation is reversed. This analysis highlights a fundamental difference with respect to standard credit risk settings, that often ignore collateralization and close-out terms, or models for XVA in which collateralization and close-out value depend on the trader's valuation process $V$
itself as in \cite{NiRutFS}. In these cases, the price of the derivative is monotone in the default intensity, while in our setting the value of the super-replicating portfolio does not necessarily have this monotonicity property. This is due to the fact that the collateralization and closeout process are
exogenous, i.e., they depend on the external valuation $\hat{V}$ of the third party, rather than on the value $V$ of the super-replicating portfolio.


\subsection{Computation of Margins} \label{sec:margins}
We develop an explicit expression for the initial margins when the two parties trade $\gamma$ units of a single name credit default swap contract. Initial margins are determined using the value-at-risk criterion, and need to be computed under the physical measure $\Px$ as opposed to the valuation measure $\Qx$. By the definition of $\VaR$, {on the set $\{\tau\wedge\tau^1>t\}$ we have}
\begin{align}
IM_t(\gamma)&=\beta \VaR_{q}\Bigl (\gamma\hat V_{(t+\delta) \wedge T}-\gamma\hat V_t\, \vert \, \tau^{1} > t\Bigr)^{+} \nonumber \\
%
&= \beta\inf \Bigl\{ K \in \R_{>0} \colon \Px\bigl[\gamma \hat  V_{(t+\delta)\wedge T} + \gamma\hat V_t > -K\bigl\vert \tau^{1} >t \bigr] \ge 1-q \Bigr\}.
\label{eq:margcds}
\end{align}
Thus, differently from the variation margin $VM(\gamma) = \alpha \gamma \hat V_t\ind_{\{\tau \wedge \tau^{(N)} > t\}}$ that is linear in $\gamma$, the initial margin $IM(\gamma)$ {is only positively homogeneous in $\gamma$}. We will therefore distinguish the cases $\gamma = 1$ and $\gamma = -1$. Note first that
\begin{align}
\gamma\left(\hat V_{(t+\delta)\wedge T} - \hat V_t\right) =-\gamma \left\{ \begin{array}{ll}
\eta^1 \left((t+\delta)\wedge T - t\right) &\text{ if } \tau^{1}  \ge (t+\delta)\wedge T,\\
-L^1 + \eta^1 \left(\tau^{1} - t\right) &\text{otherwise}.
\end{array}
\right.
\end{align}
The case $\gamma=1$ is less frequently observed in practice. We typically expect $L^1 > \eta^1T$, as {a protection buyer is unlikely to pay more than what he would receive in the event of a default (notice that $\eta^1 T$ is the maximum payment the buyer would make).} In this case, the exposure of the protection seller to the protection buyer would be negative, resulting in negative $\VaR$. 
In the case $\gamma=-1$, we obtain
\begin{align}
 &\Px\big[- \hat V_{(t+\delta)\wedge T} + \hat V_t > -K\bigl\vert \tau^{1} >t \big] = \frac{\Px\big[ -\hat V_{(t+\delta)\wedge T\wedge\tau^{1}} + \hat V_{t\wedge \tau^{1}} > -K \big]}{\Px\left[\tau^{1} >t \right]}\nonumber\\
 &=\frac{ \int_t^{(t+\delta)\wedge T}(K-L^1 + \eta^1(u-t))^{+} e^{-h^{1,\Px}_u} du}{\int_t^\infty e^{-h^{1,\Px}_u} du}.\label{eq:VaR1}
\end{align}
Because the right hand side of Eq.~\eqref{eq:VaR1} is continuous and increasing in $K$, the inequality {specifying the probability event in}~\eqref{eq:margcds} that characterizes the initial margins becomes an equality, and thus the $\VaR$ can be numerically evaluated.

\begin{example}
Assume constant default intensities and $t< T-\delta$. To calculate the initial margin $IM_t^{\gamma}$ for $\gamma=-1$, we note that
\begin{align}
 \frac{ \int_t^{(t+\delta)\wedge T}(K-L^1 + \eta^1(u-t))^{+} e^{-h^{1,\Px}_u} du}{\int_t^\infty e^{-h^{1,\Px}_u} du}=1-e^{-h^{1,\Px} \bigl(\frac{L^1-K}{\eta^1} \bigr)\wedge\delta } .
\end{align}
Therefore, the value of $K$ solving the above equation, i.e., the initial margin, is explicitly given by
\begin{align}
IM_t(-1)=\left\{\begin{array}{ll}
 \beta \Bigl(L^1 + \eta^1 \frac{\log q}{h^{1,\Px}}\Bigr)  & \text{ if } q> e^{-h^{1,\Px}\delta},\\
 0 &  \text{otherwise}.
 \end{array}\label{eq:IM-ex}
 \right.
\end{align}
{The initial margin formula \eqref{eq:IM-ex} has a direct economic interpretation. First, we notice that the term multiplying the spread $\eta^1$ is negative, because the value-at-risk level $q$ is between $0$ and $1$ and hence $\log q$ is negative. Thus, when the initial margin is nonzero, it is affine both in the loss rate and the CDS spread, increasing in the loss rate and decreasing in the CDS spread. This is intuitive: the protection seller increases the margin requirement if he has to make a larger payment at the credit event, and decreases the requirement if the spread premium received from the protection buyer is higher. Moreover, the required margin is increasing in the default intensity of the reference entity and, in the limiting case of an infinite default intensity,  it converges to the product $L^1 \beta$. This reflects economic intuition: as the credit event becomes more likely to occur, the protection seller asks the buyer to pay exactly the amount he would receive at the credit event. Finally, the value of initial margins is decreasing in the value-at-risk level and
 linearly increasing in the collateralization rate.}
\end{example}

\subsection{Credit Swap Portfolios}\label{sec:creditport}
In this section, we generalize the analysis conducted in the previous sections to a portfolio of single name credit default swaps, each referencing a different entity. To capture direct default contagion, we let the default intensities of surviving entities depend on past defaults. Throughout the section, we use the superscript $\mathcal{J}$, $\mathcal{J}\subset \{1, ..., N\}$, to denote the set of defaulted entities. 
For instance, ${V}^{{\cal J}}$ denotes the replicating process of the CDS portfolio where the defaulted reference entities are exactly those in the set {$\mathcal{J}$}. We denote by $\tau^{{\cal J}}$ the last default time of a reference entity in ${\cal J}$ (i.e., $\tau^{{\cal J}} = \max_{j \in \cal{J}} \tau^j$, {assuming $\max_{j \in \cal{J}} \tau^j < \tau^i, ~i\not\in \cal{J}$}), and for $i\not\in \cal{J}$ we use $\tau^{1,{\cal J}}$ to denote the default time of the $i$-th reference entity in the economic scenario where all reference entities in ${\cal J}$ have already defaulted.

First, we study the dynamics of the third party valuation process $\hat{V}$. Note that if all entities have defaulted, then $\hat V^{\{1, ...,N\}}=0.$  The case when all entities except for $i$ have already defaulted, that is ${\cal J} = \{1, ..., N\}\backslash \{i\}$ (in this case $\tau^{1,{\cal J}} = \tau^{\{1, .., N\} } $), is analogous to the case of a single CDS contract, whose price process has been given in Eq.~\eqref{eq:BSDE-hatV}. Hence
\begin{align}
 -d\hat{V}^{{\cal J}}_t  = \Bigl(-r_D\hat{V}^{{\cal J}}_t - \eta^{i}  \Bigr) \, dt - \hat{Z}^{1,{\cal J}}_t \, \varpi_t^{i,\cal{J},\Qxx}, \qquad \qquad
 \hat{V}^{{\cal J}}_{\tau^{i,{\cal J}} \wedge T} =  L^{i} \ind_{\{\tau^{i,{\cal J}} < T\}}.
\label{eq:BSDE-hatV-multi1}
\end{align}
Next, we provide an inductive relation which relates the investor's wealth price process in the state where all entities in {$\mathcal{J}$} have defaulted, to that in the state where the reference entity $i\not \in \cal{J}$ additionally defaults. The base case $|{\cal J}|=N-1$ has been given in~\eqref{eq:BSDE-hatV-multi1}. For the case $|{\cal J}| < N-1$, we obtain
\begin{align}
 -d\hat{V}^{{\cal J}}_t & = -\Bigl(r_D\hat{V}^{{\cal J}}_t -  \sum_{k\notin \cal{J}} \eta^{k}  \Bigr)\, dt - \sum_{k \notin \cal{J}}  \hat{Z}^{k,{\cal J}}_t \, \varpi_t^{k,{\cal J},\Qxx},  \nonumber \\
 \hat{V}^{{\cal J}}_{T \wedge \min_{j\notin \cal{J}} \tau^{j,{\cal J}} } & = \sum_{k\notin \cal{J}}\Bigl( L^{k} + \hat V_{\tau^{k,{\cal J} }}^{\{k\}\cup \cal{J}} \Bigr) \ind_{\{\tau^{k,{\cal J}} =\min_{j\notin {\cal J}} \tau^{j,{\cal J}} \}}\ind_{\{\tau^{k,{\cal J}} <T  \}}.\label{eq:BSDE-hatV-multi2}
\end{align}
The price process and the replicating strategy are then obtained by considering all possible subsets of defaulted entities, leading to
\begin{align}
\hat{V}_t & = \sum_{{{\cal J}} \in 2^{\{1, ..., N\}}} \hat{V}^{{\cal J}}_t \ind_{\{\tau^{{\cal J}}  \wedge \tau^C \wedge \tau^I \wedge T < t \le\min_{k \notin \cal{J}} \tau^{k,{\cal J}} \wedge \tau^C \wedge \tau^I \wedge T\}}, \\
\hat{Z}_t^{i} &= \sum_{{{\cal J}} \in 2^{\{1, ..., N\}\backslash \{i\} }} \hat{Z}^{1,{\cal J}}_t\ind_{\{\tau^{{\cal J}}  \wedge \tau^C \wedge \tau^I \wedge T < t \le\min_{k \notin \cal{J}} \tau^{k,{\cal J}} \wedge \tau^C \wedge \tau^I \wedge T\}}.		
\end{align}
The collateral process $M$ is still given by Eq.~\eqref{eq:rulecoll} along with the above  expression for $\hat V$.

Next, we use a similar inductive argument to define the $V:=(V_t)_{t\geq 0}$ process. Clearly, $V^{{\{1, ...,N\}}}=0.$ Consider now the state when all but entity $i$ have defaulted, that is ${\cal J} = \{1, ..., N\}\backslash \{i\}$ (and $\tau^{1,{\cal J}} = \tau^{\{1, .., N\} } $). The corresponding expression to~\eqref{eq:BSDE-sell} in the multi-name case is then given by
\begin{align}
-dV_t^{{\cal J}} &= f\Bigl(t,V_t^{{\cal J}},Z_t^{1,{\cal J}},Z_t^{I,{\cal J}},Z_t^{C,{\cal J}}; M^{{\cal J}}, {\cal J} \Bigr) \, dt - Z^{1,{\cal J}}_t\, d\varpi_t^{i,\cal{J},\Qxx}  - Z_t^{I,{\cal J}} \, d\varpi_t^{I,\Qxx}  - Z_t^{C,{\cal J}} \, d\varpi_t^{C,\Qxx},\label{eq:BSDE-sell-multi}\\
V^{{\cal J}}_{\tau \wedge \tau^{i,\cal{J} } } & = L^{i} \ind_{\{\tau^{i,\cal{J} } < \tau \wedge T\}} +  \theta^I(\hat{V}^{{\cal J}}_\tau, M^{{\cal J}}_{\tau-})\ind_{\{\tau < \tau^{1,{\cal J}} \wedge \tau^C \wedge T\}} + \theta^C(\hat{V}^{{\cal J}}_\tau, M^{{\cal J}}_{\tau-})\ind_{\{\tau < \tau^{1,{\cal J}} \wedge \tau^I \wedge T\}} ,\nonumber
\end{align}
where $f$ takes a similar form as in the single name case treated in~\eqref{eq:f+}, and is given by
\begin{align}
    f\bigl(t,v,z,z^I,z^C; M, {\cal J}\bigr) &:= -\biggl(\rfp \bigl(v+ z + z^I + z^C  -  M\bigr)^+ -\rfm \bigl(v + z + z^I + z^C  -   M_t\bigr)^- \nonumber\\
    & \phantom{==} - r_D z - r_D z^I - r_D z^C + \rcp M^+ - \rcm M^-  + \sum_{k \notin \cal{J}} \eta^k  \biggr ).
\label{eq:f+.1}
\end{align}
Similar to Eq.~\eqref{eq:BSDE-hatV-multi2}, the wealth replicating process in the state where all reference entities in {$\mathcal{J}$} have defaulted is related to the state where the additional entity $i\not \in \cal{J}$ defaults:
\begin{align}
  &-dV_t^{{\cal J}} = f\Bigl(t,V_t^{{\cal J}},  \sum_{k \notin \cal{J}} Z_t^{k,{\cal J}} , Z_t^{I,{\cal J}},Z_t^{C,{\cal J}}; M^{{\cal J}}, {\cal J}\Bigr) \, dt  \label{eq:BSDE-sell01}\\
 &\qquad\qquad- \sum_{k \notin \cal{J}}  Z^{k,{\cal J}}_t\, d\varpi_t^{k,{\cal J},\Qxx} - Z_t^{I,{\cal J}} \, d\varpi_t^{I,\Qxx}  - Z_t^{C,{\cal J}} \, d\varpi_t^{C,\Qxx},\nonumber\\
   & V^{{\cal J}}_{\tau \wedge \min_{j\notin {\cal J}} \tau^{j,{\cal J}}} = \sum_{k \notin \cal{J}}\Bigl(L^{k} + V_{\tau^{k,{\cal J} }}^{\{k\}\cup \cal{J}} \Bigr) \ind_{\{\tau^{k,{\cal J}} =\min_{j\notin {\cal J}} \tau^{j,{\cal J}} \wedge \tau  \}}\ind_{\{\tau^{k,{\cal J}} <T  \}}   \nonumber \\
    &\quad+ \theta^I\Bigl(\hat{V}^{{\cal J}}_\tau,M_{\tau-}\Bigr)\ind_{\{\tau^I < \min_{j\notin {\cal J}} \tau^{j,{\cal J}}\wedge \tau^C \wedge T\}} + \theta^C\Bigl(\gamma \hat{V}^{{\cal J}}_\tau, M_{\tau -}\Bigr)\ind_{\{\tau^C < \min_{j\notin {\cal J}} \tau^{j,{\cal J}} \wedge \tau^I \wedge T\}}  .\nonumber
\end{align}
%

Altogether, we obtain

\begin{align}
V_t &= \sum_{{\cal J} \in 2^{\{1, ..., N\}}} V^{{\cal J}}_{t}\ind_{\{\tau^{{\cal J}}  \wedge \tau^C \wedge \tau^I \wedge T < t \le\min_{k \notin \cal{J}} \tau^{k,{\cal J}} \wedge \tau^C \wedge \tau^I \wedge T\}}, \\
Z^i_t& = \sum_{{\cal J} \in 2^{\{1, ..., N\}  \backslash \{i\} }} Z^{1,{\cal J}}_{t}\ind_{\{\tau^{{\cal J}}  \wedge \tau^C \wedge \tau^I \wedge T < t \le\min_{k \notin \cal{J}} \tau^{k,{\cal J}} \wedge \tau^C \wedge \tau^I \wedge T\}}, \\
Z^I_t& = \sum_{{\cal J} \in 2^{\{1, ..., N\}}} Z^{I,{\cal J}}_{t}\ind_{\{\tau^{{\cal J}}  \wedge \tau^C \wedge \tau^I \wedge T < t \le\min_{k \notin \cal{J}} \tau^{k,{\cal J}} \wedge \tau^C \wedge \tau^I \wedge T\}}, \\
Z^C_t &= \sum_{{\cal J} \in 2^{\{1, ..., N\}}} Z^{C,{\cal J}}_{t}\ind_{\{\tau^{{\cal J}}  \wedge \tau^C \wedge \tau^I \wedge T < t \le\min_{k \notin \cal{J}} \tau^{k,{\cal J}} \wedge \tau^C \wedge \tau^I \wedge T\}}.
\end{align}

Proceeding along the lines of Section \ref{sec:OneEntity}, we can obtain a BSDE for the $\XVA$ process in  Eq.~\eqref{eq:XVAdef}:
\begin{align}
-d\XVA_t^{{\cal J}} & = \tilde{f}^{}\Bigl(t,\XVA_t^{{\cal J}},\sum_{k \notin \cal{J}} \tilde Z_t^{k,{\cal J}} , \tilde Z_t^{I,{\cal J}},\tilde Z_t^{C,{\cal J}}; M^{{\cal J}}, {\cal J}\Bigr) \, dt \\ 
& \phantom{=} - \sum_{k\notin  J}  \tilde Z^{k,{\cal J}}_t\, d\varpi_t^{k,{\cal J},\Qxx} - \tilde Z_t^{I,{\cal J}} \, d\varpi_t^{I,\Qxx}  - \tilde Z_t^{C,{\cal J}} \, d\varpi_t^{C,\Qxx},\nonumber\\
	\XVA_{\tau \wedge \min_{j\notin {\cal J}} \tau^{j,{\cal J}} }^{{\cal J}} &= \sum_{k \notin \cal{J}}\Bigl(L^{k} + \XVA_{\tau^{k,{\cal J} }}^{\{k\}\cup \cal{J}} \Bigr) \ind_{\{\tau^{k,{\cal J}} =\min_{j\notin {\cal J}} \tau^{j,{\cal J}} \wedge \tau  \}} \ind_{\{\tau^{k,{\cal J} } <T  \}}   \\
	&\qquad+\tilde{\theta}_I\Bigl(\hat{V}^{{\cal J}}_\tau,M_{\tau-}\Bigr)\ind_{\{\tau^I < \min_{j\notin {\cal J}} \tau^{j,{\cal J}}\wedge \tau^C \wedge T\}} \nonumber\\
   &\qquad+ \tilde{\theta}^C\Bigl(\hat{V}^{{\cal J}}_\tau, M_{\tau -}\Bigr)\ind_{\{\tau^C < \min_{j\notin {\cal J}} \tau^{1,{\cal J}} \wedge \tau^I \wedge T\}} \Bigr),\nonumber
\end{align}
where $\tilde{\theta}^{C}$ and $\tilde{\theta} ^{I}$ are given in~\eqref{eq:hats}, $\tilde{Z}_t^{1,{\cal J}}, \tilde{Z}_t^{I, {\cal J}}$ and $\tilde{Z}_t^{C, {\cal J}}$ are defined as
\begin{align}
\tilde{Z}_t^{1,{\cal J}} & := Z_t^{1,{\cal J}} - \hat{Z}^{1,{\cal J}},~i\notin {\cal J}, \qquad \tilde{Z}_t^{I, {\cal J}} := Z_t^{I, {\cal J}}, \qquad \tilde{Z}_t^{C, {\cal J}} := Z_t^{C, {\cal J}},  \nonumber
\end{align}
and	
\begin{align}
\tilde{f}\Bigl(t,xva,\tilde{z},\tilde{z}^I,\tilde{z}^C; M, {\cal J} \Bigr) &{:}= -\biggl(\rfp \Bigl(xva+ \tilde{z}+ \tilde{z}^I + \tilde{z}^C + \sum_{k \notin \cal{J}}   L^k- M_t \Bigr)^+ \\
& \phantom{=:}-\rfm \Bigl(xva + \tilde{z} + \tilde{z}^I + \tilde{z}^C +  \sum_{k \notin \cal{J}}  L^k- M_t \Bigr)^- \nonumber\\
& \phantom{=:} - r_D \tilde{z}- r_D \tilde{z}^I - r_D \tilde{z}^C + \rcp M_t^+ - \rcm M_t^- -r_D \sum_{k \notin \cal{J}}   L^k\biggr) ,\nonumber
\end{align}
where the terminal condition is $\XVA_t^{\{1, ..., N\}} =0.$ The BSDE in the reduced filtration $\mathbb{F}$ can be obtained analogously to~\eqref{eq:reduced}, and is given by
\begin{align}\label{eq:reduced-multi}
	-d\check{U}_t^{{\cal J}} = \check{g}\Bigl(t,\check{U}_t^{{\cal J}}, \sum_{k \notin \cal{J}} \check{U}_t^{\{k\}\cup \cal{J}},  \sum_{k \notin \cal{J}} h^{k,\Qxx} \check{U}_t^{\{k\}\cup \cal{J}} ; \hat{V}^{{\cal J}}, M^{{\cal J}}, {\cal J}\Bigr) \, dt, \qquad \qquad
	\check{U}_{T}^{{\cal J}} = 0,
\end{align}
with
\begin{align}
	\check{g}\bigl(t,\check{u}, \circu, \circuo; \hat{V},M,{\cal J}\bigr) &= h^{I,\Qxx}\Bigl(\tilde{\theta} ^{I}(\hat{V}_t^{{\cal J}}, M_{t-}^{{\cal J}})-\check{u}\Bigr) + h^{C,\Qxx}\Bigl(\tilde{\theta} ^{C}(\hat{V}_t^{{\cal J}}, M_{t-}^{{\cal J}})-\check{u}\Bigr) + \Bigl( \circuo- \sum_{k \notin \cal{J}}  h^{k,\Qxx} \check{u} \Bigr)\nonumber\\
	& \phantom{}+ \tilde{f}\Bigl(t, \check{u}, \circu -(N-\abs{{\cal J}})\check{u}, \tilde{\theta} ^{I}(\hat{V}_t^{{\cal J}}, M_{t-}^{{\cal J}})-\check{u}, \tilde{\theta}^{C}(\hat{V}_t^{{\cal J}}, M_{t-}^{{\cal J}}) -\check{u}; M\Bigr).
\label{eq:g+-multi}
\end{align}
For the starting point of the recursion, we set  $\check{U}_t^{\{1, ..., N\}} =0.$

\begin{remark}
For large portfolios, i.e., those referencing a high number $N$ of entities, this system of ODEs is computationally intractable. A solution to the ODE~\eqref{eq:g+-multi} would need to be obtained for each subset of $\{1, ..., N\}$, that is a total of $2^N$ solutions need to be computed. This system becomes tractable only if the reference entities have identical characteristics (spreads, loss rates and default intensities), and the default intensities depend only on the \textit{number} of occurred defaults, but not on the identity of the defaulted entities. In this case, the complexity grows linearly and it is required to compute $N$ ODE solutions. The dependence structure between ODEs has similar characteristics to that arising in a binomial tree. Computations on a non-recombining tree are generally prohibitively expensive, and thus recombining trees are usually used.

For calibration purposes, it is important to construct a parsimonious model, that minimizes the number of parameters to be estimated from data. A common strategy is to split firms into groups, each characterized by unique default risk profile. In other words, firms in the same group are assumed to have the same default intensity. For instance, a tractable specification for the default intensity is $h^{i,\Qxx}_t(|{\cal J}|) = h_t^{\kappa(i),\Qxx}(|{\cal J}|)$, where $\kappa(i)$ maps firm $i$ to the group to which $i$ belongs to. Examples of tractable specifications include the linear counterparty risk model, where $ {h_t^{\kappa(i),\Qxx}(|{\cal J}|)} = \lambda_0^{\kappa(i)} + \lambda_1^{\kappa(i)} |{\cal J}|$, i.e., upon default of some firm the default intensity of the surviving firms increases by a constant amount, so that default dependence increases linearly with the number of defaults. The constant $\lambda_0$ and $\lambda_1$ may change from group to group, but are the same for all firms in the group. Another specification accounting for between groups and within group dependencies has been calibrated by \cite{FreyBackhaus04} using CDO tranche spreads. Their model accounts for the fact that defaults of firms belonging to the same industry have a stronger dependence structure than defaults of firms from different industries. The functional specification assumes that a default event at time $t$ increases the default intensity of surviving firms only if (i) the realized number of defaults within a group exceeds the expected number of defaulted firms for that group by $t$, or (ii) the total realized number of defaults within a group exceeds the expected total number of defaulted firms by $t$. One of the calibrated parameters in their model reflects the strength of interaction between firms from the same industry group relative to the global interaction between defaults in the entire portfolio.

\end{remark}

Assume that $h^{i,\Qxx}, i\in\{1, ...,N\}$, are all piecewise (deterministic) continuous. The extension of the theorems developed in Section \ref{sec:ODE} for the case of a single CDS to the case of portfolios referencing multiple  entities is straightforward. For notation consistency, denote $\hat v^{{\cal J}}_t = \hat V_{T-t}^{{\cal J}}, \, \check u^{{\cal J}}_t = \check U_{T-t}^{{\cal J}}$ for all ${\cal J}\subset \{1, ..., N\}$. 
The following proposition is the multi-dimensional extension of Proposition \ref{thm:existence}. Its proof uses exactly the same arguments and is omitted here.
\begin{proposition}\label{thm:existence-multi}  {There exists a unique (piecewise smooth) solution to the system of ODEs:
\begin{align}
\partial_t \hat v^{{\cal J}} &=-r_D\hat{v}^{{\cal J}} - \sum_{k \notin \cal{J}} \eta^{k}    +\sum_{k \notin \cal{J}}\Bigl( L^{k} + \hat v^{\{k\}\cup \cal{J}} - v^{{\cal J}}  \Bigr)h_k^\Qxx ,~{\cal J}\subset\{1, ..., N\},\label{eq:hatV-ODE-multi}\\
\hat v^{{\cal J}}_0 &=0,\nonumber\\
\partial_t \check u^{{\cal J}} &= \check{g}\left(t,\check u^{{\cal J}}, \sum_{k \notin \cal{J}} \check{u}^{\{k\}\cup \cal{J}}, \sum_{k \notin \cal{J}} h^{k,\Qxx}  \check{u}^{\{k\}\cup \cal{J}}; \hat v^{{\cal J}}, m, {\cal J}\right),~{\cal J}\subset\{1, ..., N\},\label{eq:creditswap1-multi}\\
\check u^{{\cal J}}_0 &=0.\nonumber
\end{align}
}
\end{proposition}
Next, we present the multi-dimensional extension of Theorem \ref{thm:comp}. The proof presents an additional induction step compared with the proof of Theorem \ref{thm:comp}, and the details are reported in the appendix. 
\begin{theorem}[Comparison Theorem]
\label{thm:comp-multi}
Let ${\cal J}\subset\{1, ..., N\}$. {Assume in addition to Assumption \ref{ass:nec+suff} that $ \rfm< \min_{i \in \{1,\ldots N, I\}} \mu^i{\wedge\underline\mu^C}.$} Moreover, assume that there exists
$\overline \mu_C \ge \underline \mu^C >r_D$ such that
$\overline \mu^C \ge  {\mu_t^{C,\Qxx}} \ge \underline \mu^C$, and let $\check u^{{\cal J}}$ be the solution of ODE~\eqref{eq:creditswap1-multi}.
Let
\begin{align}
{(\mu^C)^*}(\hat{v}, m,  \check u) &= {\overline \mu^C} \ind_{\{ \tilde{\theta} ^{C}(\hat{v},m)-\check u\ge 0\}} + {\underline \mu^C} \ind_{\{ \tilde{\theta} ^{C}(\hat{v},m)-\check u \le 0\}},\\
{{(\mu^C)}_*}(\hat{v}, m,  \check u) &= {\underline \mu^C} \ind_{\{ \tilde{\theta} ^{C}(\hat{v},m)-\check u\ge 0\}} + {\overline \mu^C} \ind_{\{ \tilde{\theta} ^{C}(\hat{v},m)-\check u \le 0\}},
\end{align}
and define ${g}^{*}$ and ${g}_{*}$ plugging $(h^{C,\Qxx})^{*}$ and $(h^{C,\Qxx})_{*}$ into $\check g$ given in Eq.~\eqref{eq:g+-multi}, i.e.,
\begin{align}
&{g}^{*}\bigl(t,\check u, \circu, \circuo; \hat{v},m,{\cal J}\bigr) := h^{I,\Qxx}\bigl(\tilde{\theta} ^{I}(\hat{v}_t,m_t)-\check{u}_t\bigr) +{( {(\mu^C)^{*} }(\hat v_t, m_t,\check u_t) - r_D )} \bigl(\tilde{\theta} _{C}(\hat{v}_t,m_t)-\check{u}_t\bigr)\nonumber \\
& \qquad+ \bigl(\circuo-  \sum_{i\notin {\cal J}} h^{i,\Qxx} \check{u}_t\bigr)+ \tilde{f}\bigl(t, \check{u}_t, \circu-(N-\abs{{\cal J}})\check{u}_t, \tilde{\theta} ^{I}(\hat{v}_t,m_t)-\check{u}_t, \tilde{\theta}^{C}( \hat{v}_t,m_t) -\check{u}_t;\hat{v},m,{\cal J}\bigr),\nonumber\\
&{g}_{*}\bigl(t,\check u,\circu, \circuo; \hat{v},m,{\cal J}\bigr) : = h^{I,\Qxx}\bigl(\tilde{\theta}^{I}(\hat{v}_t, m_t)-\check{u}_t\bigr) +{( {(\mu^C)_{*} }(\hat v_t, m_t,\check u_t) - r_D )} \bigl(\tilde{\theta}^{C}(\hat{v}_t,m_t)-\check{u}_t\bigr)\nonumber \\
& \qquad+ \bigl(\circuo -  \sum_{i\notin {\cal J}} h^{i,\Qxx} \check{u}_t\bigr)+ \tilde{f}\bigl(t, \check{u}_t, \circu-(N-\abs{{\cal J}})\check{u}_t, \tilde{\theta} ^{I}(\hat{v}_t,m_t)-\check{u}_t, \tilde{\theta}^{C}( \hat{v}_t,m_t) -\check{u}_t;\hat{v},m,{\cal J}\bigr).\nonumber
\end{align}
Finally, let $\check{u}^{\cal J,*}$ be the solution to ODE~\eqref{eq:creditswap1-multi}, but with $\check{g}$ replaced by ${g}^{*}$, that is
\begin{align}
 \partial_t \check u^{\cal J,*} &= g^{*}\left(t,\check u^{{\cal J},*},  \sum_{k \notin \cal{J}}  \check{u}^{\{k\}\cup \cal{J},*},  \sum_{k \notin \cal{J}} h^{k,\Qxx}  \check{u}^{\{k\}\cup {\cal J},*}; \hat v^{{\cal J}}, m, {\cal J}\right),\qquad \qquad
\check{u}^{{\cal J},*}_0 =0,\nonumber
\end{align}
and similarly, let $\check{u}_{*}^{{\cal J}}$ be the solution of ODE~\eqref{eq:creditswap1-multi} where we replace $\check{g}$ with ${g}_{*}$. Then $\check{u}_{*}^{{\cal J}} \le \check u^{{\cal J}} \le \check{u}^{{\cal J},*}.$
\end{theorem}

It now remains to find the super-replicating strategy for the robust XVA process. Following similar arguments to those used above, the strategy will be obtained by pasting together the various quantities associated with different subsets {$\mathcal{J}$} of defaulted entities.
\begin{theorem}\label{thm:muli}
The robust $\XVA$ can be represented explicitly by
\begin{align}
\rXVA_t  &= \sum_{{\cal J} \in 2^{\{1, ..., N\}}} \check{U}_t^{{\cal J},*}\ind_{\{\tau^{{\cal J}}  \wedge \tau^C \wedge \tau^I \wedge T < t \le\min_{k \notin \cal{J}} \tau^{k,{\cal J}} \wedge \tau^C \wedge \tau^I \wedge T\}}\label{eq:rXVA-tmp} \\
&+ \Bigl( \tilde{\theta}^{C}(\hat{V}_{\tau^C}, M_{\tau^C-}) \ind_{\{\tau^C < \min_{j\notin {\cal J}} \tau^{j,{\cal J}}\wedge \tau^I \wedge T\}} + \tilde{\theta}^{I}(\hat{V}_{\tau^I}, M_{\tau^I -}) \ind_{\{\tau^I < \min_{j\notin {\cal J}} \tau^{j,{\cal J}}\wedge \tau^C \wedge T\}} \Bigr) \ind_{\{ \tau^C \wedge \tau^I \le t \le T\}},\nonumber
\end{align}
where the process $\check{U}_t^{{\cal J},*} := \check u^{{ \cal J},*}_{T-t}$. Define
\begin{align}
\xi^{i,\cal{J}, *}_t & = \frac{\check{U}_t^{{\cal J},*} -\check{U}_t^{\{i\}\cup \cal{J},*} }{B_{t-}^i} \ind_{\{\tau^{{\cal J}}  \wedge \tau^C \wedge \tau^I \wedge T < t \le\min_{k \notin \cal{J}} \tau^{k,{\cal J}} \wedge \tau^C \wedge \tau^I \wedge T\}}, \label{eq:rob_strat1-mult} \\
\xi_{t}^{I,{\cal J},*} &= \frac{L^I (\hat{V}_t - M_{t-})^{+} + \check{U}_t^{{\cal J},*}}{B_{t-}^I}\ind_{\{\tau^{{\cal J}}  \wedge \tau^C \wedge \tau^I \wedge T < t \le\min_{k \notin \cal{J}} \tau^{k,{\cal J}} \wedge \tau^C \wedge \tau^I \wedge T\}}, \nonumber\\
\xi_{t}^{C,{\cal J},*} &= \frac{ -L^C (\hat{V}_t - M_{t-})^{-}+ \check{U}_t^{{\cal J},*}}{B_{t-}^C}\ind_{\{\tau^{{\cal J}}  \wedge \tau^C \wedge \tau^I \wedge T < t \le\min_{k \notin \cal{J}} \tau^{k,{\cal J}} \wedge \tau^C \wedge \tau^I \wedge T\}}, \nonumber\\
\xi_t^{f,{\cal J},*} &= \frac{- \check{U}_t^{{\cal J},*}  - \sum_{i\notin {\cal J}}\big( \check{U}_t^{{\cal J},*} -\check{U}_t^{\{i\}\cup \cal{J},*} \big)+ L^C (\hat{V}_t - M_{t-})^{-}- L^I  (\hat{V}_t -M_{t-})^{+} - M_{t-}}{B_t^{r_f}}\nonumber\\
&\quad\times \ind_{\{\tau^{{\cal J}}  \wedge \tau^C \wedge \tau^I \wedge T < t \le\min_{k \notin \cal{J}} \tau^{k,{\cal J}} \wedge \tau^C \wedge \tau^I \wedge T\}}.\nonumber
\end{align}
The super-replicating strategies for $\rXVA$ are obtained from the above conditional strategies as
\begin{align}
\xi^{i, *}_t & =  \sum_{{\cal J} \in 2^{\{1, ..., N\}  \backslash \{i\} }} \xi^{i,\cal{J},*}_{t}\ind_{\{\tau^{{\cal J}}  \wedge \tau^C \wedge \tau^I \wedge T < t \le\min_{k \notin \cal{J}} \tau^{k,{\cal J}} \wedge \tau^C \wedge \tau^I \wedge T\}},\label{eq:rob_strat2-mult} \\
\xi^{I,*}_t& = \sum_{{\cal J} \in 2^{\{1, ..., N\}}} \xi^{I,{\cal J},*}_{t}\ind_{\{\tau^{{\cal J}}  \wedge \tau^C \wedge \tau^I \wedge T < t \le\min_{k\notin {\cal J}} \tau^{k,{\cal J}} \wedge \tau^C \wedge \tau^I \wedge T\}}, \nonumber\\
\xi^{C,*}_t &= \sum_{{\cal J} \in 2^{\{1, ..., N\}}} \xi^{C,{\cal J},*}_{t}\ind_{\{\tau^{{\cal J}}  \wedge \tau^C \wedge \tau^I \wedge T < t \le\min_{k \notin \cal{J}} \tau^{k,{\cal J}} \wedge \tau^C \wedge \tau^I \wedge T\}},\nonumber \\
\xi_t^{f,*} &= \sum_{{\cal J} \in 2^{\{1, ..., N\}}} \xi_t^{f,{\cal J},*}\ind_{\{\tau^{{\cal J}}  \wedge \tau^C \wedge \tau^I \wedge T < t \le\min_{k\notin {\cal J}} \tau^{k,{\cal J}} \wedge \tau^C \wedge \tau^I \wedge T\}}, \nonumber
\end{align}
together with the number of shares held in the collateral account given by
\begin{align}
\psi_t^{m,*}&=  - \frac{M_{t-}}{B_t^{r_m}}  \ind_{\{t\le \tau^{\{1, ..., N\}}  \wedge \tau^C \wedge \tau^I \wedge T \}}. \label{eq:rob_strat3-mult}
\end{align}
\end{theorem}

{We expect that, as the number of reference entities increases, so does the difference between sub-replication and super-replication valuations. This may be intuitively understood as follows. The terminal/closeout condition {of the super-replication} in the case of a single reference entity matches the terminal/closeout condition of the $\XVA$. However, as shown in Theorem \ref{thm:muli}, in the case of a credit default swap portfolio where multiple reference entities appear, the terminal/closeout of the super-replication dominates that of the $\XVA$. Thus, in the case of two reference entities, the terminal/closeout condition includes the jump {to the} closeout/terminal condition for the single reference entity case, in addition to the cash flows accumulated prior to default. It is thus expected that the difference between the super-replication and the $\XVA$ in the case of two reference entities is greater than the corresponding difference when the portfolio consists of a single CDS. Iterating this reasoning inductively, we conclude that the difference between the super-replication valuation and the $\XVA$ grows as more CDS contracts are added to the portfolio. This highlights the importance of the proposed robust approach, as opposed to the naive approach, which plugs one of the two extremes of the counterparty's intensity uncertainty interval into the valuation formulas.}

\subsection{Extension to Portfolios with Jump-Diffusion Risk}  \label{sec:jumpdiffrisk}
{
	The portfolio considered in this paper consists of CDSs subject to jump risk only. We sketch a possible generalization of the valuation equations to the case that firms' default intensities follow a diffusion process. Specifically, we considers CDSs referencing publicly traded firms, and assume their default intensities to depend on the firm's stock price via a multivariate extension of the  defaultable stock model proposed by \cite{Linetsky}. In his model, stock prices follow diffusion dynamics prior to default and jump to zero at a default time. The intensity of the default time depends on the pre-default value of the stock. We start by assuming the existence of a $N$-dimensional correlated Brownian motion $ (W^{1,\Px}, ..., W^{N,\Px})$, with correlation matrix $R$, assumed to be invertible. We then augment the filtration $\mathbb F$ in our paper to additionally include the independent filtration generated by the Brownian motions $ (W^{1,\Px}, ..., W^{N,\Px})$.
	The dynamics of the defaultable stock prices is given by
	\begin{align}
	dS^i_t = S^i_t \Bigl( \mu^{i,S} dt + \sigma^i dW^{i,\Px}_t - d\varpi^{i,\Px}_t\Bigr), ~{0\le t <\tau^i, ~i=1, ..., N ,}
	\end{align}
	where $\mu^{i,S}$ and $\sigma^i$ are, respectively, the constant drift and volatility of the $i$-th stock, and $\varpi^{i,\Px},~i=1, ...,N $ is the jump-to-default martingale given by
	\begin{align}
	\varpi^{i,\Px}_t := H^i_t - \int_0^t \bigl(1-H^i_u\bigr)h^{i,\Px}_u(S^i_u) \, du,~i \in \{1,\ldots,N\},
	\end{align}
	and the default intensity function $h^{i,\Px}$ depends, besides time, on the stock price. Typically, $h^{i,\Px}$ is nonnegative and decreasing in $s$, as we expect the default probability of a firm to rise if its stock price declines. For this illustration, assume the intensity $h^{i,\Px}$ to be bounded from above and away from zero.

Additionally, we assume the rate of bond $i$ to be a smooth bounded function of the stock price and time, i.e., $\mu^i_t := { \mu_t^i(S_t^i) }$. Such an assumption is also empirically supported; see, for instance, \cite{Kwan}, who show that stocks and bonds issued by the same firm are positively correlated, being both claims on the same underlying assets. If yields of a firm's bond decline (and thus the bond price increase), the firm's stock price typically increases. 
Similar to Assumption \ref{ass:nec+suff}, we will assume that $\mu^i_t > r_D \vee r_f^{+}$ {for $t\in[0,T]$}. 
	
	We enlarge the set of instruments in the replication to also include defaultable stocks. The defaultable stocks replicate the fluctuations of the default intensity governed by the Brownian motions as well as part of the jump risk (through the jump to default of the stocks). The remaining (positive or negative) jump risk is replicated by defaultable bonds as in the main model presented in this paper.
We proceed as above, and define the Radon-Nikodym derivative process
{
\begin{align}\label{eq:RN1}
\frac{d\Qxx}{d\Px} \bigg|_{\mathcal{F}_{{\tau \wedge( \tau^1 \vee ...\vee \tau^N)     }}} & := e^{\sum_{i=1}^N \theta^i W^{i,\Px}_{\tau\wedge \tau^i} - \frac12\sum_{i,j=1}^N \theta^i R_{i,j} \theta^j (\tau\wedge \tau^i\wedge\tau^j) }\\
& \times\prod_{i\in \{1, \ldots, N, I, C\}} \Biggl(\frac{\int_0^{\tau \wedge \tau^{i}} (\mu^{i}_u - r_D)du}{\int_0^{\tau \wedge \tau^{i}} h^{i,\Px}_u(S^i_u)du }\Biggr)^{H^{i}_{\tau \wedge \tau^{i}}} e^{\int_0^{\tau \wedge \tau^{i}}(r_D-\mu^{i}_u+h^{i,\Px}_u(S^i_u))du},
\end{align}
where $(\theta^1, ..., \theta^N)^T = R^{-1}\left(\frac{\mu^{1,S} -r_D}{\sigma^1},...,\frac{\mu^{N,S} -r_D}{\sigma^N}\right)^T. $ 
It then follows that the measure change is well-defined. Note that, under $\Qxx$, we continue having the relation $h^{i,\Qxx}_t = \mu^i_t-r_D$, $ i \in \{ 1, \ldots, N\}$. However, the functions $h^{i,\Qxx}$, in addition to time, also depend on stock price $S^i$, and are bounded.}
The wealth process (equivalent of \eqref{eq:wealth}), is now given by
	\begin{equation}\label{eq:wealth1}
	V_t := \sum_{i=1}^N \left(\xi_t^i B_t^i  + \xi_t^{S,i} S^i_t\right)+ \xi_t^I B_t^I + \xi_t^C B_t^C + \xi_t^f B_t^{r_f} - \psi_t^{m} B_t^{r_m}.
	\end{equation}
	Assume, for illustration purposes, a portfolio consisting of a single defaultable stock, i.e., $N=1$. Setting $\gamma=1$, the BSDE equivalent of \eqref{eq:BSDE-sell} satisfied by the wealth process takes the form
	\begin{align}\label{eq:vtlast21}
	\nonumber dV_t &= \Bigl(\rfp \bigl(V_t +Z_t^{S^1}+ Z_t^{1} + Z_t^{I} + Z_t^{C} -  M_t \bigr)^+ -\rfm \bigl(V_t +Z_t^{S^1}+ Z_t^{1} + Z_t^{I} + Z_t^{C} -   M_t \bigr)^-\\
	\nonumber & \phantom{=}  - r_D Z_t^{S^1} - r_D Z_t^{1} - r_D Z_t^{I} - r_D Z_t^{C}  +  \rcp   M_t^+ - \rcm  M_t^-  +  \eta^1 \Bigr)\,  dt \\
	& \phantom{=}+ Z_t^{S^1}dW^{1,\Qxx}_t + Z_t^{1} \,d\varpi_t^{1,{ \Qxx}}  + Z_t^{I} \, d\varpi_t^{I,{ \Qxx}}  + Z_t^{C} \, d\varpi_t^{C,\Qxx},\\
	V_{\tau \wedge \tau_1} & =    L^1 \ind_{\tau^1 < \tau} + \theta^I(\hat{V}_\tau, M_{\tau-})\ind_{\{\tau < \tau^1 \wedge \tau^C \wedge T\}} + \theta^C(\hat{V}_\tau, M_{\tau -})\ind_{\{\tau < \tau^1 \wedge \tau^I \wedge T\}}.
	\end{align}
	In the above expressions, the processes $Z^1,Z^I, Z^C$ have a similar role as in Eq.~\eqref{eq:Zetas}, and we additionally have the process $Z^{S^1}$ that represents the volatility adjusted amount invested in the stock to replicate fluctuations of firm 1's default intensity. 	The BSDE \eqref{eq:BSDE-hatV} describing the valuation process  $\hat V$ of the valuation party now becomes
	\begin{align}
	- d\hat{V}_t = \bigl(-r_D \hat{V}_t  - \eta^1 \bigr)\, dt  -\hat Z_t^{S^1}dW^{1,\Qx}-\hat{Z}_t^1 \, d \varpi_t^{1,\Qxx}, \qquad \qquad
	\hat{V}_{\tau^1 \wedge T} = L^1 \ind_{\tau^1 < T},\label{eq:BSDE-hatV11}
	\end{align}
	{(i.e., the underlying assumption is that the valuation party uses the minimal (entropy) martingale measure for the valuation)}.
	Using the expression for $V_t$ and $\hat{V}_t$, we obtain the BSDE for the XVA process,  i.e., the analogous of \eqref{eq:XVABSDE} but for the case where the underlying credit swap portfolio is subject to both diffusion and jump-to-default risk:
	\begin{align}
	-d\XVA_t^{\buysell} & =\tilde{f}^{S}\bigl(t,\XVA_t^{\buysell},\tilde{Z}_t^{S^1},\tilde{Z}_t^{1},\tilde{Z}_t^{I},\tilde{Z}_t^{C}; M\bigr)dt
	- Z_t^{S^1}dW^{1,\Qxx}_t- \tilde{Z}_t^{1}\, d\varpi_t^{1,\Qxx}  - \tilde{Z}_t^{I} \, d\varpi_t^{I,\Qxx}  - \tilde{Z}_t^{C} \, d\varpi_t^{C,\Qxx},\nonumber\\
	\XVA_{\tau \wedge \tau^1}^{\buysell} &= \tilde{\theta}^{C}(\hat V_\tau, M_{\tau -}) \ind_{\{\tau <\tau^1 \wedge \tau^I \wedge T\}} + \tilde{\theta}^{I}(\hat V_\tau, M_{\tau -})\ind_{\{\tau <\tau^1 \wedge \tau^C \wedge T \}},
	\label{eq:XVABSDE1}
	\end{align}
	where
	\begin{align}
	&\!\!\!\!\!\tilde{f}^S\bigl(t,xva,\tilde{z}^{S_1},\tilde{z}^1,\tilde{z}^I,\tilde{z}^C; M\bigr) = \tilde{f}\bigl(t,xva,\tilde{z}^{S^1}+\tilde{z}^1,\tilde{z}^I,\tilde{z}^C; M\bigr)\nonumber\\
	& \qquad\qquad:= -\Bigl(\rfp \bigl(xva+ \tilde{z}^{S_1}+\tilde{z}^1 + \tilde{z}^I + \tilde{z}^C +  L^1- M_t \bigr)^+ -\rfm \bigl(xva +\tilde{z}^{S^1}+ \tilde{z}^1 + \tilde{z}^I + \tilde{z}^C +  L^1- M_t \bigr)^- \nonumber\\
	&\qquad\qquad  -r_D \tilde{z}^{S^1}- r_D \tilde{z}^1 - r_D \tilde{z}^I - r_D \tilde{z}^C + \rcp \bigl(M_t\bigr)^+ - \rcm\bigl(M_t\bigr)^- -r_D L^1\Bigr).
	\end{align}
	The projection technique by \cite{CrepeyRed} on the (now non-trivial) filtration $\mathbb{F}$ leads to
	\begin{align}\label{eq:reduced1}
	-d\check{U}_t^{\buysell} & = \check{g}^S\bigl(t,\check{U}_t^{\buysell},  \check Z_t^{S^1}; \hat{V}, M\bigr) \, dt - \check Z_t^{S^1}dW^{1,\Qxx}_t, ~~	\check{U}_{T}^{\buysell} = 0,
	\end{align}
	with driver
	\begin{align}
	\check{g}^S\bigl(t,\check{u},\check z; \hat{V},M\bigr) &= h^{I,\Qxx}\bigl(\tilde{\theta}^{I}(\hat{V}_t, M_{t-})-\check{u}\bigr) + h^{C,\Qxx}\bigl(\tilde{\theta} ^{C}(\hat{V}_t, M_{t-})-\check{u}\bigr) - h^{1,\Qxx} \check{u} \nonumber\\
	& \phantom{==}+ \tilde{f}^S\bigl(t, \check{u}, \check z, -\check{u}, \tilde{\theta}^{I}( \hat{V}_t, M_{t-})-\check{u}, \tilde{\theta}^{C}( \hat{V}_t, M_{t-}) -\check{u}; M\bigr)\label{eq:g+1}.
	\end{align}
{The boundedness of the intensity $ h^{1,\Qxx}$ ensures that $\check g^S$ is Lipschitz.} In turn, existence and uniqueness of the solution to the continuous BSDE \eqref{eq:reduced1} with Lipschitz driver is a classical result (cf., e.g., \cite[Theorem 2.1.]{ElKaroui}).
	We now define $\check U^{*}$ and $\check U_{*}$ as the solutions of the two BSDEs given by
	\begin{align}\label{eq:reduced2}
	-d\check{U}^{*}_t & = \check{g}^{S,*}\bigl(t,\check{U}^{*}_t,  \check Z^{S_1,*}_t; \hat{V}, M\bigr) \, dt - \check Z^{S^1,*}_tdW^{1,\Qxx}_t, ~~	\check{U}^{*}_T= 0,\\
	-d\check{U}_{*_t} & = \check{g}_{*}^{S}\bigl(t,\check{U}_{*_t},  \check Z^{S^1}_{*_t}; \hat{V}, M\bigr) \, dt - \check Z^{S^1}_{*_t} dW^{1,\Qxx}_t, ~~	\check{U}_{*_T}= 0,
	\end{align}
	where $Z^{S^1}_{*_t}$ and $Z^{S^1,*}_t$ are, respectively, the volatility adjusted cash amount invested in the stock under the replication strategy associated with the (pre-default) lower and upper bound of XVA.
	The existence and uniqueness of solutions $\check{U}^{*}$ and $\check{U}_{*}$ follow from the same classical result referenced above.
	Finally, the comparison $\check{U}^{*} \ge \check{U} \ge \check{U}_{*}$ can be achieved by the standard comparison theorem for BSDEs, (c.f. e.g., \cite[Theorem 3.2.2]{Delong}).
	
}

\section{Comparative Statics Analysis}\label{sec:numeranalysis}
This section performs a comparative statics analysis of the monotonicity patterns of XVA and its replication strategies, for a portfolio consisting of five credit default swaps. Section~\ref{sec:defcont} sets up the default contagion model. Section~\ref{sec:numerics} presents the numerical results.

\subsection{Default Contagion Model}\label{sec:defcont}


We use the following specification for the defaultable intensities of reference entities, investor and her counterparty:
\begin{eqnarray*}
h^{I,\Qxx}_t &=& a_{10} + a_{12} \ind_{\tau^C \leq t} + a_{13}
\left(\ind_{\tau^1 \leq t} + \ind_{\tau^2 \leq t} + \ldots +
\ind_{\tau^N \leq t} \right) \\
h^{C,\Qxx}_t &=&  a_{20} + a_{21} \ind_{\tau^I \leq t} + a_{23}
\left(\ind_{\tau^1 \leq t} + \ind_{\tau^2 \leq t} + \ldots +
\ind_{\tau^N \leq t} \right)  \\
h^{i,\Qxx}_t &=&  a_{30} + a_{31} \ind_{\tau^I \leq t} + a_{32}
\ind_{\tau^C \leq t} + a_{33} \left(\ind_{\tau^1 \leq t} + \ldots
\ind_{\tau^{i-1} \leq t} + \ind_{\tau^{i+1} \leq t} + \ldots +
\ind_{\tau^N \leq t} \right).
\end{eqnarray*}
Recall that primary quantities in our model, i.e., the return rates of the defaultable accounts $\mu^i, i\in\{1, ..., N, I,C\}$, are defined from
Eq.~\eqref{Bieq} as $\mu^i = h^{i,\Qxx} +r_D$.
Notice that the above specification defines a homogenous credit portfolio, i.e., the default intensities of all reference entities are identical. These differ, however,
from default intensities of the investor $I$ and her counterparty $C$. The above specification of contagion via direct credit dependence was first introduced by
\cite{Jarrow}.




\subsection{Numerical Results}\label{sec:numerics}

Throughout the section, we use the following benchmark parameters: $r_D = 0.0001$, $N=5$, $L^I=0.5$, $L^C = 0.5$, $\alpha = 0.8$, $r_f^-=0.05$, $r_f^+=0.08$, $r_m^-=r_m^+=0.0001$. We set the default intensity parameters to $a_{10}=0.05$, $a_{13}=0.05$, $a_{20}=0.05$, $a_{23}=0.01$, $a_{30}=0.01$, $a_{33}=0.01$. Those values are in line with empirical estimates used by \cite{Yu}. We set the contractual credit default swap parameters $S = 0.02$ and $L^i = 0.5$ for $i=1,\ldots,5$. Because the replication process ends at the earliest of the investor and counterparty's default time, the parameters $a_{12}$, $a_{21}$, $a_{31}$, and $a_{32}$ do not play any role in the analysis. We set $\underline{\mu}^C= a_{20} + r_D$, and $\overline{\mu}^C = a_{20} +r_D + N a_{23}$. Throughout the section, we perform a comparative statics analysis with respect to the account rate parameters at the initial time of the transaction, i.e., $t=0$. We fix the maturity of the portfolio to $T=1$. In all graphs, we plot the robust XVA (denoted by ``upper'' in the legend and corresponding to $\check{U}^*$), the actual XVA (denoted by ``actual'' and corresponding to $\check{U}$), and the best-case XVA (denoted by ``lower'' and corresponding to $\check{U}_*$). We also plot the replication strategies associated with these three  XVA processes.

Observe that the number of shares in the reference entity defaultable account used in the replication strategy of upper, actual, and lower XVA does not need to preserve the monotonicity pattern of upper, actual, and lower XVA (such a monotonicity pattern is violated, for example, in Figure~\ref{fig:a23}). This is because the value of $\xi^{i}$ in Eq.~\eqref{eq:rob_strat1-mult} depends on the {\it difference} between the XVA in the state where all five reference entities are alive and that in the state where one of the reference entities has defaulted.

\subsubsection{Idiosyncratic component of counterparty's default intensity}

\begin{figure}
	\includegraphics[width=0.49\textwidth]{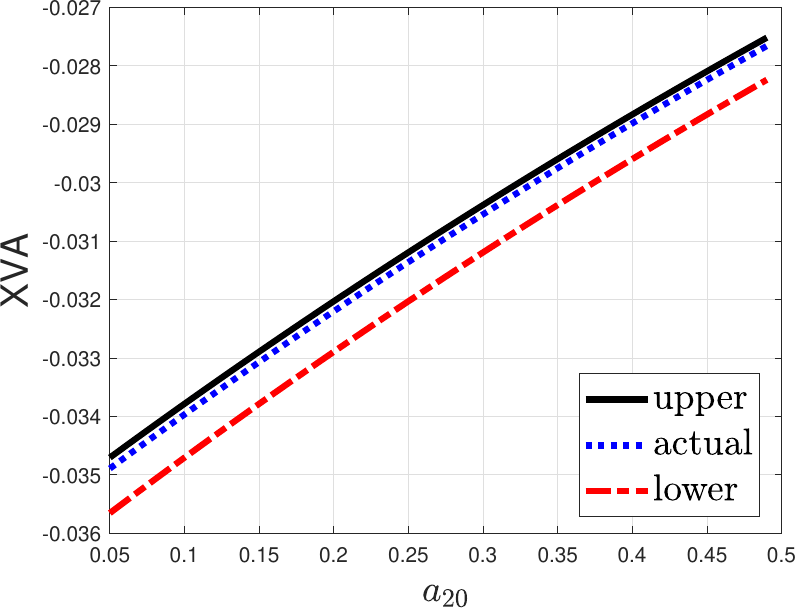}
	\includegraphics[width=0.49\textwidth]{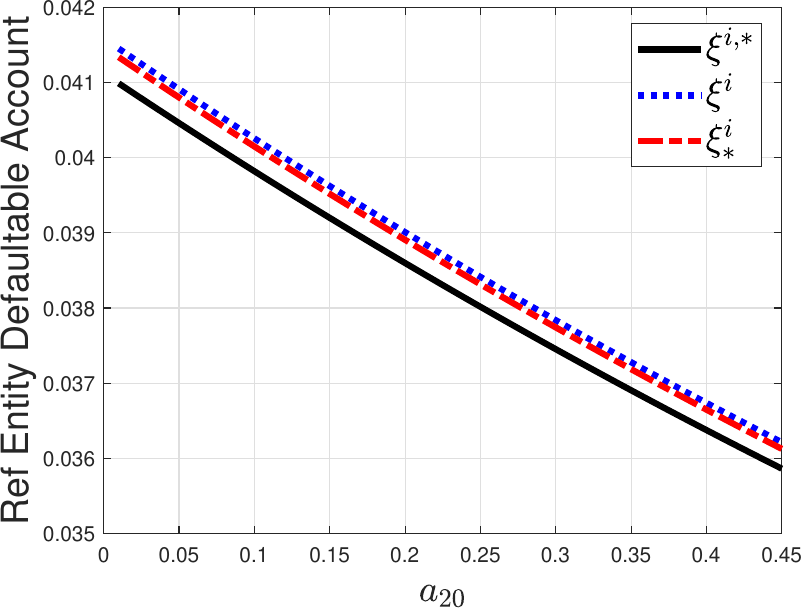}
    \includegraphics[width=0.49\textwidth]{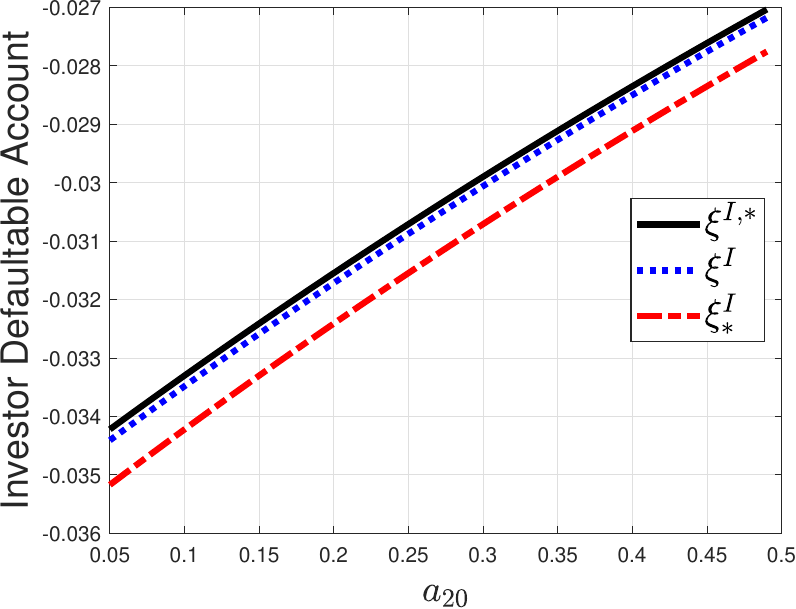}
	\includegraphics[width=0.49\textwidth]{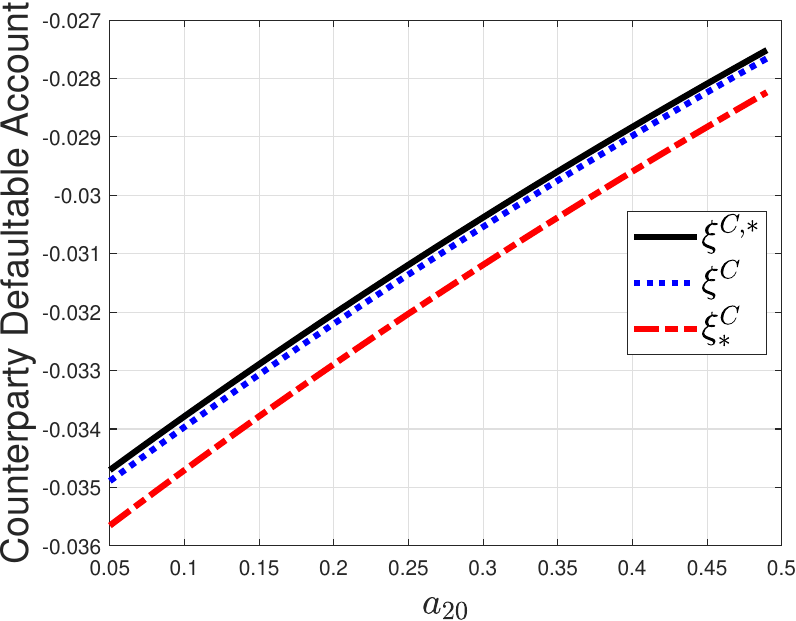}
    \includegraphics[width=0.49\textwidth]{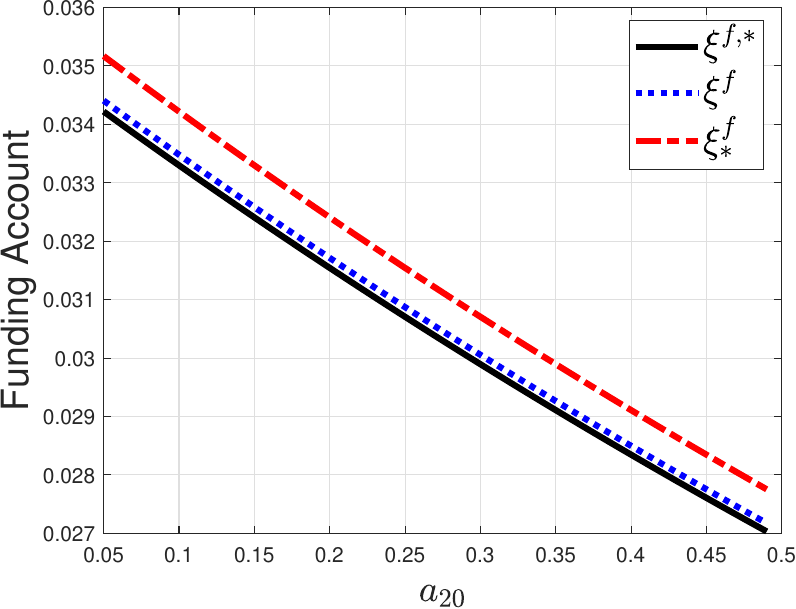}
	\caption{Top Left panel: XVA. Top Right panel: Value of a reference entity defaultable account. Medium Left panel: Value of the investor defaultable account. Medium Right panel: Value of the counterparty defaultable account. Bottom Left panel: Value of the funding account.}
	\label{fig:a210}
\end{figure}

Figure~\ref{fig:a210} shows that as the idiosyncratic component of the counterparty's {default} intensity $a_{20}$ increases, the XVA decreases in absolute value.
This may be explained by the fact that, as the default of the counterparty becomes more likely, all costs associated with the replicating portfolio (including
those from funding the position and remunerating collateral) will be incurred for a shorter period of time. {As a result the size of the XVA jump at
closeout decreases, which in turn results in fewer shares of investor and her counterparty's defaultable account. Moreover, because the counterparty's default becomes more likely, the XVA for a CDS with fewer reference entities will also decrease but a slower rate, and hence the number of shares of the reference entity account will increase (see also \eqref{eq:rob_strat1-mult}).}\footnote{Observe that $\xi^i$ must be the same for all $i\in\{1,\ldots,5\}$ because the account rates dynamics are the same for all reference entities.} In our specific setup, the value of the CDS portfolio $\hat{V}_0$ is rather small (equal to 0.00483043), the transaction is collateralized at $80\%$ of its market value, and the loss given default rates of the investor and counterparty are identical. Then, the jump to closeout value 
is entirely determined by the value of the replicating portfolio immediately before the closeout time. This explains why the number of shares in the investor and counterparty's account is approximately the same.

\subsubsection{Sensitivity of counterparty's default intensity to portfolio credit risk}

Figure~\ref{fig:a23} illustrates the dependence of XVA and replication strategies on $a_{23}$, i.e., the parameter quantifying the sensitivity of counterparty default intensity to credit risk of the underlying portfolio.  Let us start observing that the $XVA$ is negative. Moreover, $\hat{V}_0$ is positive and small, thus $\tilde{\theta} _{C}(\hat{v},m)-\check u\ge 0$. Correspondingly, the upper XVA always uses the rate ${\overline \mu_C}$ in the replication strategy, while the lower XVA always uses the rate ${\underline \mu_C}$. Notice that in our model specification the lower bond $\underline{\mu}_C$ is independent of $a_{23}$ while the upper bound $\bar{\mu}_C$ increases linearly with $a_{23}$. This explains why the upper XVA varies with $a_{23}$, while the lower XVA is constant with respect to it.
 Consistently with the graph in Figure~\ref{fig:a210}, an increase in $a_{23}$ raises the likelihood that the transaction will terminate earlier. Hence, the financing costs of the transaction will be smaller, and the XVA decreases.

 \begin{figure}
 	\includegraphics[width=0.49\textwidth]{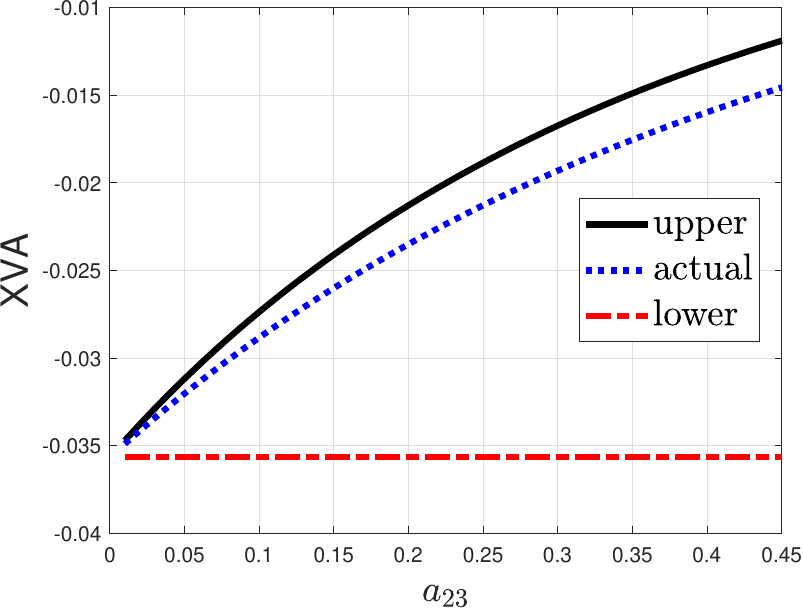}
 	\includegraphics[width=0.49\textwidth]{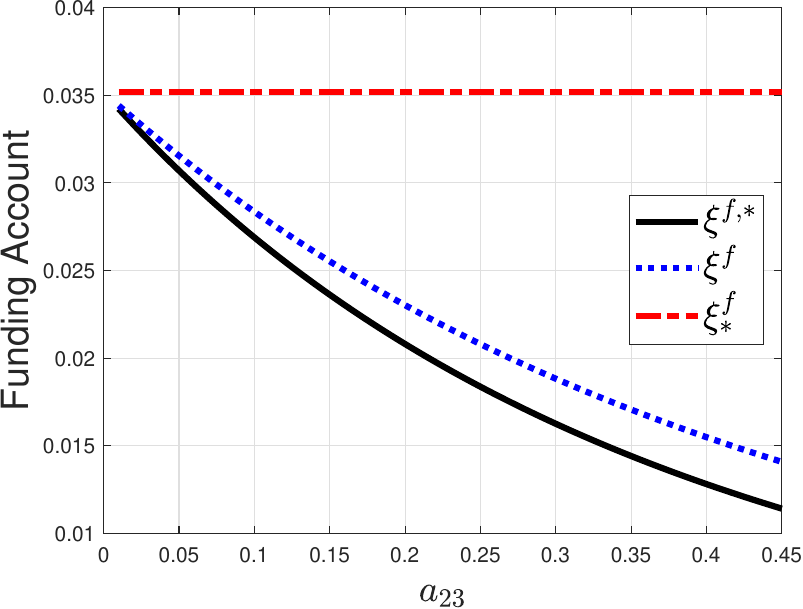}
 	\caption{Left panel: XVA. Right panel: Value of the funding account.}
 	\label{fig:a23}
 \end{figure}

\subsubsection{Sensitivity of portfolio credit risk to contagion}

Figure~\ref{fig:a33} illustrates the dependence of XVA and replication strategies on $a_{33}$, i.e., the parameter quantifying the impact of credit contagion on
the default risk of the reference entities in the portfolio. To better explain the results, we also illustrate the sensitivity of $\hat V_0$ with respect to
$a_{33}$. 

We start by observing that $\hat{V}_0$ is increasing in $a_{33}$. Because we view the CDS payoff from the payer's perspective and the spread premium $S$ is fixed, the moneyness of the contract increases if the default risk of the portfolio increase, which is the case if direct contagion effects are stronger. Moreover, it is easily seen from the expression~\eqref{eq:rob_strat1-mult} that $\xi^I$ increases with $\hat{V}$ while $\xi^C$ is decreasing in $\hat{V}$. Altogether, this implies the number of shares in the investor defaultable account is higher than the corresponding number of shares in the counterparty defaultable account. From a financial perspective, this can be understood in terms of DVA and CVA. As the investor is in the money because $\hat{V}>0$, he would need to additionally replicate the DVA benefit $L^I (1-\alpha) \hat{V}^+$ at his own default time. By contrast, he does not need to replicate any CVA loss at the counterparty default time because $L^I (1-\alpha) \hat{V}^- = 0$.

\begin{figure}
	\includegraphics[width=0.49\textwidth]{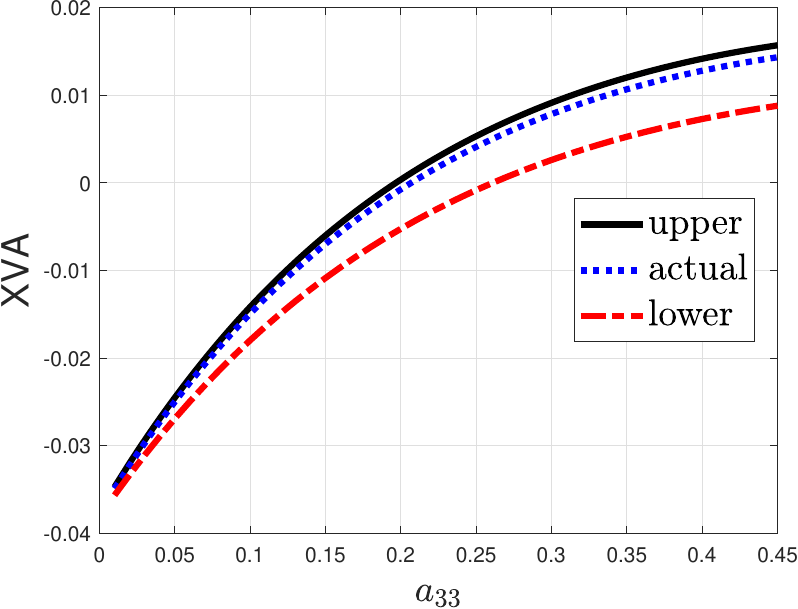}
	\includegraphics[width=0.49\textwidth]{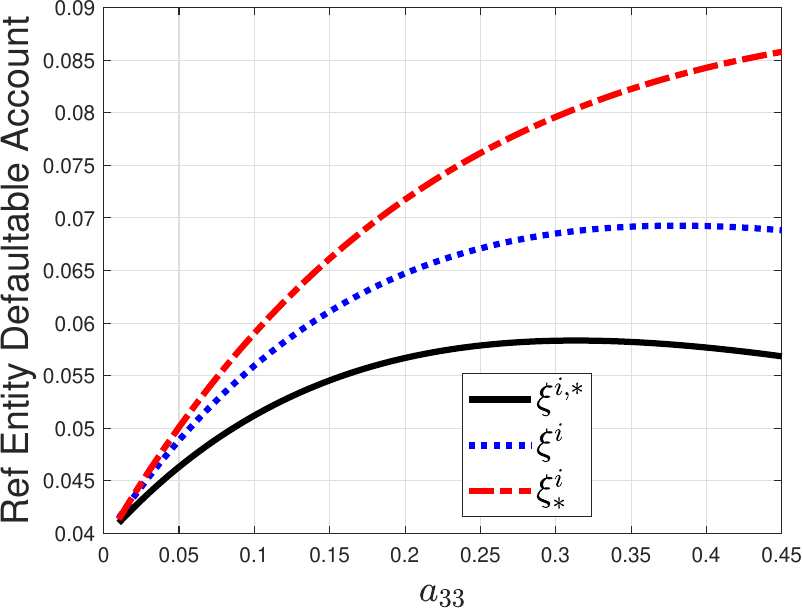}
	\includegraphics[width=0.49\textwidth]{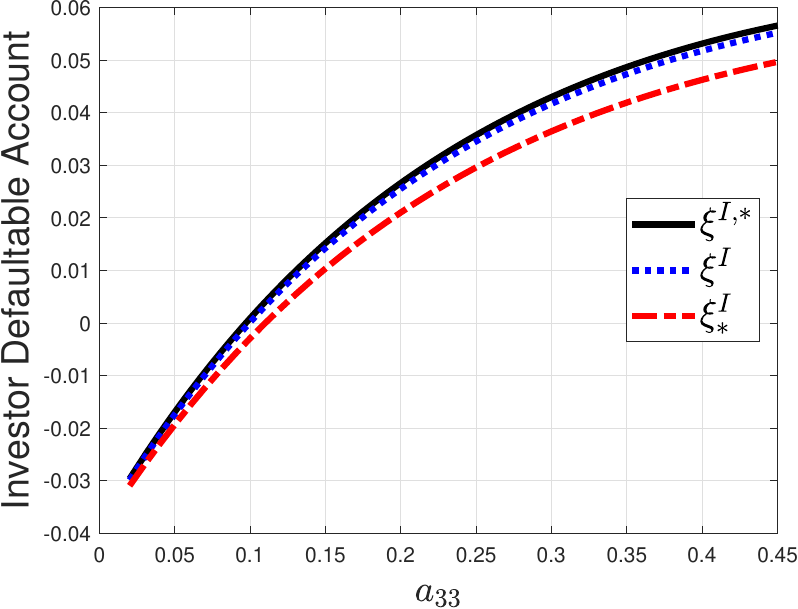}
	\includegraphics[width=0.49\textwidth]{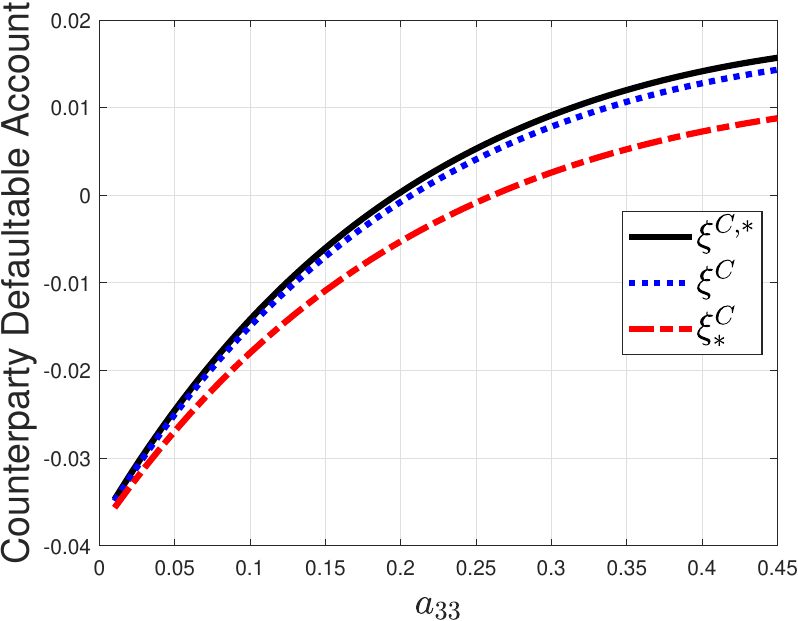}
	\includegraphics[width=0.49\textwidth]{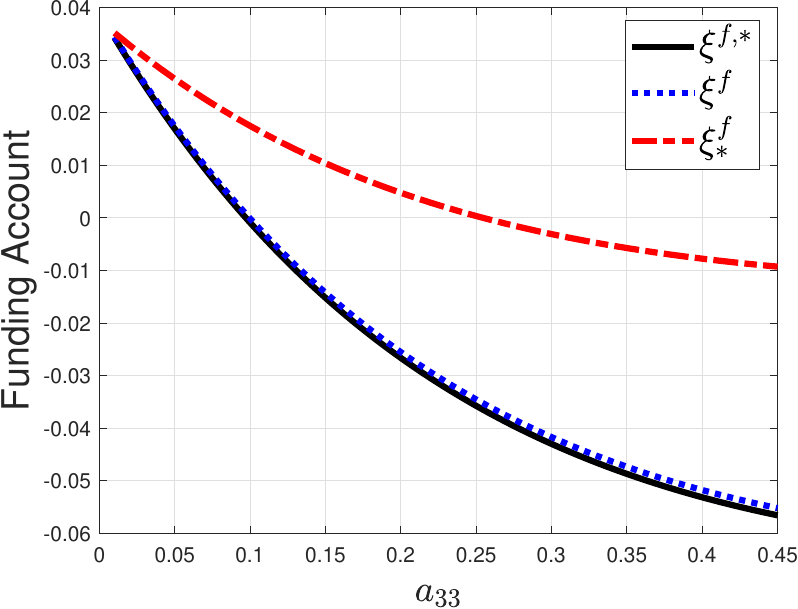}
	\includegraphics[width=0.49\textwidth]{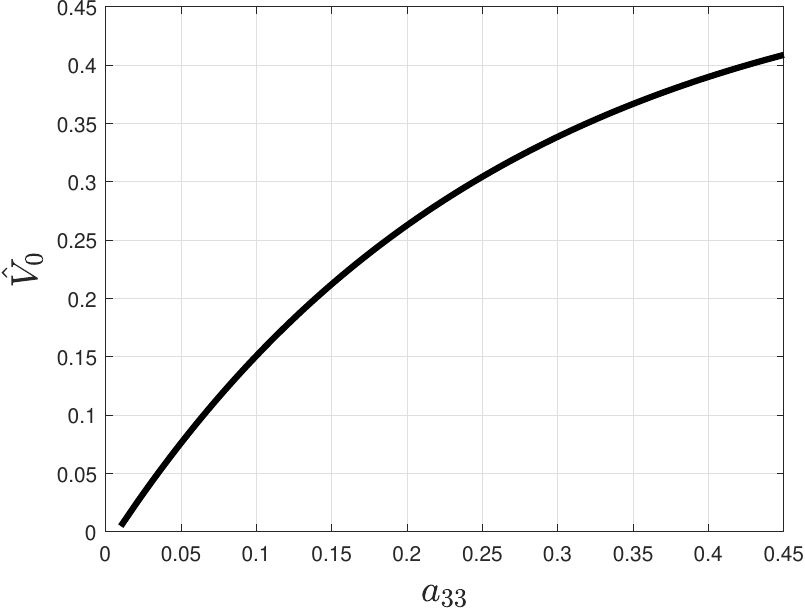}
	\caption{Top Left panel: XVA. Top Right panel: Value of a reference entity defaultable account. Medium Left panel: Value of the investor defaultable account. Medium Right panel: Value of the counterparty defaultable account. Bottom Left panel: Value of the funding account. Bottom Right Panel: Value of $\hat{V}_0$.}
	\label{fig:a33}
\end{figure}

The top left graph of Figure~\ref{fig:a33} highlights the prominent role played by default contagion. As $a_{33}$ increases, the default intensity of all reference entities in the portfolio increases. This in turn has an {\it indirect} effect on the default intensity of the investor and of the counterparty (through the coefficients $a_{13}$ and $a_{23}$ respectively), because both intensities jump upward if any of the five reference entities in the portfolio defaults. This indirect effect on the counterparty default intensity (magnified by a factor equal to five, i.e., equal to the number of entities in the portfolio) is higher than the direct effect resulting from increasing $a_{23}$; compare the top left graph of Figure~\ref{fig:a33} with the corresponding graph in Figure~\ref{fig:a23}. Because of this amplification effect created by the default contagion, the replicating strategy associated with the reference entity defaultable account has a more concave, rather than linear, dependence on $a_{23}$. This may be understood as follows:
{while initially, as $a_{33}$ increases, the XVA of a portfolio with five reference entities increases faster than the XVA of a portfolio with four reference
entities, eventually, as $a_{33}$ becomes high enough, the contagion among the reference entities in the portfolio is much higher than the contagion effect
imposed by the default of reference entities on the investor or her counterparty. This in turn means that if
one reference entity were to default, other reference entities will likely default shortly after (default clustering is strong).
Thus, the XVA will be approximately the same regardless the number of reference entities in the CDS portfolio. This induces
a decrease in the number of shares of reference entities account.}

\subsubsection{Idiosyncratic component of portfolio credit risk}
As $a_{30}$ increases, the idiosyncratic default risk of each reference entity gets higher. This increases the variance in the number of defaulting entities, and thus results in a larger uncertainty on the CDS payoff. As a result, the difference between lower and upper XVA increases, as confirmed from the top left panel of Figure~\ref{fig:a30}. A higher value of $a_{30}$ implies a higher portfolio credit risk, and thus a larger value of $\hat{V}_0$. However, $\hat{V}_0$ is less sensitive to changes in $a_{30}$ than to changes in $a_{33}$. In the latter case, there is an amplification effect due to increased contagion. 
This in turn implies that investor's defaultable account is less sensitive to $a_{30}$ than to $a_{33}$. A direct comparison of the investor account shares in the graphs of Figures~\ref{fig:a33} and~\ref{fig:a30} visually confirms this statement.

\begin{figure}
	\includegraphics[width=0.49\textwidth]{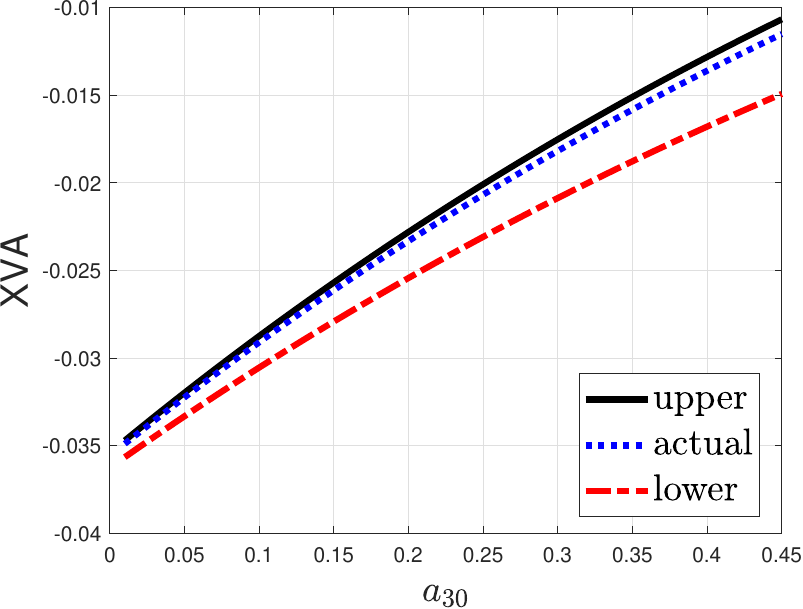}
	\includegraphics[width=0.49\textwidth]{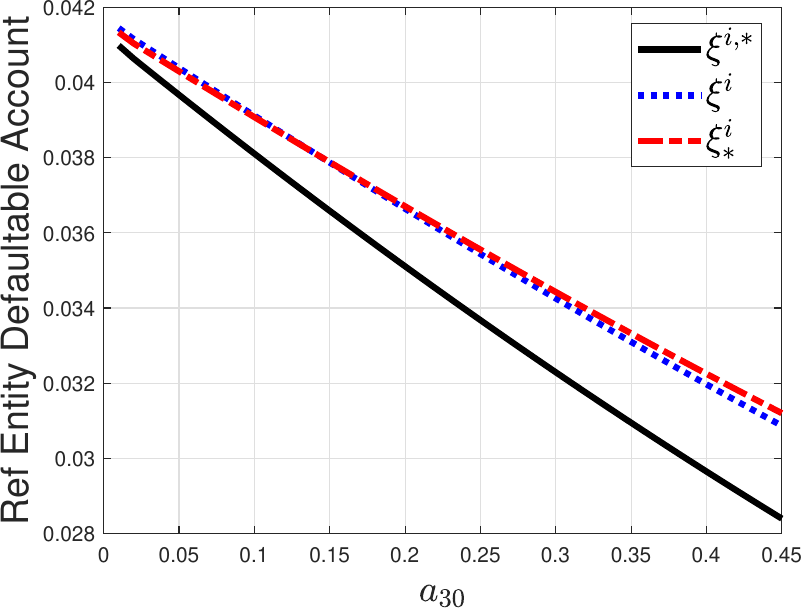}
	\includegraphics[width=0.49\textwidth]{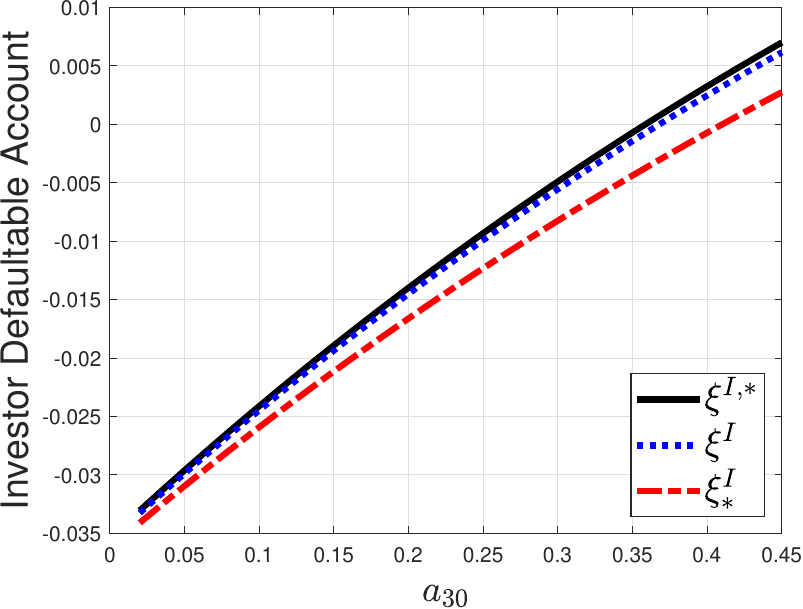}
	\includegraphics[width=0.49\textwidth]{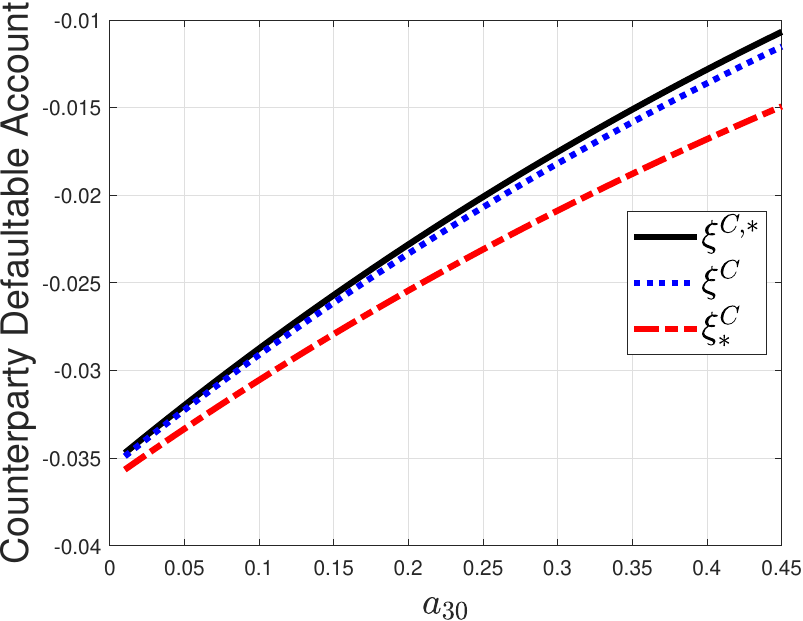}
	\includegraphics[width=0.49\textwidth]{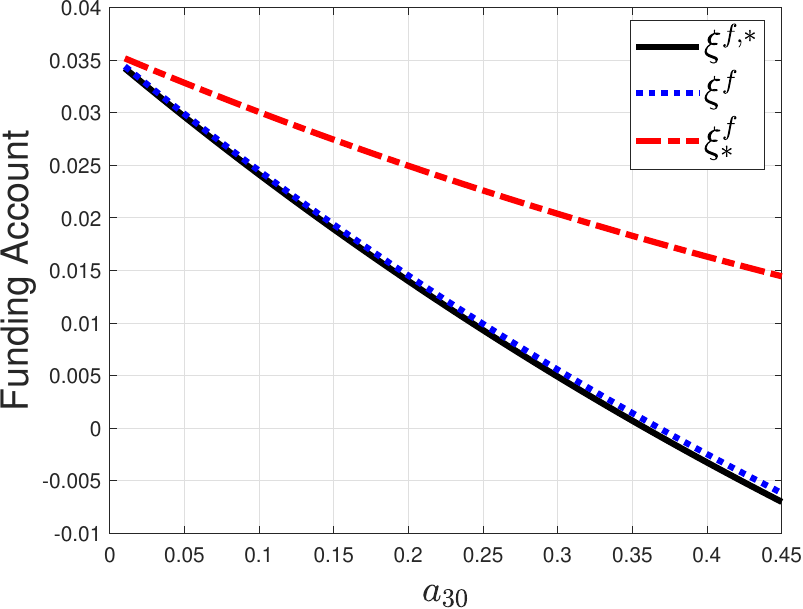}
	\includegraphics[width=0.49\textwidth]{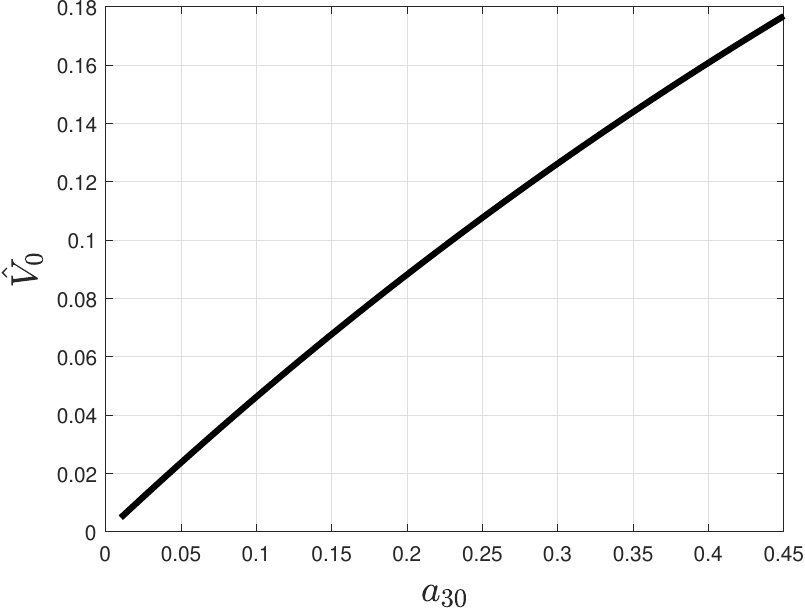}
	\caption{Top Left panel: XVA. Top Right panel: Value of a reference entity defaultable account. Medium Left panel: Value of the investor defaultable account. Medium Right panel: Value of the counterparty defaultable account. Bottom Left panel: Value of the funding account. Bottom Right Panel: Value of $\hat{V}_0$.}
	\label{fig:a30}
\end{figure}

\section{Conclusions}\label{sec:conclusions}
We have developed a framework to calculate the robust XVA of a credit default swap portfolio. We have considered the situation where the trader faces uncertainty on the return rate of the money market account associated with the counterparty. The credit default swap portfolio is replicated by the investor using defaultable accounts associated with the same entities referencing the single name credit default swap contracts in the portfolio. By constraining the return rate of the counterparty account to lie within an uncertainty interval, we have derived lower and upper bounds for the XVA. Our analysis highlights a nontrivial interaction between the value process of the trade that accounts for all financing costs, and the closeout process that depends only on the clean price of the transaction. The latter is obtained by pricing the cash flow of the trade, ignoring all other costs involved.
Our comparative statics analysis highlights the nontrivial role played by credit contagion on XVA and on the corresponding replication strategies.
Higher portfolio credit risk increases the value of the underlying portfolio for the CDS payer, and results in a larger number of reference entity account
shares. If the counterparty is more sensitive to default events of the underlying portfolio, XVA is lower because the transaction terminates earlier and is
less costly to implement.

Our framework constitutes a first step towards understanding the impact of model uncertainty on XVA. In a future continuation of the study,
we would like to explore the impact of default contagion on gap risk, which may build up during the margin period of risk. 
Prices of credit default swaps are especially affected by gap risk, because unlike other
swap contracts whose payoff is not credit sensitive (e.g. interest rate swaps), the mark-to-market value of a CDS jumps at the default time of the investor or her
counterparty. 
The calculation of initial margins should then account for such a risk. The development of an initial margins formula that takes into account contagion effects is, however, far from trivial because of the complex dependence structure of  defaults in the underlying portfolio. We also plan to generalize our framework to deal with uncertainty in the portfolio credit risk. While individual firms' default probabilities may be estimated from single name CDS spreads,
default correlation is hardest to estimate and subject to model risk. We leave all these questions for future research.

\section*{Acknowledgment}

The authors are grateful to two anonymous referees for valuable comments and suggestions, which contributed to improve and enrich the paper.


\appendix
\section{Construction of $\mathbb{F}$ Filtration  and Default Intensity Processes} \label{sec:defintconstr}
We use the following stepwise  procedure: Assume $h^{i,\Px, 0}_t \in \mathcal{F}^0 \otimes \mathcal{B}([0,t))$ and define $\tau^{i,0} := \sup\bigl\{ t\ge0 \colon \int_0^t h^{i,\Px, 0}_sds > \mathcal E^i\bigr\}$. Then, we can define $\mathcal F_t^1 := \sigma\bigl( H^j_u; u \leq t \wedge \tau^0_{(1)} \colon j\in\{1, .., N\} \bigr)$, where $\tau^0_{(1)}$ is the time of the first default (the first order statistics). For $k \geq 1$, choose $\tilde{h}^{i,\Px, k}_t \in \mathcal{F}_t^k \otimes \mathcal{B}([0,t))$,  and define recursively $h^{i,\Px, k}_t := h^{i,\Px, k-1}_t {\bf 1}_{[0,\tau^{k-1}_{(k)})}(t) + \tilde{h}^{i,\Px, k}_t {\bf 1}_{[\tau^{k-1}_{(k)},\infty)}(t)$, where we use the notation $\tau^k_{(i)}$ to denote the $i$-th order statistics of the $k$-level stopping time $\tau_i^k$, and ${\bf 1}_{A}(t)$ is the characteristic function which equals one if $t \in A$ and zero otherwise. Then, we define $\tau_i^k := \sup\bigl\{ t\ge0 \colon \int_0^t h^{i,\Px, k}_sds > \mathcal E^i\bigr\}$ as well as $\mathcal F_t^{k+1} := \sigma\bigl( H^j_u; u \leq t \wedge \tau^k_{(k+1)} \colon j\in\{1, .., N\} \bigr)$ for $k \in \{ 1, \ldots, N\}$. In this way, the intensity $h^{i,\Px,k}$ agrees with $h^{i,\Px,k-1}$ up to the $k$-th default, but accounts for information after the $k$-th default thereafter. Finally, we define the full filtration $\mathbb{F}$ as $\mathbb{F} = \bigl(\mathcal{F}_t^{N+1}\bigr)_{t\geq 0}$.}

\section{Proofs of lemmas and propositions}
\begin{proof} [{\bf Proof of Eq.~\eqref{eq:creditswap}.}]
First, observe that the linear BSDE~\eqref{eq:BSDE-hatV} admits the solution given by
\begin{align}
	\hat{V}_t = \hat{C}^1_t &= -\Exx^{\Qxx} \biggl[ \int_t^{\tau^1 \wedge T} e^{-r_D(u-t)}\eta^1 \, du - L^1 e^{-r_D(\tau^1-t)} \ind_{\tau^1 \leq T} \, \bigg\vert \, \mathcal{F}_t\biggr] \ind_{\{t \leq \tau^1\}}. \nonumber\\
\end{align}
Moreover, as the default distribution is characterized by $\Qxx[\tau^1 > u] = e^{-\int_t^u h^{1,\Qxx}_sds}$, on the same event $\{t \leq \tau^1\}$ we have that
\begin{align}
	\hat{V}_t &=- \Exx^{\Qxx}\biggl[ \int_t^{ T} \ind_{\{u \ge \tau^1\}}  e^{-r_D(u-t)}\eta^1 \, du - \int_t^T L^1 e^{-r_D(u-t)} h^{1,\Qxx}_u e^{-\int_t^u h^{1,\Qxx}_sds} du \, \bigg\vert \, \mathcal{F}_t\biggr]\ind_{\{t \leq \tau^1\}} \nonumber\\
	&= -\Exx^{\Qxx}\biggl[ \int_t^{ T}  e^{-\int_t^u (h^{1,\Qxx}_s  +r_D) ds} \eta^1 \, du - \int_t^T L^1 h^{1,\Qxx}_u e^{-\int_t^u (h^{1,\Qxx}_s + r_D) ds} du \, \bigg\vert \, \mathcal{F}_t\biggr]\ind_{\{t \leq \tau^1\}}.
\end{align}
\end{proof}

\begin{proof} [{\bf Proof of Proposition \ref{prop:arb-market}.}]

	To facilitate the no-arbitrage argument, we will express the wealth process under a suitable measure $\tilde{\Px}$ specified via the stochastic exponential
	\[
	\frac{d\tilde{\Px}}{d\Px} = \prod_{i \in \{1, \ldots, N, I, C\}} \Biggl(\frac{(\mu^i - \rfp)(\tau \wedge \tau^{(N)})}{\int_0^{\tau \wedge \tau^{(N)}} h^{i,\Px}_s \, ds}\Biggr)^{H^i_\tau} \exp{\biggl(\int_0^{\tau \wedge \tau^{(N)}}(\rfp-\mu^i + h^{i,\Px}_s)\, ds}\biggr).
	\]
	By Assumption \ref{ass:nec+suff} this change of measure is well defined.
	Moreover, while the measure $\tilde{\Px}$ is unknown to the investor, there is no issue with using it from an abstract point of view to rule out arbitrage. By Girsanov's theorem, the dynamics of the risky assets are given by
	\begin{align}
	dB_t^i & = \rfp B_t^i \, dt - B_{t-}^i d\varpi_t^{i,\tilde{\Px}}
	\end{align}
	for $i \in \{1, \ldots, N,I,C\}$ where $\varpi^{i,\tilde{\Px}}:=(\varpi_t^{i,\tilde{\Px}}; \;0 \leq t \leq \tau)$ are $(\mathbb{G},\tilde{\Px})$-martingales. The $\rfp$ discounted assets $\tilde{P}_t^i := e^{-\rfp t} B_t^C$ are thus $(\mathbb{G},\tilde{\Px})$-martingales. In particular, the default intensities under $\tilde{\Px}$ are given by $h^{i,\tilde{\Px}} = \mu^i-\rfp$, which are positive by Assumption \ref{ass:nec+suff}.
	
	Denote the wealth process associated with $(B_t^i; i \in\{1, \ldots, N,I,C\})_{t \geq 0}$ in the underlying market by $\check{V}_t$. Using the self-financing condition, its dynamics are given by
	\begin{align}
	\nonumber d\check{V}_t &= r_f \xi_t^f B_t^{r_f}  \, dt + \sum_{i \in\{1, \ldots, N,I,C\}} \rfp  \xi_t^i  \, dB_t^i \\
	&= \biggl(r_f \xi_t^f B_t^{r_f} + \sum_{i \in\{1, \ldots, N,I,C\}} \rfp  \xi_t^i B_t^i \biggr) \, dt - \sum_{i \in\{1, \ldots, N,I,C\}}  \xi_t^i B_t^i \, d\varpi_t^{i,\tilde{\Px}} .
	\end{align}
	Then we observe that $r_f \xi_t^f  \leq  \rfp \xi_t^f$ and thus
	\begin{align}
	\check{V}_\tau({\bm\varphi},x) - \check{V}_0({\bm \varphi},x) & = \int_0^\tau (r_f \xi_t^f B_t^{r_f} + \sum_{i \in\{1, \ldots, N,I,C\}} \rfp  \xi_t^i B_t^i  \bigr) \, dt - \sum_{i \in\{1, \ldots, N,I,C\}}  \int_0^\tau {\xi_t ^i} B_{t-}^i\, d \varpi_t^{i,\tilde{\Px}} \\
	&  \leq \int_0^\tau (\rfp \xi_t^f B_t^{r_f} + \sum_{i \in\{1, \ldots, N,I,C\}} \rfp  \xi_t^i B_t^i  \bigr) \, dt - \sum_{i \in\{1, \ldots, N,I,C\}}  \int_0^\tau {\xi_t ^i} B_{t-}^i\, d \varpi_t^{i,\tilde{\Px}}  \\
	& = \int_0^\tau \rfp \check{V}_t({\bm\varphi},x) \, dt + \sum_{i \in\{1, \ldots, N,I,C\}} \int_0^\tau\rfp  \xi_t^i  \, dB_t^i.
	\end{align}
	Therefore, it follows that
	\[
	e^{-\rfp\tau} \check{V}_{\tau}({\bm \varphi},x) -\check{V}_0({\bm \varphi},x) \leq \sum_{i \in\{1, \ldots, N,I,C\}}\int_0^\tau \rfp  \xi_t^i  \, d\tilde{P}_t^i.
	\]
	Note that the right hand side of the above inequality is a local martingale bounded from below (as the value process is bounded from below by the admissibility condition), and therefore is a supermartingale. Taking expectations, we conclude that
	\[
	\Exx^{\tilde{\Px}} \bigl[  e^{-\rfp\tau} \check{V}_\tau({\bm \varphi},x) - \check{V}_0({\bm \varphi},x) \bigr] \leq 0.
	\]
	Thus either $\tilde{\Px} \bigl[\check{V}_\tau({\bm \varphi},x) = e^{\rfp\tau} x \bigr] = 1$ or $\tilde{\Px} \bigl[ \check{V}_\tau({\bm \varphi},x) < e^{\rfp\tau} x \bigr] >0 $. As $\tilde{\Px}$ is equivalent to $\Px$, this shows that arbitrage opportunities for the investor are precluded in this model (he would receive $e^{\rfp \tau}x$ by lending the positive cash amount $x$ to the treasury desk at the rate $\rfp$). \hfill

\end{proof}

\begin{proof} [{\bf Proof of Proposition \ref{thm:reduction}.}]
	
	As the filtration $\mathbb{F}$ is trivial, the $\mathbb{F}$-BSDE is in fact an ODE. The existence and uniqueness to this ODE is shown in Proposition \ref{thm:existence}. The equivalence of the full $\mathbb{G}$-BSDEs and the reduced $\mathbb{F}$-BSDEs follows from the projection result \cite[Theorem 4.3]{CrepeyRed} as condition (A) in their paper is satisfied by our assumptions on the filtrations and their Condition (J) is also satisfied (as the terminal condition does not depend on $\tilde{Z}$, $\tilde{Z}^I$ and $\tilde{Z}^C$). Finally, by the martingale representation theorems with respect to $\mathbb{F}$ and $\mathbb{G}$ (see \cite[Section 5.2]{bielecki01}; their required assumptions are satisfied because our intensities are bounded), the solution of our BSDEs and those of the martingale problems considered in \cite{CrepeyRed} coincide.\hfill
	
\end{proof}

\begin{proof} [{\bf Proof of Proposition \ref{thm:existence}.}]
{The existence and uniqueness of a solution  to ODE \eqref{eq:hatV-ODE} on the time interval $[0,T]$ follows from the classical Picard-Lindel\"{o}f Theorem, together with Corollary II.3.2 of \cite{hartman}.

We now show existence and uniqueness of a solution to ODE  \eqref{eq:creditswap1}.} The existence again follows from the classical Picard-Lindel\"{o}f Theorem on every continuity interval of $h_i$'s. For simplicity of exposition we will assume that all  $h_i$'s are continuous on $[0,T]$. In case, of a discontinuity, the solution will not be differentiable there, but will remain continuous.

First note that $\check u$ is bounded. To see this, observe that $ \check g$ is Lipschitz in its second argument, and $\abs{\check{g}(t,0; \hat v_t, m_t )}\le K_0$ is uniformly bounded, by possibly increasing the constant $K_0$ if needed. It thus follows that
\begin{align}
\abs{\check{g}(t,\check u; \hat v_t, m_t)} \le \abs{\check{g}(t,\check u_t; \hat v_t, m_t) - \check{g}(t,0; \hat v_t, m_t)} + \abs{\check{g}(t,0; \hat v_t, m_t)} \le K_0  \abs{\check u} + K_0.
\end{align}
Then, assuming, $\check u$ is differentiable, we can employ Gronwall inequality and deduce that if
\begin{align}
\partial_t \check u_t &\le  K_0\check u_t + K_0,\\
\check u_t &=0,\nonumber
\end{align}
then $\check u_t \le K_1 \define K_0 T e^{K_0 T},$ for $t\in[0,T]$. Similar for the lower bound, if
\begin{align}
\partial_t \check u_t &\ge  -K_0\check u_t - K_0\\
\check u_0 &=0,\nonumber
\end{align}
from which it follows that $\check u_t \ge -K_1.$

It remains to check that the differentiability condition needed for the Gronwall inequality. By the classical Picard-Lindel\"{o}f Theorem, the solution to \cref{eq:reduced} exists on some interval $[0, T_0),\times[-K_1-1, K_1+1]$, that is, for $t\in[0,T_0)$ it holds that $\abs{\check u_t} \le K_1+1$, and it is unique there. Hence, we are guaranteed differentiability in this interval. Assume by contradiction that it cannot be extended (to the right) beyond $T_0$ and that $T_0<T$ (the same argument applies, if $T_0=T$, but the solution cannot be extended to the closed interval $[0,T]$). Then by Corollary II.3.2 of \cite{hartman} we have that $\lim\limits_{t\to T_0} \abs{\check u_t} = K_1+1.$ We now reach a contradiction, by employing Gronwall inequality argument above that shows that $\abs{\check u} \le K_1$.
\hfill

\end{proof}

\begin{proof} [{\bf Proof of Theorem~\ref{thm:comp}}]
First, note that similar to the proof of Proposition \ref{thm:existence}, the functions $\check{u}^{*}$ and $\check{u}_{*}$, defined as the solutions to the ODEs in~\eqref{eq:creditswap2} exist and are unique. This follows from the fact that the functions $g^*$ and $g_*$ are Lipschitz continuous in all arguments.

Assume, by contradiction, that there exists $T_0\le T$ for which
$\check{u}^{*}_{T_0}< \check u_{T_0},$ and set
$T_1 = \sup\left\{ t\le T_0 \, \vert \,\right.$ $\left.\check{u}^{*} (t)\ge \check u_t\right\}.$ We have that $T_1$ is well defined, and $T_1\ge0$, because $\check{u}^{*}_0= \check u_0 =0$ and $\check{u}^{*}_t< \check u_t$ for $t\in( T_1, T_0).$
Using the facts that {$\underline\mu^C>r_D$} and that
$(\mu^{C,\Qxx})^{*}(\hat{v}, m, \check u)(\tilde\theta^C(\hat v,m) - \check u) \ge {\mu^{C,\Qxx}_t} (\tilde\theta^C(\hat v,m) - \check u)$ for any $t\in[0,T]$, we have that
\begin{align}
&\partial_t \check{u}^{*}(T_1) = g^{*}(T_1,\check{u}^{*}; \hat v,m) \label{eq:PDE-compare-ineq} \\
 &\qquad=h^{I,\Qxx}\bigl(\tilde{\theta} ^{I}(\hat{v}_{T_1}, m_{T_1})-\check{u}^{*}_{T_1}\bigr)- h^{1,\Qxx} \check{u}^{*}_{T_1} \nonumber \\
&\qquad+{( (\mu^C)^{*}(\hat v_{T_1},m_{T_1}),\check u^{*}_{T_1}) - r_D)} \bigl(\tilde{\theta} _{C}(\hat{v}_{T_1}, m_{T_1})-\check{u}^{*}_{T_1}\bigr) \nonumber\\
& \qquad+ \tilde{f}\bigl(T_1, \check{u}^{*}_{T_1}, -\check{u}^{*}_{T_1}, \tilde{\theta}^{I}( \hat{v}_{T_1}, m_{T_1})-\check{u}^{*}_{T_1}, \tilde{\theta} ^{C}( \hat{v}_{T_1}, m_{T_1}) -\check{u}^{*};\hat{v}_{T_1}, m_{T_1}\bigr) \nonumber\\
&\quad\ge h^{I,\Qxx}\bigl(\tilde{\theta} ^{I}(\hat{v}_{T_1}, m_{T_1})-\check{u}^{*}_{T_1}\bigr) +h^{C,\Qxx}_{T_1} \bigl(\tilde{\theta} ^{C}(\hat{v}_{T_1}, m_{T_1})-\check{u}^{*}_{T_1}\bigr) - h^{1,\Qxx} \check{u}^{*}_{T_1} \nonumber\\
&\qquad+ \tilde{f}\bigl(T_1, \check{u}^{*}_{T_1}, -\check{u}^{*}_{T_1}, \tilde{\theta}^{I}( \hat{v}_{T_1}, m_{T_1})-\check{u}^{*}_{T_1}, \tilde{\theta}^{C}( \hat{v}_{T_1}, m_{T_1}) -\check{u}^{*};\hat{v}_{T_1}, m_{T_1}\bigr)\nonumber\\
&\quad= \check{g}(T_1,\check{u}; \hat v_{T_1}, m_{T_1}) dt = \partial_t \check{u}_{T_1}.\nonumber
\end{align}
It follows that there exists an $\epsilon>0$, such that  $\check{u}^{*}_t\ge \check u_t$ for $t\in[T_1, T_1+\epsilon].$ This contradicts the assumption, and proves the theorem. \hfill
\end{proof}

\begin{proof}[{\bf Proof of Theorem~\ref{thm:robust}}]
	The proof consists of two parts. In the first part, we verify that the expression of $\rXVA$ given in Eq.~\eqref{eq:rXVA} is the smallest super-replicating price. In the second part, we show that the strategy given in \eqref{eq:rob_strat1} is a super-replicating strategy. This requires showing that the implementation of this strategy does not require any cash infusion, and that the wealth process controlled by this strategy is exactly the $\rXVA$ process.

	{Define $\XVA_t^\mu := V_t^\mu - \hat{V}_t$ for $\mu \in \mathbb{F}$, $\underline{\mu}^C \leq \mu \leq \overline{\mu}^C$ and $\XVA_t^* := V_t^{(\mu^C)^*} - \hat{V}_t$, where we recall that $V_t^\mu$ is the valuation process of the replicating portfolio obtained by setting the counterparty {account} rate equal to $\mu$; see also the discussion before Eq.~\eqref{eq:rXVAdef}. First, note that $\XVA_t^* \geq \XVA_t^\mu$. This follows directly from Theorem \ref{thm:comp}, which provides a comparison result for the term $\check{U}_t$ appearing on the right hand side of the $\XVA$ expression~\eqref{eq:reduced_identity2}. Therein, it is enough to observe that the risk-neutral default intensity is just $\mu - r_D$, and notice that the two closeout terms are just independent of the rate $\mu$. Hence, the right hand side of Eq.~\eqref{eq:rXVA} is smaller than the left hand side: the latter represents a specific $\mathbb{F}$-predictable intensity process satisfying the boundary conditions, while the former is the supremum over all such intensity processes. This shows that the left side of Eq.~\eqref{eq:rXVA} is less or equal than the right side. To show the reverse inequality, i.e., that the left side of Eq.~\eqref{eq:rXVA} is greater or equal than the right side, we note that the family $(\XVA_t^\mu)_{\mu \in \mathbb{F}, \mu \in [\underline{\mu}^C, \overline{\mu}^C]}$ is directed upwards, i.e., for $\mu', \, \mu'' \in \mathbb{F}, \underline{\mu}^C \leq \mu', \mu'' \leq \overline{\mu}^C$, there exists a process $\mu''' \in \mathbb{F}, \underline{\mu}^C \leq \mu''' \leq \overline{\mu}^C$, such that $\XVA^{\mu'} \vee \XVA^{\mu''} \leq \XVA^{\mu'''}$. Indeed, setting $A := \{\omega \in \Omega \, : \, XVA^{\mu'}_t > XVA^{\mu''}_t\}$ we can define $\mu'''$ directly by setting  $\mu'''_s := \mu'_s \ind_A + \mu''_s \ind_{A^c}$ for $s \geq t$, and $\mu'''_s=0,$ for $0\le s<t$. Such a process is clearly $\mathcal{F}_s $-measurable because $A$ is $\mathcal{F}_t$-measurable. As the essential supremum of an upward directed set can be written as monotone limit (see \cite[Theorem A.32]{FoellmerSchied}), $\lim_{n \to \infty} \XVA^{\mu^{(n)}} = \rXVA$. Thus, as the countable union of nullsets is still a nullset we have that, for all $t$, $\rXVA_t$ is smaller or equal than the right side of Eq.~\eqref{eq:rXVA}.}
	
	Next, we provide the expressions for the super-replicating strategies. These are derived by replacing $\check{U}_t$ with $\check{U}_t^{*}$ into equations \eqref{eq:xi1}--\eqref{eq:xif}.
	{Using the replicating strategies defined in \eqref{eq:rob_strat1}, we obtain that,}  on the set $\{t<\tau\}$, the value of the replicating portfolio at time $t$ is
	\begin{align}
	\xi_u^{1,*} \, B_u^1+ \xi_t^{I,*} \, B_t^I + \xi_t^{C,*} \, B_t^C + \xi_t^{f,*} \, B_t^{r_f} -\psi_t^{m,*} \, B_t^{r_m} & =  \check{U}_t^{*}.
	\label{eq:super-hedge-value}
	\end{align}
	On the set $\{t<\tau\}$, the change in value of the portfolio is
	\begin{align}
	&\xi_u^{1,*}\,d B_u^1+ \xi_t^{I,*} \, d B_t^I + \xi_t^{C,*} \, d B_t^C + \xi_t^{f,*} \, d B_t^{r_f} -\psi_t^{m,*} \,d  B_t^{r_m} \label{eq:U*1}\\
	&= \Big( {\mu_1} \check{U}_t^*+ \left(L^I (\hat{V}_t - M_{t-})^{+} + \check{U}_t^*\right) {\mu^I} +  (\check{U}_t^{*}  -L^C (\hat{V}_t - M_{t-})^{-}) {\mu^C}  \nonumber\\
	& \ \ + r_m^{+}M_t^{+} + r_m^{-}M_t^{-}  + r_f^{+} (- 2 \check{U}_t^* + L^C (\hat{V}_t - M_{t-})^{-}- L^I  (\hat{V}_t -M_{t-})^{+} - M_{t-} )^{+} \nonumber\\
	& \ \ +  r_f^{-} (- 2 \check{U}_t^* + L^C (\hat{V}_t - M_{t-})^{-}- L^I  (\hat{V}_t -M_{t-})^{+} - M_{t-} )^{-} \Big)dt. \nonumber
	\end{align}
	
	{Additionally, for the replicating strategy \eqref{eq:xi1}--\eqref{eq:xif} to be self-financing, we need to include the cash flow}
	\begin{align}
	\Big(r_f(\xi_t^{f,*}) -r_D\Big)L^1dt. \label{eq:U*2}
	\end{align}
	{The presence of this cash flow is due to the fact that the clean valuation $\hat V$ is computed using the publicly available discount rate $r_D$, while the private valuation $V$ is obtained using the funding rate $r_f$. Such a cash flow needs to be accounted for in the implementation of the super-replicating strategy.} Taken together, equations~\eqref{eq:U*1} and \eqref{eq:U*2} describe the change in value of the super-replicating portfolio. Next, we compare it with the change in value of the robust XVA process given by 
	\begin{align}
	d \check{U}_t^{*} & = \Big((r_D+h^{I,\Qxx})\bigl(L^I (\hat{V}_t - M_{t-})^{+} + \check{U}_t^*\bigr) -(\mu^C)^{*}(\hat V_t ,M_t,\check U^{*}_t) \bigl(L^C (\hat{V}_t - M_{t-})^{-}- \check{U}_t^{*}  \bigr)  \nonumber\\
	&  + {\mu^1} \check{U}^{*}_t+ r_m^{+}M_t^{+} + r_m^{-}M_t^{-}+ r_f^{+}(- 2 \check{U}_t^* + L^C (\hat{V}_t - M_{t-})^{-}- L^I  (\hat{V}_t -M_{t-})^{+} +L^1- M_{t-})^{+}\nonumber\\
	&+ r_f^{-}(- 2 \check{U}_t^* + L^C (\hat{V}_t - M_{t-})^{-}- L^I  (\hat{V}_t -M_{t-})^{+} +L^1- M_{t-})^{-} -r_D L^1\Bigr) dt\label{eq:U*3}.
	\end{align}
	Using the fact that
	\begin{align}
	\left ({(\mu^C)^{*}(\hat V_t,M_t,\check U^{*}_t) - \mu^C } \right) \bigl(\tilde{\theta} ^{C}(\hat{V}_t, M_t)-\check{U}^{*}_t\bigr)\ge0,
	\end{align}
	{it follows that \eqref{eq:U*1} together with \eqref{eq:U*2} dominate \eqref{eq:U*3} from above, i.e.,}
	\begin{align}
	&\xi_u^{1,*}\,d B_u^1+ \xi_t^{I,*} \, d B_t^I + \xi_t^{C,*} \, d B_t^C + \xi_t^{f,*} \, d B_t^{r_f} -\psi_t^{m,*} \,d  B_t^{r_m} + \Big(r_f(\xi_t^{f,*}) -r_D\Big)L^1dt \ge d \check{U}_t^{*}.
	\label{eq:superhedge-ineq}
	\end{align}
	{The above computations were done on the set $\{t<\tau\}$. At the stopping time $\tau$ it can be easily checked that both $\check U^{}$ and the super-replicating portfolio are zero. Together with \eqref{eq:superhedge-ineq} and Theorem \ref{thm:comp}, it follows that the super-replicating portfolio dominates $\check{U}^{}$ for all times $t$.}

\end{proof}

\begin{proof} [{\bf Proof of Theorem~\ref{thm:comp-multi}}]
The proof of the super-replicating strategies is done by induction over $\abs{{\cal J}}$, i.e., the cardinality of $\mathcal{J}$. Without loss of generality, we may assume $\gamma\in\{1,-1\}$. {We present the proof for $\gamma=1$, as this is identical to the case $\gamma=-1$. For notation convenience, we drop the superscript $\gamma$.}
If $\abs{{\cal J}}=N-1$, the thesis follows directly from Theorem \ref{thm:robust}.{ By induction over the cardinality of $\abs{{\cal J}}$, assume that the result holds in case of when the entities in the set ${\cal J}$ have not defaulted yet, with $\abs{{\cal J}}=n+1\ge1$. Next, we prove the result for the case when the set ${\cal J}$ of entities that have not defaulted yet has cardinality $n$. Fix such a set ${\cal J}$ for which $\abs{{\cal J}}=n$.} Assume, by contradiction, that there exists $T_0\le T$ for which $\check{u}^{{\cal J},*}_{T_0}< \check u_{T_0}$, and set
$T_1 = \sup\left\{ t\le T_0 \vert \check{u}^{{\cal J},*} _t\ge \check u_t\right\}.$ Then, $T_1$ is well defined, and $T_1\ge0$ since $\check{u}^{{\cal J},*}_0= \check u_0 =0$ and $\check{u}^{{\cal J},*}_t< \check u_t$ for $t\in( T_1, T_0).$
Denote $Z^{f,{\cal J},*} =  \sum_{k \notin \cal{J}}  \check{u}^{\{k\}\cup \cal{J},*}-(N-\abs{{\cal J}}+1) \check{u}^{{\cal J},*}_{T_1} + \tilde{\theta}^{I}( \hat{v}_{T_1}, m_{T_1})  + \tilde{\theta} ^{C}( \hat{v}_{T_1},m_{T_1}) +   \sum_{k \notin \cal{J}}  L^k- M_t$, and similarly,
$Z^{f,{\cal J}} =  \sum_{k \notin \cal{J}}  \check{u}^{\{k\}\cup \cal{J}}-(N-\abs{{\cal J}}+1) \check{u}^{{\cal J},*}_{T_1} + \tilde{\theta}^{I}( \hat{v}_{T_1}, m_{T_1})  + \tilde{\theta}^{C}( \hat{v}_{T_1},m_{T_1}) +   \sum_{k \notin \cal{J}}  L^k- M_t$.

Using the facts that {$\underline\mu^C>r_D$} and that  $(\mu^{C,\Qxx})^{*}(\hat{v}, m, \check u)(\tilde\theta^C(\hat v,m) - \check u) \ge \mu^{C,\Qxx}_t (\tilde\theta_C(\hat v,m) - \check u)$ for any $t\in[0,T]$ we have that
\begin{align}
& \partial_t \check u^{ {\cal J} ,*} = g^{*}\Bigl(T_1,\check u^{ {\cal J} ,*},  \sum_{k \notin \cal{J}}  \check{u}^{\{k\}\cup \cal{J},*},  \sum_{k \notin \cal{J}} h^{k,\Qxx}  \check{u}^{\{k\}\cup \cal{J},*}; \hat v^{{\cal J}}, m, {\cal J}\Bigr) \label{eq:PDE-compare-ineq-multi}\\
&\quad=h^{I,\Qxx}_{T_1}\bigl(\tilde{\theta} ^{I}(\hat{v}_{T_1}, m_{T_1})-\check{u}^{{\cal J},*}_{T_1}\bigr) + \sum_{i\notin {\cal J}}h^{i,\Qxx} \left(\check{u}^{ \{i\}\cup \cal{J},*}_{T_1} - \check{u}^{ {\cal J} ,*}_{T_1}\right)\nonumber\\
&\qquad+{( (\mu_C)^{*}(\hat v_{T_1},m_{T_1},\check u^{ {\cal J} ,*}_{T_1}) - r_D)}  \bigl(\tilde{\theta}^{C}(\hat{v}_{T_1}, m_{T_1})-\check{u}^{{\cal J},*}_{T_1}\bigr) \nonumber\\
& \qquad+ \tilde{f}\bigl(T_1, \check{u}^{{\cal J},*}_{T_1},  \sum_{k \notin \cal{J}}  \check{u}^{\{k\}\cup \cal{J},*}-(N-\abs{{\cal J}}) \check{u}^{{\cal J},*}_{T_1}, \tilde{\theta}^{I}( \hat{v}_{T_1}, m_{T_1})-\check{u}^{{\cal J},*}_{T_1},  \nonumber\\
&\qquad\qquad \tilde{\theta}^{C}( \hat{v}_{T_1},m_{T_1}) -\check{u}^{{\cal J},*};\hat{v}_{T_1}, m_{T_1},{\cal J}\bigr) \nonumber \\ 
&\quad=h^{I,\Qxx}_{T_1}\bigl(\tilde{\theta} ^{I}(\hat{v}_{T_1}, m_{T_1})-\check{u}^{{\cal J},*}_{T_1}\bigr) - \sum_{i\notin {\cal J}}h^{i,\Qxx} \check{u}^{{\cal J} ,*}_{T_1}\nonumber\\
&\qquad+{( (\mu^C)^{*}(\hat v_{T_1},m_{T_1},\check u^{{\cal J} ,*}_{T_1}) - r_D)}  \bigl(\tilde{\theta}^{C}(\hat{v}_{T_1}, m_{T_1})-\check{u}^{{\cal J},*}_{T_1}\bigr) -r_f(Z^{f,{\cal J},*})Z^{f,{\cal J},*}\nonumber\\
&\qquad+\sum_{k \notin \cal{J}}\Big ( h^{i,\Qxx} + r_D  \Big) \check{u}^{\{i\}\cup \cal{J},*}_{T_1}\nonumber\\
& \qquad+ r_D \left(-(N-\abs{{\cal J}} +2) \check{u}^{{\cal J},*}_{T_1} + \tilde{\theta}^{I}( \hat{v}_{T_1}, m_{T_1})+\tilde{\theta}^{C}( \hat{v}_{T_1},m_{T_1}) \right) - r_m(M_t) M_t +r_D\sum_{k \notin \cal{J}}   L^k \nonumber\\
&\quad\ge h^{I,\Qxx}_{T_1}\bigl(\tilde{\theta} ^{I}(\hat{v}_{T_1}, m_{T_1})-\check{u}^{{\cal J},*}_{T_1}\bigr) - \sum_{i\notin {\cal J}}h^{i,\Qxx} \check{u}^{ {\cal J} ,*}_{T_1}\nonumber\\
&\qquad+{( (\mu^C)^{*}(\hat v_{T_1},m_{T_1},\check u^{ {\cal J} ,*}_{T_1}) - r_D)} \bigl(\tilde{\theta}^{C}(\hat{v}_{T_1}, m_{T_1})-\check{u}^{{\cal J},*}_{T_1}\bigr) -r_f(Z^{f,{\cal J}})Z^{f,{\cal J},*}\nonumber\\
&\qquad+\sum_{k \notin \cal{J}}\Big ( h^{i,\Qxx} + r_D  \Big) \check{u}^{\{i\}\cup \cal{J},*}_{T_1}\nonumber\\
& \qquad+ r_D \left(-(N-\abs{{\cal J}} +2) \check{u}^{{\cal J},*}_{T_1} + \tilde{\theta}^{I}( \hat{v}_{T_1}, m_{T_1})+\tilde{\theta}^{C}( \hat{v}_{T_1},m_{T_1}) \right) - r_m(M_t) M_t +r_D\sum_{k\notin \cal{J}}   L^k \nonumber\\
&\quad=h^{I,\Qxx}_{T_1}\bigl(\tilde{\theta}^{I}(\hat{v}_{T_1}, m_{T_1})-\check{u}^{{\cal J},*}_{T_1}\bigr) - \sum_{i\notin {\cal J}}h^{i,\Qxx} \check{u}^{{\cal J} ,*}_{T_1}\nonumber\\
&\qquad+{( (\mu^C)^{*}(\hat v_{T_1},m_{T_1},\check u^{ {\cal J} ,*}_{T_1}) - r_D)}  \bigl(\tilde{\theta}^{C}(\hat{v}_{T_1}, m_{T_1})-\check{u}^{{\cal J},*}_{T_1}\bigr) \nonumber\\
&\qquad-r_f(Z^{f,{\cal J}})\Bigl(-(N-\abs{{\cal J}}+1) \check{u}^{{\cal J},*}_{T_1}+ \tilde{\theta}^{I}( \hat{v}_{T_1}, m_{T_1}) + \tilde{\theta}^{C}( \hat{v}_{T_1},m_{T_1})  +  \sum_{k \notin \cal{J}}   L^k- M_t \Bigr)
\end{align}
\begin{align}
& \qquad+\sum_{k \notin \cal{J}}\Big ( h^{i,\Qxx} + r_D - r_f(Z^{f,{\cal J}}) \Big) \check{u}^{\{i\}\cup \cal{J},*}_{T_1}\nonumber\\
& \qquad+ r_D \left(-(N-\abs{{\cal J}} +2) \check{u}^{{\cal J},*}_{T_1} + \tilde{\theta}^{I}( \hat{v}_{T_1}, m_{T_1})+\tilde{\theta}^{C}( \hat{v}_{T_1},m_{T_1}) \right) - r_m(M_t) M_t +r_D\sum_{k \notin \cal{J}}   L^k \nonumber\\
&\quad\ge h^{I,\Qxx}_{T_1}\bigl(\tilde{\theta} ^{I}(\hat{v}_{T_1}, m_{T_1})-\check{u}^{{\cal J},*}_{T_1}\bigr) - \sum_{i\notin {\cal J}}h^{i,\Qxx} \check{u}^{ {\cal J} ,*}_{T_1}\nonumber\\
&\qquad+h^{C,\Qxx} \bigl(\tilde{\theta} ^{C}(\hat{v}_{T_1}, m_{T_1})-\check{u}^{{\cal J},*}_{T_1}\bigr) \nonumber\\
&\qquad-r_f(Z^{f,{\cal J}})\Bigl(-(N-\abs{{\cal J}}+1) \check{u}^{{\cal J},*}_{T_1}+ \tilde{\theta}^{I}( \hat{v}_{T_1}, m_{T_1}) + \tilde{\theta}^{C}( \hat{v}_{T_1},m_{T_1})  +  \sum_{k \notin \cal{J}}   L^k- M_t \Bigr)\nonumber\\
& \qquad+\sum_{k \notin \cal{J}}\Big ( h^{i,\Qxx} + r_D - r_f(Z^{f,{\cal J}}) \Big) \check{u}^{\{i\}\cup \cal{J}}_{T_1}\nonumber\\
& \qquad+ r_D \left(-(N-\abs{{\cal J}} +2) \check{u}^{{\cal J},*}_{T_1} + \tilde{\theta}^{I}( \hat{v}_{T_1}, m_{T_1})+\tilde{\theta}^{C}( \hat{v}_{T_1},m_{T_1}) \right) - r_m(M_t) M_t +r_D\sum_{k \notin \cal{J}}   L^k \nonumber\\
&\quad= \check{g}(T_1,\check{u}; \hat v_{T_1}, m_{T_1}) dt = \partial_t \check{u}_{T_1}.\nonumber
\end{align}
The first inequality above follows from the following inequality $r_f(Z^{f,{\cal J}})Z^{f,{\cal J,*}} \ge r_f(Z^{f,{\cal J},*})Z^{f,{\cal J,*}}$. To deduce the second inequality above, we have used that $ \rfm, \rfp< \min_{i \in \{1,\ldots N, I\}} \mu^i{\wedge\underline\mu^C}$, and the induction hypothesis for sets of cardinality $n+1$.
This implies that there exists a constant $\epsilon>0$, such that  $\check{u}^{*} (t)\ge \check u(t)$ for $t\in[T_1, T_1+\epsilon].$ This leads to a contradiction, and hence the theorem is proven.

\end{proof}

\begin{proof}[{\bf Proof of Theorem~\ref{thm:muli}}]
	The proof that $\rXVA$ dominates $\XVA^\mu$ for any  $\underline{\mu}^C \leq \mu \leq \overline{\mu}^C$ is done in the same way as in the proof of Theorem \ref{thm:robust}. 
	To prove that the super-replicating strategy is given by equations~\eqref{eq:rob_strat2-mult} and~\eqref{eq:rob_strat3-mult}, fix ${\cal J}\subset \{1, ..., N\}$. Then the value of the portfolio associated with this strategy
	at time $t$ on the set $\{\tau^{{\cal J}}  \wedge \tau^C \wedge \tau^I \wedge T < t <\min_{k \notin \cal{J}} \tau^{k,{\cal J}} \wedge \tau^C \wedge \tau^I \wedge T\}$ is
	\begin{align}
	&\sum_{i\notin {\cal J}}\xi_u^{i,*} \, B_u^i+ \xi_t^{I,*} \, B_t^I + \xi_t^{C,*} \, B_t^C + \xi_t^{f,*} \, B_t^{r_f} -\psi_t^{m,*} \, B_t^{r_m} \nonumber\\
	&= \sum_{i\notin {\cal J}}\xi_u^{i,\cal{J},*} \, B_u^i+ \xi_t^{I,{\cal J},*} \, B_t^I + \xi_t^{C,{\cal J},*} \, B_t^C + \xi_t^{f,{\cal J},*} \, B_t^{r_f} -\psi_t^{m,*} \, B_t^{r_m} =  \check{U}_t^{*} = \check{U}_t^{{\cal J},*}.\nonumber
	\end{align}
	The change in the value of the portfolio on this set is
	\begin{align}
	&\sum_{i\notin {\cal J}}\xi_u^{i,*}\,d B_u^i+ \xi_t^{I,{\cal J},*} \, d B_t^I + \xi_t^{C,{\cal J},*} \, d B_t^C + \xi_t^{f,{\cal J},*} \, d B_t^{r_f} -\psi_t^{m,*} \,d  B_t^{r_m} \label{eq:U*1-mult}\\
	&= \Big( \sum_{i\notin {\cal J}} (r_D+h_i^{\Qxx})\left(\check{U}_t^{{\cal J},*} -\check{U}_t^{\{i\}\cup \cal{J},*}\right)+ \left(L^I (\hat{V}_t - M_{t-})^{+} + \check{U}_t^{{\cal J},*}\right) {\mu^I} \nonumber\\
	&+  ( -L^C (\hat{V}_t - M_{t-})^{-}+ \check{U}_t^{{\cal J},*} ) {\mu^C} + r_m(M_t) M_t  + r_f(\xi_t^{f,{\cal J},*} ) \xi_t^{f,{\cal J},*} B_t^{r_f}    \Big)dt. \nonumber
	\end{align}
	{Similarly to the case of a single name credit default swap, the super-replicating strategy needs to also include the cash flow}
	\begin{align}
	\Big(r_f(\xi_t^{f,{\cal J},*} ) -r_D\Big) \sum_{i\notin {\cal J}}L^idt, \label{eq:U*2-mult}
	\end{align}
	{due to the fact that $\hat V$ is obtained by discounting at the rate $r_D$, rather than $r_f$, and hence the loss given default rates $\sum_{i\notin {\cal J}}L^idt$ accrues interest at rate $r_D$.}
	
	The change in value of the super-replicating portfolio is obtained by using equations~\eqref{eq:U*1-mult} and \eqref{eq:U*2-mult}, and needs to be compared with the change in the valuation process, given by
	
	\begin{align}
	d \check{U}_t^{\cal J, *} & = \Big((r_D+h^{I,\Qxx}\bigl(L^I (\hat{V}_t - M_{t-})^{+} + \check{U}_t^{ {\cal J},*}\bigr) - {(\mu^C)^{*}}(\hat V_t ,M_t,\check U^{ {\cal J},*}_t) \bigl(L^C (\hat{V}_t - M_{t-})^{-}- \check{U}_t^{{\cal J},*}  \bigr)  \label{eq:U*3-tmp}\\
	&  + \sum_{k \notin \cal{J}}   (r_D+h^{k,\Qxx}) \check{U}^{\{i\}\cup \cal{J},*}_t+ r_m(M_t) M_t  + r_f(\xi_t^{f,{\cal J},*} ) \left(\xi_t^{f,{\cal J},*} B_t^{r_f} +\sum_{i\notin {\cal J}}L^i  \right) -r_D\sum_{i\notin {\cal J}}L^i\Big)dt\nonumber.
	\end{align}
	It then follows that \eqref{eq:U*1-mult} together with \eqref{eq:U*2-mult} dominates \eqref{eq:U*3-tmp} from above because
	\begin{align}
	\left ( {(\mu^C)^{*}(\hat V_t,M_t,\check{U}_t^{ {\cal J}, *}) - \mu^C }\right) \bigl(\tilde{\theta} ^{C}(\hat{V}_t, M_t)-\check{U}_t^{{\cal J}, *}\bigr)\ge0.
	\end{align}
	To complete the proof, it is left to consider the set $\{t =\min_{k \notin \cal{J}} \tau^{k,{\cal J}} \wedge \tau^C \wedge \tau^I \wedge T\}$, i.e., when $t$ corresponds to a default time. Assume that the reference entity defaulting at $t$ is $k_0\notin {\cal J}$. Then, by the definition of super-replicating strategy in~\eqref{eq:rob_strat1-mult}, and specifically $\xi_t^{k_0, {\cal J},*}$, it follows that the value of the super-replicating portfolio drops from $\check U_t^{ {\cal J} ,*}$ to $\check U_t^{\{k_0\} \cup \cal{J},*}$. By the induction hypothesis, $\check U_t^{\{k_0\}\cup \cal{J},*} \ge\check U_t^{\{k_0\}\cup \cal{J}}$. Together with Theorem \ref{thm:comp-multi}, it follows that on the set $\{\tau^{{\cal J}} \wedge \tau^C \wedge \tau^I \wedge T < t \le\min_{k \notin \cal{J}} \tau^{k,{\cal J}} \wedge \tau^C \wedge \tau^I \wedge T\}$ the super-replicating portfolio dominates $\check{U}_t^{{\cal J}}$. By summing over the indicator sets as in \eqref{eq:rXVA-tmp}--\eqref{eq:rob_strat2-mult}, we get this dominance for all $0\le t\le \tau^C \wedge \tau^I \wedge T$.

\end{proof}

\end{document}